%% file: main.tex
\documentclass[11pt,letterpaper]{article}
\usepackage[letterpaper,margin=1in]{geometry}

\usepackage{mathpazo}
\usepackage{amsmath,amssymb,amsthm,mathtools}
\usepackage{authblk}
\usepackage{array}
\usepackage{booktabs}
\usepackage{graphicx}
\usepackage{physics}
\usepackage{cite}
\usepackage{enumitem}
\usepackage[dvipsnames]{xcolor}
\usepackage[colorlinks=true,allcolors=blue]{hyperref}
\usepackage[nottoc,notlof,notlot]{tocbibind}

\newtheorem{theorem}{Theorem}[section]
\newtheorem{lemma}[theorem]{Lemma}

\newtheorem{corollary}[theorem]{Corollary}

\theoremstyle{definition}
\newtheorem{definition}[theorem]{Definition}

\newcommand{\poly}{\operatorname{poly}}
\newcommand{\polylog}{\operatorname{polylog}}
\newcommand{\bigtimes}{\mathop{\times}\displaylimits}

\begin{document}

\title{Learning sparse quantum states from single-qubit measurements}

\author[1,2]{Su-un Lee\thanks{suun@uchicago.edu}}
\author[1]{Liang Jiang\thanks{liangjiang@uchicago.edu}}
\author[2]{Kunal Sharma\thanks{kunals@ibm.com}}
\affil[1]{The University of Chicago}
\affil[2]{IBM Research}

\maketitle

\begin{abstract}
  We study the problem of learning a sparse quantum state, an $n$-qubit quantum state whose density matrix has at most $s$ nonzero matrix entries in an unknown product basis. While such states admit compact classical descriptions, they can carry long-range entanglement that prevents reconstruction from local reduced density matrices alone. Therefore, previous learning approaches addressed such long-range-entangled states using many entangling gates to extract the necessary information. In this work, we show that sparse states can nevertheless be efficiently learned using only single-qubit measurements. Specifically, when the sparsity $s$ is constant, our algorithm can learn sparse states from single-qubit measurements with polynomial sample complexity and classical computational complexity. When $s$ grows polynomially with $n$, sparse states can still be learned from single-qubit measurements with polynomial sample complexity, although efficient classical computation is not guaranteed in general. In this regime, however, the classical computational complexity becomes quasipolynomial when the state is sparse in an unknown basis that is a product of a known fixed finite set of single-qubit bases (e.g., eigenbases of Pauli operators). These results establish efficient learning of sparse states with long-range entanglement without entangling gates, and the single-qubit measurement requirements make our algorithms compatible with current quantum devices.
\end{abstract}

\clearpage

\tableofcontents
\clearpage

\input{01_introduction}
\input{02_technical_overview}
\input{03_preliminaries_shared_tools}
\input{04_selected_entry_tomography}
\input{05_low_order_pauli_moments}
\input{06_tree_merging}
\input{07_sparsity_certification}


\bibliographystyle{unsrturl}
\bibliography{references}
\end{document}

%% file: 01_introduction.tex
\section{Introduction}
\label{sec:introduction}

Learning an unknown quantum state from measurement data is a fundamental task in quantum information, quantum complexity, and quantum metrology, with applications to the characterization and benchmarking of quantum devices~\cite{gebhartLearningQuantumSystems2023,anshu2023surveycomplexitylearningquantum}. However, a generic $n$-qubit state requires exponentially many parameters to describe, which makes even storing and processing a full classical description prohibitively costly. Fully reconstructing the state, i.e., quantum state tomography, therefore requires exponentially many measurements in general~\cite{haahSampleOptimalTomographyQuantum2017, odonnell2016efficient}, and these resource requirements make learning generic quantum states for large systems impractical.
To overcome this challenge, there are two complementary approaches. One is predicting only the selected properties of the state without reconstructing the full state~\cite{aaronsonShadowTomographyQuantum2020, Huang_2020PredictingManyProperties}. The other approach, which we consider in this work, retains full reconstruction of the state while restricting attention to classes of quantum states that admit efficient classical representations.

Since all quantum states are represented as density matrices, a natural class of states to consider consists of those represented by sparse matrices. We say an $n$-qubit state $\rho$ is $s$-sparse if its density matrix has at most $s$ nonzero matrix entries in some product basis. Here, a product basis is obtained by choosing an orthonormal basis for each qubit and taking tensor product. When the sparsity $s$ grows at most polynomially with $n$, the state admits an efficient classical description consisting of the local bases and the locations and values of its nonzero matrix entries. Despite this restriction, sparse states can have various multipartite entanglement, including Greenberger--Horne--Zeilinger (GHZ) states~\cite{GHZ}, W states~\cite{durThreeQubitsCan2000}, and Dicke states~\cite{Dickestate}.
Sparse states also play a practical role in sample-based diagonalization methods for quantum chemistry and many-body physics, in which the sparse ground state of a Hamiltonian can ensure that the Hamiltonian is efficiently diagonalizable in a subspace spanned by a small number of samples~\cite{robledomoreno2024chemistryexactsolutionsquantumcentric,yu2025quantumcentricalgorithmsamplebasedkrylov, smith2026quantumcentricsimulationhydrogenabstraction, kirby2026observationimprovedaccuracyclassical, Danilov_2025, Kaliakin2025, Sugisaki_2025, mikkelsen2025quantumselectedconfigurationinteraction, piccinelli2026quantumchemistryprovableconvergence}.

Learning algorithms with single-qubit measurements are particularly attractive for current noisy intermediate-scale quantum (NISQ) devices, as they require only local basis rotations and readouts, avoiding any entangling gates. An important class of states that can be learned from single-qubit measurements is quantum states prepared by constant-depth circuits. Refs.~\cite{huangLearningShallowQuantum2024, landauLearningQuantumStates2025, kim2024learningstatepreparationcircuits} showed these states can be efficiently learned using single-qubit measurements. Due to the constant-depth preparation, the states have a finite correlation length. This results that a collection of local reduced density matrices on regions of bounded radius contains sufficient information to reconstruct the full state, where those reduced density matrices can be efficiently estimated from random single-qubit Pauli measurements.

However, this approach does not apply to sparse states in general. This is because sparse states can carry long-range entanglement. For example, the $n$-qubit GHZ state with unknown phase $\theta$,
\begin{equation}
  \lvert\mathrm{GHZ}_\theta \rangle=\frac{\lvert0^n\rangle+e^{i\theta}\lvert1^n\rangle}{\sqrt2},
  \label{eq:ghz-theta-state}  
\end{equation}
is sparse in the computational basis, but tracing out even one qubit removes the information about the phase, so every reduced density matrix is independent of $\theta$. Consequently, no constant-depth circuit can generate such sparse quantum states, and any collection of local reduced density matrices cannot determine the full state. This prevents the previous single-qubit measurement learning approach from applying to generic sparse states.

Matrix product states (MPSs) form another important class of states, which can accommodate such long-range entanglement. Described by one-dimensional tensor networks, MPSs include all pure sparse quantum states, including GHZ states. MPS learning has been extensively studied~\cite{landoncardinal2010efficientdirecttomographymatrix, cramerEfficientQuantumState2010, bakshiLearningClosestProduct2025, linEfficientMatrixProduct2026}. To address long-range entanglement, however, the protocols in these works have to progressively disentangle the qubits. This requires many entangling gates during learning, with linear-depth circuits in one dimension and logarithmic depth with all-to-all connectivity.

In summary, the existing single-qubit learning protocols do not extend to general sparse states, which can exhibit long-range entanglement, while the methods discussed above handle such entanglement using many entangling gates. To our knowledge, even restricted to a simple case where the states are sparse in the computational basis, existing methods for efficiently learning sparse states require entangling gates~\cite{sen2026learningsparsequantumstates, Cai2016OptimalLargeScaleQuantumStateTomography, wong2025efficientquantumtomographypolynomial, benderskySelectiveEfficientQuantum2013, morrisSelectiveQuantumState2019}. This leads to our central question:
\begin{quote}
  \begin{center}
    \emph{Can sparse states be efficiently learned from single-qubit measurements?}
  \end{center}
\end{quote}
Here, efficiency refers separately to sample complexity and classical computational complexity, which are the number of independently prepared copies of $\rho$ and the classical runtime needed to process the measurement data and reconstruct the learned state, respectively.

\subsection{Results}
\label{sec:interim-results}

\begin{table}[t]
  \centering
  \footnotesize
  \setlength{\tabcolsep}{2pt}
  \renewcommand{\arraystretch}{1.25}
  \begin{tabular}{@{}
      >{\raggedright\arraybackslash}p{0.35\textwidth}
      >{\raggedright\arraybackslash}p{0.30\textwidth}
      >{\raggedright\arraybackslash}p{0.30\textwidth}@{}}
    \toprule
    \textbf{Setting} & \textbf{Sample complexity} & \textbf{Computational complexity} \\
    \midrule
    Known product basis & $O(s^4 \log(s)/\varepsilon^2)$ & $O(ns^4 \log(s)/\varepsilon^2)$ \\
    \addlinespace
    Algorithm~1; arbitrary product basis & $2^{O(s(\log s)^2)}(n/\varepsilon)^{O(s)}$ & $2^{O(s(\log s)^2)}(n/\varepsilon)^{O(s)}$ \\
    \addlinespace
    Algorithm~2; arbitrary product basis & $O(n^6s^{12}\varepsilon^{-4}\log^2(ns))$ & $(ns/\varepsilon)^{O(n)}$ \\
    \addlinespace
    Algorithm~2; product of finite single-qubit basis set & $O(n^6s^{12}\varepsilon^{-4}\log^2(ns))$ & $n^{O(\log s)}/\varepsilon^4$ \\
    \bottomrule
  \end{tabular}
  \caption{Summary of the results. Sample complexity is the number of copies of $\rho$ needed to learn the state, and computational complexity is the classical runtime required.}
  \label{tab:complexity-regimes}
\end{table}

In this work, we answer the above question affirmatively. We show that, when the sparsity $s$ is constant, sparse states can be learned from single-qubit measurements with both polynomial sample complexity and polynomial classical computational complexity. When the sparsity $s$ grows polynomially with $n$, sparse states can still be learned from single-qubit measurements with polynomial sample complexity, while the efficiency of the classical computational complexity is not guaranteed in general. However, we show that the classical computational complexity becomes quasipolynomial when the state is sparse in an unknown basis that is a product of a known fixed finite set of single-qubit bases (e.g., eigenbases of Pauli operators).

Let $\rho$ be an unknown $n$-qubit state whose density matrix has at most $s$ nonzero matrix entries in some product basis. Specifically, a product basis $B = (B_1, B_2, \ldots, B_n)$ consists of an orthonormal basis $B_i$ for each qubit $i$. We denote the matrix representation of $\rho$ in this basis by $[\rho]_B$. We say $\rho$ is $s$-sparse if there is a product basis $B$ such that $[\rho]_B$ is an $s$-sparse matrix. The goal of our learning algorithms is to output an operator $\widehat\rho$ with at most $s$ nonzero matrix entries in some product basis $B'$, such that $\widehat\rho$ is $\varepsilon$-close to $\rho$ in trace distance. Specifically, the output consists of a product basis $B' = (B'_1, B'_2, \ldots, B'_n)$ and the $s$-sparse matrix $[\widehat\rho]_{B'}$.

We first show that if a product basis $B$ in which $\rho$ is $s$-sparse is known, single-qubit measurements can efficiently learn $\rho$ even when $s$ grows polynomially with $n$. To this end, we develop \emph{selected-entry tomography}, which estimates a single entry of $[\rho]_B$. Previous works~\cite{benderskySelectiveEfficientQuantum2013, morrisSelectiveQuantumState2019, wong2025efficientquantumtomographypolynomial, patelSelectiveEfficientQuantum2026} have tackled this problem of estimating entries of density matrices, but often require entangling gates or a restricted entry set. In contrast, selected-entry tomography in our work uses only single-qubit measurements to estimate any entry of $[\rho]_B$. We show that this can be used to learn sparse states using single-qubit measurements when $B$ is known.

\begin{theorem}[Learning with known basis, Theorem~\ref{thm:known-basis-sparse-learning}]
  Let $\rho$ be an $n$-qubit $s$-sparse state, and suppose the product basis $B$ in which $\rho$ has at most $s$ nonzero matrix entries is known. Then, one can approximately learn $\rho$ in trace distance only using single-qubit measurements with high probabiltiy, from $O(s^4\log (s)/\varepsilon^2)$ copies of $\rho$ and $O(ns^4\log (s)/\varepsilon^2)$ classical runtime.
  \label{thm:known-basis-learning-summary}
\end{theorem}

This result reduces the task of learning $\rho$ to finding a basis in which $\rho$ is sparse. With that, we develop two algorithms for learning sparse states. The first algorithm, Algorithm~1, uses low-order Pauli moments, the expectation values of Pauli operators acting on only a few qubits. We could show that, when $s$ is constant, low-order Pauli moments can efficiently generate a list of candidate bases which includes at least one basis that makes $\rho$ $s$-sparse. Combined with the selected-entry tomography, this allows Algorithm~1 to learn $\rho$ with polynomial sample complexity and classical computational complexity, given $s$ is constant.

\begin{theorem}[Algorithm~1, Theorem~\ref{thm:pauli-moment-learning}]
  Let $\rho$ be an $n$-qubit $s$-sparse state with $s=O(1)$. One can approximately learn $\rho$ in trace distance only using single-qubit measurements with high probability, from ${\rm poly}(n, 1/\varepsilon)$ copies of $\rho$ and ${\rm poly}(n, 1/\varepsilon)$ classical runtime.
  \label{thm:constant-sparsity-learning-summary}
\end{theorem}

Algorithm~1 has polynomial sample complexity and classical runtime for constant $s$, but the polynomial degree depends on $s$. Therefore, the complexities of Algorithm~1 rapidly grows with $s$, and it is not efficient when $s={\rm poly}(n)$. To overcome this limitation, we develop another algorithm, Algorithm~2, which achievs polynomial sample complexity in both $n$ and $s$. In particular, Algorithm~2 divides $n$ qubits into small blocks of qubits, and learns local states. Then, it merges small blocks to learn larger blocks using selected-entry tomography, and eventually outputs the full state after merging all blocks to a single block of $n$ qubits.

While it consumes polynomially many copies of $\rho$, the classical runtime grows rapidly in $n$ and $s$. However, the classical runtime can be greatly reduced when the local basis for each qubit belongs to a known fixed finite set. Specifically, let $\mathcal A$ be a finite set of single-qubit bases, and suppose that $\rho$ is $s$-sparse in some basis $B=(B_1,\ldots,B_n)\in\mathcal A^n$. Then, we show that the classical runtime becomes quasipolynomial in $n$ when $s$ grows polynomially in $n$. Notably, the Donoho--Stark uncertainty principle~\cite{donohoUncertaintyPrinciplesSignal1989,eladGeneralizedUncertaintyPrinciple2002,boggiattoTwoAspectsDonoho2016}, which was originally developed for classical signal processing, plays a central role in reducing the classical runtime of Algorithm~2.

\begin{theorem}[Algorithm~2, Theorem~\ref{thm:tree-continuous-basis} and Theorem~\ref{thm:tree-finite-bases}]
  Let $\rho$ be an $n$-qubit $s$-sparse state. One can approximately learn $\rho$ in trace distance only using single-qubit measurements with high probabiltiy, from ${\rm poly}(n, s, 1/\varepsilon)$ copies of $\rho$ and $(ns/\varepsilon)^{O(n)}$ classical runtime. When $\rho$ is $s$-sparse in a product basis $B=(B_1,\ldots,B_n)\in\mathcal A^n$ for some finite set of single-qubit bases $\mathcal A$, the classical runtime becomes $n^{O(\log s)}/\varepsilon^4$.
\end{theorem}

In summary, we show that sparse states can be learned from single-qubit measurements with polynomial sample complexity, and polynomial classical computational complexity when $s$ is constant. When $s$ grows polynomially with $n$, the classical computational complexity is not guaranteed in general sparse states, but it becomes quasipolynomial when the state is sparse in a basis that is a product of a finite set of single-qubit bases. Table~\ref{tab:complexity-regimes} summarizes our results.

\subsection{Discussion}


\begin{figure}
  \centering
  \includegraphics[width=0.7\linewidth]{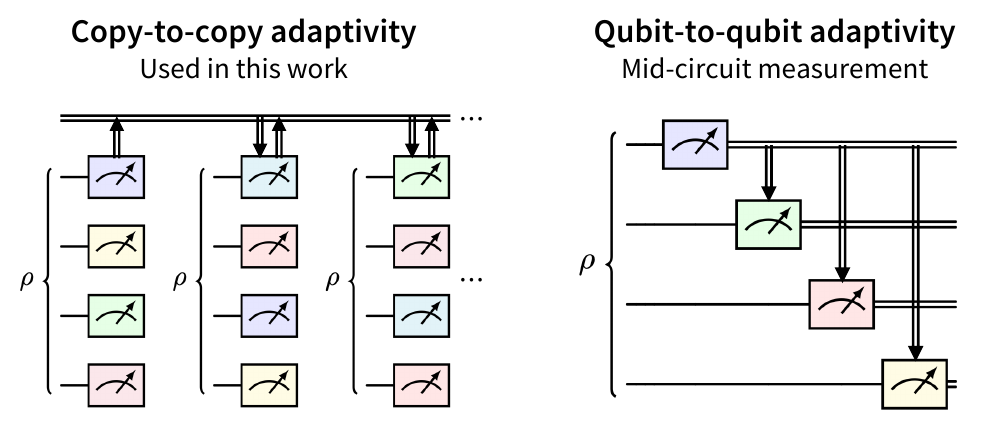}
  \caption{Two kinds of adaptivity in single-qubit measurements. Our protocol uses copy-to-copy adaptivity, where measurement bases on qubits for each copy of $\rho$ may depend on outcomes from earlier copies, but not on outcomes within the same copy. On the other hand, qubit-to-qubit adaptivity permits the basis of each qubit to depend on measurement outcomes on other qubits within the same copy. This essentially requires mid-circuit measurement and feed-forward.}
  \label{fig:adaptivity}
\end{figure}

We first note that our algorithms can be used to test whether an arbitrary input state is $s$-sparse. So far, we assumed that the input state $\rho$ is guaranteed to be an $s$-sparse state, but this promise may not be known to hold. Specifically, if $\rho$ is an arbitrary quantum state, which may or may not be $s$-sparse, one may want to determine whether $\rho$ is $s$-sparse. We show in Section~\ref{sec:sparsity-certification} that a simple variant of Algorithm~1 can be used to achieve this: If $\rho$ is indeed approximately $s$-sparse, Algorithm~1 can be used to output $\mathsf{YES}$ along with an $s$-sparse operator $\widehat \rho$ that is $\varepsilon$-close to $\rho$ in trace distance. On the other hand, if the input state is $\varepsilon$-far from any $s$-sparse state in trace distance, it outputs $\mathsf{NO}$. Therefore, for a state whose structure and classical description are a priori unknown, this enables a search for a compact description whose accuracy can be certified upon acceptance.

We emphasize that \emph{adaptivity} plays a crucial role in our learning algorithms. There are two notions of adaptivity in single-qubit measurements [Fig.~\ref{fig:adaptivity}]. The first is \emph{copy-to-copy adaptivity}, where the measurement bases on qubits for each copy of $\rho$ may depend on outcomes from earlier copies. This is the adaptivity used in our algorithms, and it does not require any additional resources within measurement circuits beyond ordinary single-qubit measurements. The other is \emph{qubit-to-qubit adaptivity}, where the measurement basis of each qubit can depend on the measurement outcomes of other qubits within the same copy of $\rho$. While it has been shown that this is strictly more powerful than the copy-to-copy adaptivity (e.g., Ref.~\cite{gupta2026fewsinglequbitmeasurements}), this essentially requires mid-circuit measurement and feed-forward, which is relatively noisy and costly in current quantum devices~\cite{riste2012feedbackcontrolsolidstatequbit, riste2012initializationmeasurement, corcoles2021exploitingdynamicquantum, gupta2024probabilisticerrorcancellation, ivashkov2024highfidelitymultiqubitgeneralized, lund2026constantdepthquantumimaginarytime, chu2026learningmidcircuitmeasurementbackaction}. Throughout this work, we do not use qubit-to-qubit adaptivity.

To illustrate the role of copy-to-copy adaptivity in our algorithms, consider the GHZ state in Eq.~\eqref{eq:ghz-theta-state}, but now expressed in an unknown product of single-qubit Pauli bases. We may take this product basis to be the computational basis without loss of generality, although it remains unknown to the algorithm. A product Pauli measurement then can reveal the phase $\theta$ only if every qubit is measured in the $X$ or $Y$ basis; measuring even one qubit in the $Z$ basis makes the joint output distribution independent of $\theta$. However, without any prior information about the product basis, the best we can do is randomly guess the correct Pauli measurement basis for each qubit and the probability of choosing such a basis is $(2/3)^n$, which is exponentially unlikely. On the other hand, our algorithms use outcomes from earlier copies to guide the choice of measurement bases on later copies. This adaptive basis choice ensures that our algorithms eventually find and measure in the correct basis.

Lastly, our algorithms estimate expectation values of observables and probability masses to reconstruct the output state, and therefore are compatible with standard error mitigation methods~\cite{temme2017errormitigation, bravyi2021mitigatingmeasurementerrors, vandenbergProbabilisticErrorCancellation2023, Maciejewski2020mitigationofreadout, nation2021scalablemitigationmeasurementerrors, cai2023quantumerrormitigation}. Especially, as our algorithms only use expectation values of product observables obtained from single-qubit measurements, probabilistic error cancellation is achieved with additional sampling overhead $\exp(O(ne))$, where $e$ denotes the single-qubit readout error rate~\cite{temme2017errormitigation, bravyi2021mitigatingmeasurementerrors, vandenbergProbabilisticErrorCancellation2023}. Meanwhile, current state-of-the-art quantum processors already achive error rates of $e \lesssim 1/n$~\cite{morvanPhaseTransitionsRandom2024, google2025quantumerrorcorrection, ransford2026qubittrappedionquantumcomputer, zhu2022quantumcomputationaladvantage}, rendering error mitigation overhead manageable. These observations suggest that our learning algorithms can be implemented in current quantum devices, robust to readout errors with only a modest additional sampling cost.

The rest of the paper is organized as follows. Section~\ref{sec:technical-overview} explains the main ideas behind both Algorithm~1 and 2. Section~\ref{sec:shared-tools} introduces the preliminary tools used throughout our analyses, and Section~\ref{sec:selected-entry-tomography} develops selected-entry tomography and the learning method for the states sparse in a known product basis. Section~\ref{sec:continuous-fixed-sparsity} and Section~\ref{sec:algorithm1-results} present the results and analyses of Algorithm~1 and Algorithm~2, respectively. Finally, Section~\ref{sec:sparsity-certification} extends Algorithm~1 to testing sparsity of arbitrary input states.

%% file: 02_technical_overview.tex
\section{Technical overview}
\label{sec:technical-overview}

Let $\rho$ be an $n$-qubit state and let $B=(B_1,\ldots,B_n)$ be a product basis in which $\rho$ is $s$-sparse. Here, $B_i=\{\lvert\phi_0^{(i)}\rangle,\lvert\phi_1^{(i)}\rangle\}$ is an orthonormal basis for qubit $i$. For a bitstring $x\in\{0,1\}^n$, we write $\lvert\phi_x\rangle=\bigotimes_{i=1}^n\lvert\phi_{x_i}^{(i)}\rangle$. We denote the matrix of $\rho$ in this basis by $[\rho]_B$, with entries $[\rho]_B(x,y):=\langle\phi_x\rvert\rho\lvert\phi_y\rangle$. Then the sparsity condition is $\bigl\lvert\operatorname{supp}([\rho]_B)\bigr\rvert\le s$. Since $\rho$ is positive semidefinite, we have
\begin{equation}
  \bigl\lvert[\rho]_B(x,y)\bigr\rvert^2
  \le [\rho]_B(x,x)[\rho]_B(y,y).
  \label{eq:overview-positivity}
\end{equation}
Therefore, if an off-diagonal matrix entry $[\rho]_B(x,y)$ is nonzero, then both diagonal entries $[\rho]_B(x,x)$ and $[\rho]_B(y,y)$ must also be nonzero. Consequently, with
\begin{equation}
  S = \{x\in\{0,1\}^n:[\rho]_B(x,x)>0\},
\end{equation}
the support of $[\rho]_B$ is confined to an $S\times S$ submatrix, where $S$ is a set of at most $s$ bitstrings.

If the product basis $B$ were known, the learning procedure would proceed as follows. Direct measurements in $B$ produce samples from the distribution $p_B(x)=[\rho]_B(x,x)$. With sufficiently many samples, we can find all bitstrings $x$ in $S$ except those whose probabilities are negligible. Once $S$ is found, we can estimate the submatrix of $[\rho]_B$ on $S\times S$ only using single-qubit measurements. For this purpose, we introduce selected-entry tomography, a simple method for estimating selected matrix entries of $\rho$ using randomized single-qubit measurements, building upon the previous works on selective quantum tomography~\cite{benderskySelectiveEfficientQuantum2013, morrisSelectiveQuantumState2019, patelSelectiveEfficientQuantum2026}.

To explain selected-entry tomography, let $B$ be the computational basis without loss of generality. Then, we have $\rho=\sum_{u,w\in\{0,1\}^n}[\rho]_B(u,w)\lvert u\rangle\langle w\rvert$, and the goal is to estimate $[\rho]_B(x,y)$ for a fixed pair of bitstrings $(x,y)$. To this end, we choose $z\in\{0,1\}^n$ uniformly at random and note that $\operatorname{Tr}[X^{x\oplus y}Z^z\rho]=\sum_{u}(-1)^{z\cdot u}[\rho]_B(u,u\oplus x\oplus y)$. Here, for a bitstring $a\in\{0,1\}^n$, we write $X^a=\bigotimes_{i=1}^nX^{a_i}$ and $Z^a=\bigotimes_{i=1}^nZ^{a_i}$, and $\oplus$ denotes the bitwise XOR. Since $\mathbb E_z[(-1)^{z\cdot(x\oplus u)}]=\mathbf 1[x=u]$, we have
\begin{equation}
  \mathbb E_z[(-1)^{z\cdot x}\operatorname{Tr}(X^{x\oplus y}Z^{z}\rho)] = [\rho]_B(x,y).
  \label{eq:overview-entry-estimator}
\end{equation}
Therefore, we can estimate $[\rho]_B(x,y)$ with an additive error by measuring $X^{x\oplus y}Z^z$ with randomly chosen $z$. This procedure only uses single-qubit measurements because $X^{x\oplus y}Z^z$ is proportional to a tensor product of single-qubit Pauli operators. In this way, we can estimate the submatrix of $[\rho]_B$ on $S\times S$ by repeating selected-entry tomography $|S|^2$ times. This concludes the learning procedure for the case where the product basis $B$ is known.

The learning problem is therefore reduced to finding a product basis in which the state has a sparse matrix representation, or a sufficiently accurate candidate for such a basis. Note that such product bases need not be unique, and any candidate product basis that permits the required sparse representation is sufficient. Next, we describe how our algorithms, Algorithm~1 and Algorithm~2, construct such a basis.

\subsection{Algorithm 1: low-order Pauli moments}
\label{sec:algorithm2-overview}

\begin{figure}
  \centering
  \includegraphics[width=0.9\linewidth]{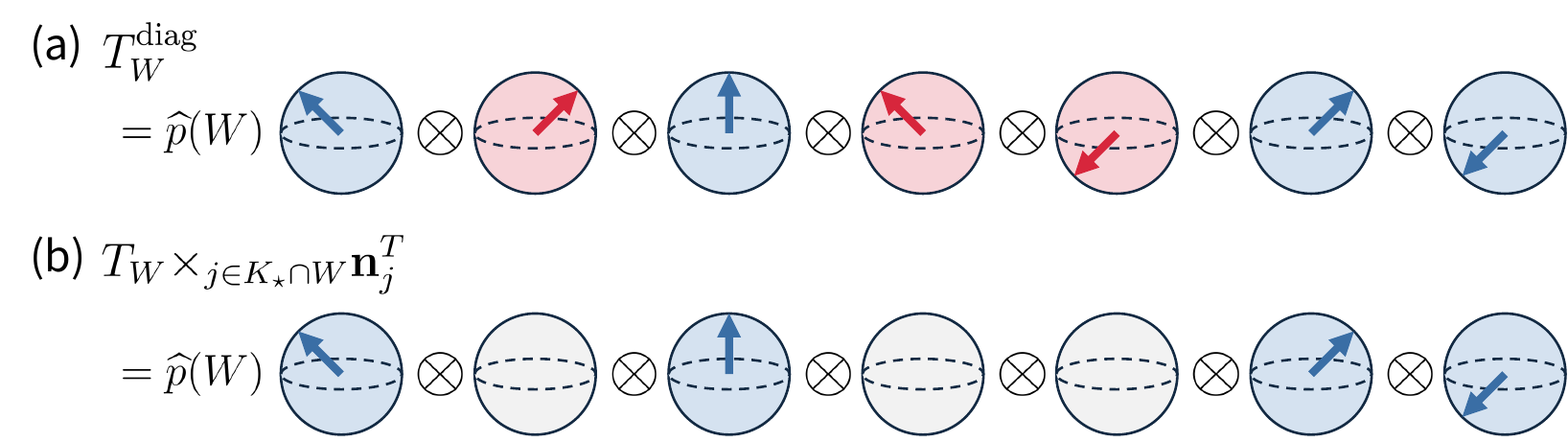}
  \caption{Overview of Algorithm 1. The red qubits are the ones in the difficult set $K_\star$, and the blue ones are the remaining qubits. (a) Without the contribution of the off-diagonal entries in $[\rho]_B$, the Pauli moment tensor is a tensor product of single-qubit basis vectors. (b) By contracting the Pauli moment tensor with the basis vectors of difficult qubits, off-diagonal contributions are canceled out and we can extract the basis vectors of the remaining qubits.
    \label{fig:low-order-correlations}}
\end{figure}

Algorithm~1 learns states that are sparse in an arbitrary unknown product basis, with polynomial sample complexity and classical runtime for a fixed sparsity $s$. Its central idea is to generate candidate local bases by measuring the low-order Pauli moments, which are expectation values of products of single-qubit Pauli operators. The algorithm then uses the selected-entry tomography to test each candidate and reconstruct the state in an accepted basis.

To explain how low-order Pauli expectation values supply candidate local bases, let $B=(B_1,\dots,B_n)$ with $B_i=\{\lvert\phi_0^{(i)}\rangle,\lvert\phi_1^{(i)}\rangle\}$ be the unknown product basis in which the input state $\rho$ is $s$-sparse. Let $p(x)=[\rho]_B(x,x)$ denote the diagonal entries of $\rho$ in basis $B$. Since $[\rho]_B$ has at most $s$ nonzero matrix entries, $p$ is supported on at most $s$ bitstrings. For $W\subseteq[n]$, we write the (unnormalized) Fourier coefficient of $p$ as
\begin{equation}
  \widehat{p}(W):=\sum_{x\in\{0,1\}^n} p(x)\prod_{i\in W}(-1)^{x_i}.
\end{equation}
We denote the Bloch vector of the basis state $\lvert\phi_0^{(i)}\rangle$ by $\mathbf n_i\in\mathbb R^3$ so that $\langle\phi_b^{(i)}\rvert\boldsymbol\sigma\lvert\phi_b^{(i)}\rangle=(-1)^b\mathbf n_i$, where $\boldsymbol\sigma=(X,Y,Z)$ is the vector of single-qubit Pauli matrices. For each nonempty $W\subseteq[n]$ and $\alpha=(\alpha_i)_{i\in W}\in\{X,Y,Z\}^W$, we define the \emph{Pauli moment tensor} $T_W$ by
\begin{equation}
  [T_W]_{\alpha}
  :=\operatorname{Tr}\!\left[
    \left(\bigotimes_{i\in W}\alpha_i\right)\rho
    \right],
  \label{eq:overview-product-pauli-tensor}
\end{equation}
Each entry of $T_W$ is the expectation value of a Pauli observable supported on $W$ and can therefore be estimated using single-qubit measurements.

If $\rho$ were diagonal in $B$, i.e., $\rho=\sum_x p(x)\lvert\phi_x\rangle\langle\phi_x\rvert$, then $T_W$ would be useful for learning the Bloch vectors $\mathbf n_i$ for $i\in W$. As $\langle\phi_b^{(i)}\rvert\boldsymbol\sigma\lvert\phi_b^{(i)}\rangle=(-1)^b\mathbf n_i$, the Pauli moment tensor of the diagonal state is given by
\begin{equation}
  T_W^{\mathrm{diag}}
  =\widehat p(W)\bigotimes_{i\in W}\mathbf n_i.
  \label{eq:overview-diagonal-factorization}
\end{equation}
Thus, whenever $\widehat p(W)\ne0$, $T_W^{\mathrm{diag}}$ is a rank-one tensor; see Figure~\ref{fig:low-order-correlations}(a). To extract the basis on a qubit $i\in W$, we view $T_W^{\mathrm{diag}}$ as a matrix with row index $\alpha_i$ and column index $(\alpha_j)_{j\in W\setminus\{i\}}$ and perform singular value decomposition (SVD). Its leading left singular vector is then $\mathbf n_i$, up to sign. Hence we can learn the basis on qubit $i$ provided that (i) there exists a set $W\subseteq[n]$ containing $i$ that is small enough for $T_W$ to be estimated, (ii) $\lvert\widehat p(W)\rvert$ is sufficiently large so that we can robustly extract $\mathbf n_i$ via SVD, and (iii) the contributions of off-diagonal entries of $[\rho]_B$ to $T_W$ can be removed.

We first show that the conditions (i) and (ii) hold for all qubits except a small number of them. Specifically, our Fourier analysis of sparse Boolean functions in Section~\ref{sec:sparse-distribution-fourier-analysis} shows that there exists a subset of qubits $L\subseteq[n]$ with $\lvert L\rvert = O(\log s)$ such that for any $i\notin L$, there exists $W\subseteq[n]$ with $i\in W$ satisfying
\begin{equation}
  \lvert W\rvert = O(\log s),
  \qquad
  \lvert\widehat p(W)\rvert>\frac{1}{4s^{\log_2s}}.
  \label{eq:overview-low-order-signal}
\end{equation}
Thus, except for at most $O(\log s)$ qubits, every qubit appears in a small-size Pauli moment tensor with a non-negligible Fourier coefficient.

What remains is to remove the contributions of off-diagonal entries of $[\rho]_B$ to $T_W$. To this end, we define the flip set by $F(x,y):=\{i\in[n]:x_i\ne y_i\}$ for a nonzero matrix entry $[\rho]_B(x,y)$ with $x\ne y$. This entry can contribute to $T_W$ only if $F(x,y)\subseteq W$; otherwise, the contribution of $[\rho]_B(x,y)$ in Eq.~\eqref{eq:overview-product-pauli-tensor} is traced out. Therefore, for $T_W$ with $\lvert W\rvert=O(\log s)$, it suffices to remove the off-diagonal entries $[\rho]_B(x,y)$ whose nonempty flip sets $F(x,y)$ are contained in $W$.

Let $\mathcal F:=\{F(x,y):[\rho]_B(x,y)\ne0,\ x\ne y\}$ be the set of distinct nonempty flip sets arising from nonzero off-diagonal entries. For each $F\in\mathcal F$, choose one representative qubit $j_F\in F$ and define $J_\star:=\{j_F:F\in\mathcal F\}$. Since each distinct flip set arises from one of the at most $s$ nonzero matrix entries, $\lvert\mathcal F\rvert\le s$ and hence $\lvert J_\star\rvert\le s$. By construction, $J_\star$ intersects the flip set of every nonzero off-diagonal entry.

Then, the set $J_\star$ removes the off-diagonal contributions as follows. Given a nonzero off-diagonal entry $[\rho]_B(x,y)$ with $x\ne y$ contributing to $T_W$, we have $F(x,y)\subseteq W$, so there exists $j_\star \in J_\star\cap W$. Then, we contract $T_W$ with $\mathbf n_{j_\star}$:
\begin{equation}
  \left[T_W \times \mathbf n_{j_\star}^T\right]_{\alpha_{W\setminus\{j_\star\}}}
  := \sum_{\alpha_{j_\star}} [T_W]_{\alpha} [\mathbf n_{j_\star}]_{\alpha_{j_\star}}.
  \label{eq:overview-single-contraction}
\end{equation}
Since $\mathbf n_j\cdot\boldsymbol\sigma = \lvert\phi_{0}^{(j)}\rangle\!\langle\phi_{0}^{(j)}\rvert - \lvert\phi_{1}^{(j)}\rangle\!\langle\phi_{1}^{(j)}\rvert$ for all $j$, we have $\langle\phi_{y}\rvert \mathbf n_{j_\star}\cdot\boldsymbol\sigma \lvert\phi_{x}\rangle=0$
as $x_{j_\star} \ne y_{j_\star}$. Contraction in Eq.~\eqref{eq:overview-single-contraction} therefore annihilates the contribution of the off-diagonal term $[\rho]_B(x,y)\lvert\phi_{x}\rangle\!\langle\phi_{y}\rvert$ in $T_W$ of Eq.~\eqref{eq:overview-product-pauli-tensor}. Thus, contracting $T_W$ against $\mathbf n_j$ for every $j\in J_\star\cap W$ removes all off-diagonal contributions.

Therefore, all three conditions for learning local bases are satisfied for all qubits except those in $L$ and $J_\star$. Combining $L$ and $J_\star$, we define the set $K_\star$ of \emph{difficult qubits}:
\begin{equation}
  K_\star:=L\cup J_\star,
  \qquad
  \lvert K_\star\rvert\le s+O(\log_2s).
  \label{eq:overview-automatic-enumeration-set}
\end{equation}
Then, for any $i\notin K_\star$, there exists $W$ such that $i\in W$ and
\begin{equation}
  T_W\bigtimes_{j\in K_\star\cap W}\mathbf n_j^T
  =\widehat p(W)
  \bigotimes_{k\in W\setminus K_\star}\mathbf n_k,
  \label{eq:overview-exact-filtering}
\end{equation}
with $\lvert W\rvert=O(\log s)$ and $\lvert\widehat p(W)\rvert>1/(4s^{\log_2s})$; see Figure~\ref{fig:low-order-correlations}. Consequently, given the local basis vectors $\mathbf n_j$ for all $j\in K_\star$, one can learn the bases of all remaining qubits $i\notin K_\star$ by performing an SVD of the tensor in Eq.~\eqref{eq:overview-exact-filtering}.

We do not a priori know which qubits form the difficult set $K_\star$ or their local bases. We therefore guess them by brute force, trying every qubit set $K$ of size at most $s+O(\log_2s)$ together with epsilon-net candidates for the local bases on those qubits $(\mathbf n_j)_{j\in K}$. For each guess of $K$ and associated local bases, we contract the low-order Pauli moment tensors using those local bases to find the basis $\mathbf n_i$ for all $i\notin K$ via SVD as described above. In this way, we generate a list of candidate bases $B'$ in which $\rho$ may be $s$-sparse.

While most of these candidate bases are incorrect, at least one candidate will have $K=K_\star$ and sufficiently accurate local bases of qubits in $K_\star$; thus, there is at least one correct basis in the candidate list. Algorithm~1 then tests if each basis $B'$ makes $[\rho]_{B'}$ $s$-sparse: By performing the selected-entry tomography with that basis, it estimates $[\rho]_{B'}$ and counts the number of significant entries. If the number of significant entries is at most $s$, then the basis $B'$ is accepted; otherwise, it is rejected and the algorithm moves on to the next candidate. Finally, after finding a correct basis, Algorithm~1 outputs an $s$-sparse operator $\widehat \rho$ that is close to the input state $\rho$.

Here, the enumeration for finding the difficult qubit set contributes $n^{O(s)}$ choices, and the epsilon-net and low-order Pauli-moment tensors also have polynomial size for every fixed $s$. Therefore, both the sample and computational complexities are $n^{O(s)}$, which is polynomial for fixed $s$.

\subsection{Algorithm 2: tree-structured merging}
\label{sec:algorithm1-overview}

\begin{figure}
  \centering
  \includegraphics[width=0.9\linewidth]{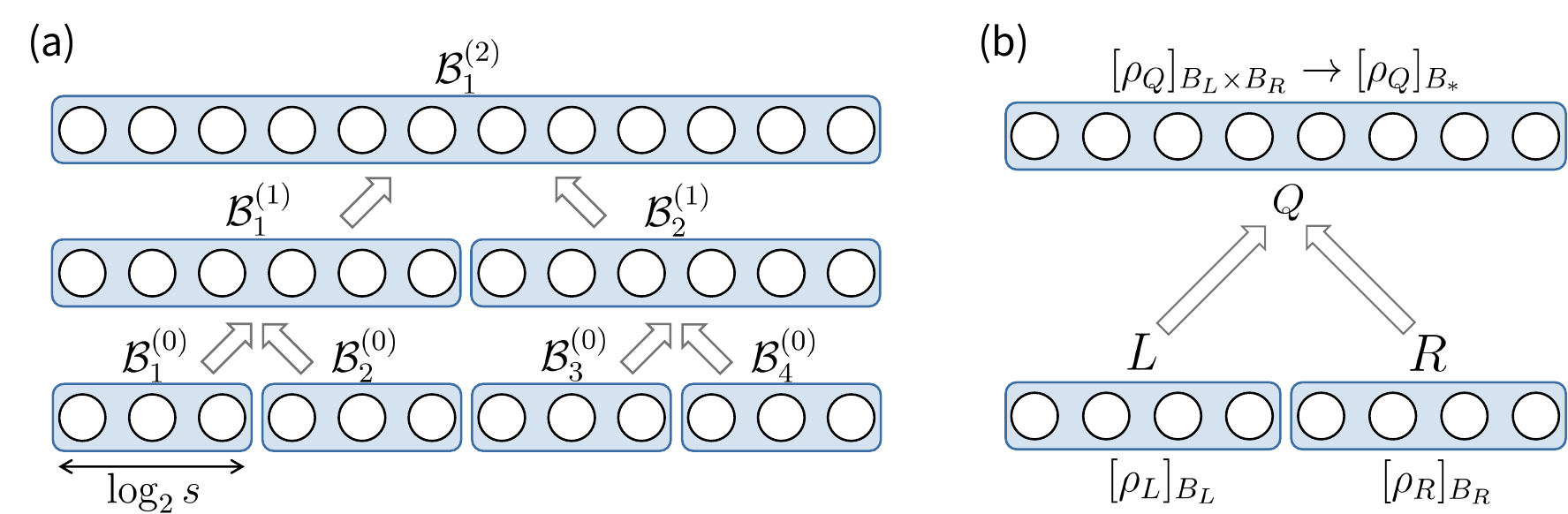}
  \caption{Overview of Algorithm~2. (a) Tree structure of Algorithm~2: starting from the leaves $\mathcal B_j^{(0)}$, each parent block is formed by merging two adjacent child blocks. (b) Each child block $L$ and $R$ retains a subspace spanned by at most $s$ basis vectors in $B_L$ and $B_R$, respectively. Their tensor product spans at most $s^2$ basis vectors in $B_L \times B_R$, so selected-entry tomography estimates at most $s^4$ matrix entries. The algorithm then finds a basis $B_*$ and retains at most $s$ basis vectors.\label{fig:tree-merging}}
\end{figure}

Algorithm~1 has polynomial sample complexity and classical runtime for fixed $s$, but it becomes quickly impractical when $s$ grows. Complementing this shortcoming, we present Algorithm~2 which achieves polynomial sample complexity in both $n$ and $s$. Its classical runtime depends on the basis $B$ in which the input state $\rho$ is $s$-sparse: it is super-exponential if $B$ can be an arbitrary product basis, and becomes quasipolynomial if $B$ is a product of single-qubit bases drawn from a fixed finite set, when $s$ grows polynomially with $n$.

Algorithm~2 first divides the $n$ qubits into $K=O(n/\log s)$ disjoint blocks $\mathcal B_j^{(0)}$, for $j=1,\ldots,K$, each containing at most $\lfloor\log_2s\rfloor$ qubits. Here, we take $s\ge2$ without loss of generality. It then recursively merges adjacent pairs: for $d=0,1,\ldots,D-1$ with $D=O(\log n)$, blocks $\mathcal B_{2j-1}^{(d)}$ and $\mathcal B_{2j}^{(d)}$ are combined into $\mathcal B_j^{(d+1)}$, i.e., $\mathcal B_j^{(d+1)}=\mathcal B_{2j-1}^{(d)}\sqcup\mathcal B_{2j}^{(d)}$. At level $D$, the single block $\mathcal B_1^{(D)}$ contains all $n$ qubits; see Figure~\ref{fig:tree-merging}(a). As will be shown, Algorithm~2 recursively learns the reduced density matrices of $\rho$ on the blocks on each level by selected-entry tomography using the previous levels' outcomes, proceeding toward level $D$. The central observation is that sparsity is inherited by reduced density matrices. If $\rho$ has at most $s$ nonzero matrix entries in a product basis $B$, then every reduced density matrix on $Q \subseteq[n]$, $\rho_Q=\operatorname{Tr}_{Q^c}(\rho)$ has at most $s$ nonzero entries in the restriction of $B$ to $Q$. Indeed, the partial trace discards some entries and sums others together, so it cannot increase their number. Thus, each block admits an $s$-sparse representation.

With that inherited sparsity, Algorithm~2 proceeds as follows. Initially, it sets $B_j^{(0)}$ to be the computational basis on each $\mathcal B_j^{(0)}$. Note that for $S_{\mathcal B_j^{(0)}} = \{x:[\rho_{\mathcal B_j^{(0)}}]_{B_j^{(0)}}(x,x)>0\}$, we have $|S_{\mathcal B_j^{(0)}}| \le s$, as the Hilbert-space dimension of each leaf block is at most $s$. Now, consider merging two child blocks $L=\mathcal B_{2j-1}^{(d)}$ and $R=\mathcal B_{2j}^{(d)}$ to form a parent block $Q=\mathcal B_j^{(d+1)}$. Suppose that $\rho_L$ and $\rho_R$ are supported on at most $s$ vectors in $B_L$ and $B_R$, respectively. Specifically, let $S_L:=\{x:[\rho_L]_{B_L}(x,x)>0\}$ and define $S_R$ analogously. Then $\lvert S_L\rvert,\lvert S_R\rvert\le s$. Then, $\rho_Q$ is supported on at most $s^2$ vectors in the product basis $B_L\times B_R$. In particular, we have
\begin{equation}
  \operatorname{supp}([\rho_Q]_{B_L \times B_R})
  \subseteq (S_L\times S_R)\times(S_L\times S_R).
  \label{eq:overview-tree-merge}
\end{equation}
Therefore, $\rho_Q$ is confined to a submatrix with at most $s^4$ entries in the basis $B_L\times B_R$. With that basis $B_L\times B_R$, we learn $\rho_Q$ by performing selected-entry tomography that estimates the submatrix of $[\rho_Q]_{B_L\times B_R}$ on $S_L\times S_R$ with at most $s^4$ matrix entries.

In this way, we can learn the reduced density matrices of $\rho$ on all blocks at level $d+1$ using the bases and retained subspaces of the child blocks at level $d$. However, if we naively repeat this procedure, the sparsity will rapidly increase with the level $d$, as it is raised to the fourth power at each merge. To prevent this, we compact the $s^4$-entry representation of $\rho_Q$ by classically searching for a product basis $B_*$ that makes $[\rho_Q]_{B_*}$ supported on at most $s$ basis vectors. Note that there exists such a basis $B_*$ as $\rho_Q$ inherits the sparsity of $\rho$. After repeating the merges until level $D$, final reconstruction searches for a basis and truncates to at most $s$ matrix entries; see Figure~\ref{fig:tree-merging}(b).

Although the sample complexity of this procedure remains polynomial, the classical computational challenge lies in finding a product basis $B_*$ and a subspace spanned by at most $s$ of its basis vectors at each merge. In the general case where the product basis in which $\rho$ is sparse is arbitrary, Algorithm~2 searches for $B_*$ by discretizing each local basis using a $\poly(\varepsilon/n)$-net and trying all possible product bases. This results in a large search space size, leading to $(ns/\varepsilon)^{O(n)}$ classical runtime.

However, the overhead of finding such a basis $B_*$ is greatly reduced when the unknown basis $B$ consists of a fixed, known finite set $\mathcal A$, i.e., $B_i\in\mathcal A$ for every qubit $i$ and thus $B\in\mathcal A^n$. In this case, the naive search space size is still exponential in $n$, $|\mathcal A|^n$. Nevertheless, we can reduce the search space by using the Donoho--Stark uncertainty principle \cite{donohoUncertaintyPrinciplesSignal1989,eladGeneralizedUncertaintyPrinciple2002,boggiattoTwoAspectsDonoho2016}, which was originally developed in the context of classical signal processing. In its matrix form, the Donoho--Stark uncertainty principle states the following: let $M$ and $N$ be non-zero $s$-sparse and $t$-sparse matrices, respectively, and suppose that $N=VMV^\dagger$ for a unitary matrix $V$. Then, we have
\begin{equation}
  \sqrt{st}\ge \frac{1}{\lVert V\rVert_{\max}^2},
  \label{eq:overview-matrix-donoho-stark}
\end{equation}
where $\lVert V\rVert_{\max}:=\max_{a,b}\lvert V_{ab}\rvert$.
As an illustrative example, take $\mathcal A$ to be the eigenbases of the Pauli operators $X$, $Y$, and $Z$. For a parent block $Q$ at one tree merge, consider that $N=[\rho_Q]_{B'}$ has at most $s^4$ entries in the current Pauli product basis $B'$, and let $M=[\rho_Q]_B$ be the compressed matrix with at most $s$ entries in another Pauli product basis $B$. Then $N=VMV^\dagger$, where $V=\bigotimes_{i\in Q}V_i$ and each $V_i$ is a single-qubit Clifford gate mapping $B_i$ to $B'_i$. If $B$ and $B'$ differ on $h$ qubits, we have
\begin{equation}
  \lVert V\rVert_{\max}=\prod_{i\in Q}\lVert V_i\rVert_{\max}=2^{-h/2}.
\end{equation}
The Donoho--Stark uncertainty principle then gives \(h\le \frac52\log_2s\), so the two product bases differ on only $O(\log s)$ qubits. Therefore, instead of searching all of $\mathcal A^{\lvert Q\rvert}$, we can enumerate only those product bases that differ from $B'$ on at most $O(\log s)$ qubits. Then, the number of possible product bases is at most $\lvert Q\rvert^{O(\log s)}$, and it leads to the $n^{O(\log s)}$ classical runtime of Algorithm~2. The same reasoning applies to the case of any fixed finite set $\mathcal A$; see Section~\ref{sec:algorithm1-results}.

The Donoho--Stark uncertainty principle does not, however, yield a useful reduction of the search space for arbitrary local bases. In that setting, the finite set $\mathcal A$ is replaced by an epsilon-net with inverse-polynomial resolution. Two bases in such a net can have a change-of-basis unitary satisfying $\lVert V_i\rVert_{\max}=1-1/\operatorname{poly}(n)$, rather than being bounded away from one. Consequently, $\lVert V\rVert_{\max}$ decays too slowly with the number $h$ of changed local bases, and Eq.~\eqref{eq:overview-matrix-donoho-stark} gives only $h\le\operatorname{poly}(n)\log s$. This bound can exceed $n$ and is therefore vacuous, leaving the search space exponentially large.

In these ways, both algorithms combine selected-entry tomography with their respective approaches to find an appropriate product basis that makes the input state $\rho$ sparse. So far, however, we have neglected the errors from the intermediate observable estimation and selected-entry tomography. Therefore, what remains is to show that both algorithms are robust to such estimation errors. The central challenge is that both algorithms are adaptive, so the estimation errors in the previous steps can propagate through the adaptive steps. In the rest of the paper, we carefully analyze such error propagation in both algorithms, and show that the final output $\widehat\rho$ is indeed close to the input state $\rho$ in trace distance with high probability.

%% file: 03_preliminaries_shared_tools.tex
\section{Preliminaries}
\label{sec:shared-tools}

A single-qubit orthonormal basis $\{\lvert \phi_0\rangle,\lvert \phi_1\rangle\}$ is represented, up to sign, by a unit Bloch vector $\mathbf n\in\mathbb R^3$. Denoting $\boldsymbol{\sigma}=(X,Y,Z)$, the corresponding basis states satisfy
\begin{equation}
  \lvert\phi_b\rangle\!\langle\phi_b\rvert
  =\frac{I+(-1)^b\mathbf n\cdot\boldsymbol{\sigma}}{2},
  \qquad b\in\{0,1\}.
  \label{eq:basis-bloch-vector}
\end{equation}
Note that $\mathbf n$ and $-\mathbf n$ represent the same basis with the two labels 0 and 1 swapped, which does not affect the sparsity structure of a state. We therefore define the label-invariant angular distance between Bloch axes by
\begin{equation}
  d_{\mathrm{ang}}(\mathbf n,\mathbf n')
  :=\arccos\left\lvert\mathbf n\cdot\mathbf n'\right\rvert.
  \label{eq:angular-distance}
\end{equation}
With this definition of angular distance, an epsilon-net for single-qubit bases is defined as follows:

\begin{definition}[Single-qubit epsilon-net]
  \label{def:single-qubit-epsilon-net}
  A finite set of unit Bloch vectors, denoted by $\mathcal N(\varepsilon)$, is a single-qubit $\varepsilon$-net if for any unit vector $\mathbf n \in \mathbb R^3$, there exists $\mathbf n' \in \mathcal N(\varepsilon)$ such that $\mathbf n' \ne \pm \mathbf n$ and $d_{\mathrm{ang}}(\mathbf n,\mathbf n')\le\varepsilon$.
\end{definition}
An elementary sphere-covering argument shows that such a net can be chosen with
\begin{equation}
  \lvert\mathcal N(\varepsilon)\rvert
  \le (C/\varepsilon)^2
\end{equation}
for a universal constant $C$. For a subset of qubits $Q\subseteq[n]$, we use $\mathcal N(\varepsilon)^Q$ to denote the set of all possible basis assignments on qubits in $Q$. Thus, for every product basis on $Q$, the product net contains a basis whose local axes are within angular distance $\varepsilon$ on every qubit.

We now introduce the \emph{leakage error}, the approximation error caused by discarding components outside a subspace. Specifically, if $A$ is an orthogonal projector, then its leakage error is defined as $\eta$ with $\operatorname{Tr}[(I-A)\rho] \le \eta$.

\begin{lemma}[Leakage error in Frobenius distance]
  \label{lem:projection-error}
  Let $\rho$ be a quantum state and $A$ an orthogonal projector. If
  $\operatorname{Tr}[(I-A)\rho]\le\eta$ for some $0 \le \eta \le 1$, then
  \begin{equation}
    \lVert\rho-A\rho A\rVert_F^2\le2\eta.
    \label{eq:projection-error}
  \end{equation}
\end{lemma}

\begin{proof}
  Let the actual leakage be
  \begin{equation}
    \eta':=\operatorname{Tr}[(I-A)\rho]\le\eta.
    \label{eq:projection-error-actual-leakage}
  \end{equation}
  Relative to $\operatorname{Im}(A)\oplus\operatorname{Im}(I-A)$, write
  \begin{equation}
    \rho=
    \begin{pmatrix}
      X         & C \\
      C^\dagger & Z
    \end{pmatrix},
    \qquad
    \operatorname{Tr}(X)=1-\eta',
    \qquad
    \operatorname{Tr}(Z)=\eta'.
    \label{eq:projection-error-block-decomposition}
  \end{equation}
  Choose orthonormal bases $\{u_i\}$ and $\{v_j\}$ for $\operatorname{Im}(A)$ and $\operatorname{Im}(I-A)$, respectively. Since $\rho$ is positive semidefinite, we have
  \begin{equation}
    \lvert\langle u_i\rvert C\lvert v_j\rangle\rvert^2
    \le
    \langle u_i\rvert X\lvert u_i\rangle
    \langle v_j\rvert Z\lvert v_j\rangle,
  \end{equation}
  as the $2\times 2$ submatrix of $\rho$ on $u_i,v_j$ is also positive semidefinite. Summing this inequality over $i$ and $j$ yields
  \begin{equation}
    \lVert C\rVert_F^2
    \le\operatorname{Tr}(X)\operatorname{Tr}(Z)
    =\eta'(1-\eta').
    \label{eq:projection-error-cross-block}
  \end{equation}
  Since $Z\succeq0$, we also have
  \begin{equation}
    \lVert Z\rVert_F^2\le(\operatorname{Tr}Z)^2=(\eta')^2.
    \label{eq:projection-error-omitted-block}
  \end{equation}
  Substituting Eqs.~\eqref{eq:projection-error-cross-block} and~\eqref{eq:projection-error-omitted-block}
  therefore yields
  \begin{equation}
    \begin{aligned}
      \lVert\rho-A\rho A\rVert_F^2
       & =2\lVert C\rVert_F^2+\lVert Z\rVert_F^2 \\
       & \le2\eta'(1-\eta')+(\eta')^2
      \le2\eta'
      \le2\eta.
    \end{aligned}
  \end{equation}
\end{proof}

\begin{lemma}[Leakage error in trace distance]
  \label{lem:projection-error-trace}
  Let $\rho$ be a quantum state and $A$ an orthogonal projector. If $\operatorname{Tr}[(I-A)\rho]\le\eta$ for some $0\le\eta\le1$, then
  \begin{equation}
    \lVert\rho-A\rho A\rVert_1\le2\sqrt\eta.
    \label{eq:projection-error-trace}
  \end{equation}
\end{lemma}

\begin{proof}
  Again, let the actual leakage be
  \begin{equation}
    \eta':=\operatorname{Tr}[(I-A)\rho]\le\eta.
  \end{equation}
  Take a purification $\lvert\psi\rangle$ of $\rho$, and set $\lvert\varphi\rangle:=(A\otimes I)\lvert\psi\rangle$. Then, we have
  \begin{equation}
    \lVert\lvert\psi\rangle-\lvert\varphi\rangle\rVert_2=\sqrt{\operatorname{Tr}[(I-A)\rho]} = \sqrt{\eta'}.
  \end{equation}
  Since trace distance cannot increase under partial trace, we have
  \begin{equation}
    \begin{aligned}
      \lVert\rho-A\rho A\rVert_1
      &\le\left\lVert\lvert\psi\rangle\!\langle\psi\rvert
      -\lvert\varphi\rangle\!\langle\varphi\rvert\right\rVert_1\\
      &=\left\lVert(\lvert\psi\rangle-\lvert\varphi\rangle)\langle\psi\rvert
      +\lvert\varphi\rangle(\langle\psi\rvert-\langle\varphi\rvert)\right\rVert_1\\
      &\le\bigl(\lVert\lvert\psi\rangle\rVert_2
      +\lVert\lvert\varphi\rangle\rVert_2\bigr)
      \lVert\lvert\psi\rangle-\lvert\varphi\rangle\rVert_2\\
      &\le2\sqrt{\eta'}\le2\sqrt\eta.
    \end{aligned}
  \end{equation}
  Here, we used the identity $\lVert\lvert \psi \rangle\!\langle \phi\rvert\rVert_1=\lVert \psi\rVert_2\lVert \phi\rVert_2$ and $\lVert\lvert\varphi\rangle\rVert_2\le1$.
\end{proof}

Thus leakage controls both the Frobenius and trace-norm errors due to projections. To propagate leakage error through a sequence of projectors, we also need the following inequality:

\begin{lemma}[Triangle inequality for leakage error]
  \label{lem:projection-leakage-triangle}
  Let $\rho$ be a state and $A,P$ orthogonal projectors. Then,
  \begin{equation}
    \sqrt{\operatorname{Tr}[(I-P)\rho]}
    \le
    \sqrt{\operatorname{Tr}[(I-A)\rho]}
    +\sqrt{\operatorname{Tr}[(I-P)A\rho A]}.
    \label{eq:projection-leakage-triangle}
  \end{equation}
\end{lemma}

\begin{proof}
  Let $\lvert\psi\rangle$ be the purification of $\rho$. Then, the leakage errorr are given by
  \begin{equation}
    \operatorname{Tr}[(I-P)\rho]
    =\lVert(I-P)\otimes I\lvert\psi\rangle\rVert_2^2 \quad \text{and}\quad
    \operatorname{Tr}[(I-A)\rho]
    =\lVert(I-A)\otimes I\lvert\psi\rangle\rVert_2^2.
  \end{equation}
  Similarly, we have
  \begin{equation}
    \lVert(I-P)A\otimes I\lvert\psi\rangle\rVert_2^2
    =\operatorname{Tr}[(I-P)A\rho A].
  \end{equation}
  Meanwhile, inserting the identity $I = (I-A)+A$ on $(I-P)\otimes I\lvert\psi\rangle$ gives:
  \begin{equation}
    (I-P)\otimes I\lvert\psi\rangle
    =(I-P)(I-A)\otimes I\lvert\psi\rangle
    +(I-P)A\otimes I\lvert\psi\rangle.
    \label{eq:projection-leakage-triangle-vector-decomposition}
  \end{equation}
  Applying the triangle inequality to Eq.~\eqref{eq:projection-leakage-triangle-vector-decomposition}, we have
  \begin{equation}
    \begin{aligned}
      \sqrt{\operatorname{Tr}[(I-P)\rho]}
      &=\lVert(I-P)\otimes I\lvert\psi\rangle\rVert_2\\
      &\le
      \lVert(I-P)(I-A)\otimes I\lvert\psi\rangle\rVert_2
      +\lVert(I-P)A\otimes I\lvert\psi\rangle\rVert_2\\
      &\le
      \lVert(I-A)\otimes I\lvert\psi\rangle\rVert_2
      +\lVert(I-P)A\otimes I\lvert\psi\rangle\rVert_2\\
      &=
      \sqrt{\operatorname{Tr}[(I-A)\rho]}
      +\sqrt{\operatorname{Tr}[(I-P)A\rho A]}.
    \end{aligned}
  \end{equation}
\end{proof}

%% file: 04_selected_entry_tomography.tex
\section{Selected-entry tomography}
\label{sec:selected-entry-tomography}

\subsection{Estimating selected entries and submatrices}
\label{sec:selected-entries-submatrices}

Selected-entry tomography is an essential tool in both Algorithm~1 and Algorithm~2. Given a known product basis $B=(B_1,\dots,B_n)$, it efficiently estimates the matrix element $[\rho]_B(x,y)$ of a selected entry $(x,y)$. Notably, selected-entry tomography achieves this goal only by performing single-qubit measurements. To be specific, denote $B=(B_1,\dots,B_n)$ with $B_i = \{|\phi^{(i)}_0\rangle, |\phi^{(i)}_1\rangle\}$ for each $i\in[n]$. Let $V_i$ be the single-qubit unitary that transforms computational basis states to $B_i$, i.e., $V_i|0\rangle=|\phi^{(i)}_0\rangle$ and $V_i|1\rangle=|\phi^{(i)}_1\rangle$.

Then, the selected-entry tomography proceeds as follows. We choose $z\in\{0,1\}^n$ uniformly at random, and measure the following observable:
\begin{equation}
  \bigotimes_{i\in[n]} V_i \left(i^{(x_i\oplus y_i)z_i}X^{x_i\oplus y_i}Z^{z_i}\right)V_i^\dagger.
  \label{eq:single-entry-tomography-pauli-observable}
\end{equation}
Note that $i^{(x_i\oplus y_i)z_i}X^{x_i\oplus y_i}Z^{z_i} \in \{I,X,Y,Z\}$, so this observable is a tensor product of single-qubit Pauli observables rotated by the $V_i$. Therefore, this measurement can be realized by single-qubit measurements, and we get a sample of $\pm1$ from each copy. Let $O^{(j)}(x,y;z^{(j)})$ be the $j$-th measurement outcome for $j=1,2,\dots,M$, where $z^{(j)}\in\{0,1\}^n$ is the randomly chosen $z$ for the $j$-th copy of $\rho$. Finally, we output
\begin{equation}
  [\widehat{\rho}]_B(x,y) := \frac{1}{M} \sum_{j=1}^M (-1)^{z^{(j)}\cdot x} i^{-(x\oplus y)\cdot z^{(j)}} O^{(j)}(x,y;z^{(j)}).
  \label{eq:single-entry-tomography-entry-estimator}
\end{equation}
The following theorem shows that this procedure efficiently estimates the matrix element $[\rho]_B(x,y)$.

\begin{theorem}[Selected-entry tomography]
  \label{thm:single-entry-tomography}
  Fix a product basis $B=(B_1,\dots,B_n)$ and $x,y\in\{0,1\}^n$. For an $n$-qubit state $\rho$, let $\rho_{xy}:=[\rho]_B(x,y)$. The selected-entry tomography outputs $\widehat\rho_{xy}$ such that $\lvert\widehat\rho_{xy}-\rho_{xy}\rvert\le\varepsilon$ with probability at least $1-\delta$, using at most $\frac{4}{\varepsilon^2}\log\frac{4}{\delta}$ copies of $\rho$.
\end{theorem}

\begin{proof}
  Let $V_B:=\bigotimes_{i=1}^nV_i$ and $\rho_B:=V_B^\dagger\rho V_B$, so $[\rho]_B(x,y)=\langle x\rvert\rho_B\lvert y\rangle$. Averaging simultaneously over the uniform choice of $z$ and the measurement outcome gives
  \begin{equation}
    \begin{aligned}
      \mathbb E\left[\widehat\rho_{xy}\right]
       & =\mathbb E_{z\sim\{0,1\}^n} \left[(-1)^{z\cdot x}i^{-(x\oplus y)\cdot z}
      \operatorname{Tr}\!\left(i^{(x\oplus y)\cdot z}X^{x\oplus y}Z^z\rho_B\right) \right]\\
       & =\mathbb E_{z\sim\{0,1\}^n}\left[(-1)^{z\cdot x}\operatorname{Tr}\!\left(X^{x\oplus y}Z^z\rho_B\right)\right]\\
       & =\sum_u[\rho]_B(u,u\oplus x\oplus y)\mathbb E_{z\sim\{0,1\}^n}\left[(-1)^{z\cdot(x\oplus u)}\right]\\
       & =\rho_{xy}.
    \end{aligned}
    \label{eq:single-entry-tomography-unbiasedness}
  \end{equation}
  Thus $\widehat{\rho}_{xy}$ is an unbiased estimator of $\rho_{xy}$. Each phase-adjusted outcome in Eq.~\eqref{eq:single-entry-tomography-entry-estimator} has modulus one, so its real and imaginary parts lie in
  $[-1,1]$. Moreover, $\lvert \widehat{\rho}_{xy} - \rho_{xy}\rvert\ge\varepsilon$ implies
  $\lvert\operatorname{Re}(\widehat{\rho}_{xy} - \rho_{xy})\rvert\ge\varepsilon/\sqrt2$ or
  $\lvert\operatorname{Im}(\widehat{\rho}_{xy} - \rho_{xy})\rvert\ge\varepsilon/\sqrt2$. Therefore, the two-sided Hoeffding bound applied to both real and imaginary parts, followed by a union bound, yields
  \begin{equation}
    \mathbb P\!\left[
      \lvert\widehat\rho_{xy}-\rho_{xy}\rvert\ge\varepsilon
      \right]
    \le4e^{-M\varepsilon^2/4}.
    \label{eq:single-entry-tomography-hoeffding}
  \end{equation}
  Choosing $M\ge\frac{4}{\varepsilon^2}\log\frac{4}{\delta}$ concludes the proof.
\end{proof}

Using this selected-entry tomography in Theorem~\ref{thm:single-entry-tomography}, it is straightforward to obtain an estimate of a small submatrix of $[\rho]_B$:

\begin{corollary}
  \label{cor:selected-submatrix-tomography}
  Fix a product basis $B=(B_1,\dots,B_n)$ and $S\subseteq\{0,1\}^n$ of size $t$. For an $n$-qubit state $\rho$, let $\rho_S:=[\rho]_B(S,S)$ be the submatrix of $\rho$ in the basis $B$ corresponding to the indices in $S\times S$. Running selected-entry tomography on every entry in $S\times S$ can output a Hermitian matrix $\widehat\rho_S$ such that $\lVert\widehat\rho_S-\rho_S\rVert_F\le\varepsilon$ with probability at least $1-\delta$, using at most $\frac{4t^4}{\varepsilon^2}\log\frac{4t^2}{\delta}$ copies of $\rho$.
\end{corollary}

\begin{proof}
  For each $(x,y)\in S\times S$, run the selected-entry tomography on separate fresh copies with error $\varepsilon/t$ and success probability $1-\delta/t^2$, and let $Z$ collect the resulting raw estimates. Define
  \begin{equation*}
    \widehat\rho_S:=\frac{Z+Z^\dagger}{2},
  \end{equation*}
  which is Hermitian. By a union bound, with probability at least $1-\delta$, every entry satisfies $\lvert[Z]_{xy}-[\rho_S]_{xy}\rvert\le\varepsilon/t$, and hence $\lVert Z-\rho_S\rVert_F\le\sqrt{t^2(\varepsilon/t)^2}=\varepsilon$. Since $\rho_S=\rho_S^\dagger$,
  \begin{equation}
    \lVert\widehat\rho_S-\rho_S\rVert_F
    \le\lVert Z-\rho_S\rVert_F
    \le\varepsilon.
  \end{equation}
  Theorem~\ref{thm:single-entry-tomography} bounds the total number of copies by
  \begin{equation}
    t^2\left(\frac{4t^2}{\varepsilon^2}\log\frac{4t^2}{\delta}\right)
    =\frac{4t^4}{\varepsilon^2}\log\frac{4t^2}{\delta}.
  \end{equation}
\end{proof}

Therefore, once the basis $B$ is chosen, selected-entry tomography can efficiently estimate any small submatrix of $[\rho]_B$ only with single-qubit measurements, which is actively used in both of our algorithms.

\subsection{Learning sparse states in a known product basis}
\label{sec:known-basis-sparse-learning}

Selected-entry tomography directly gives a learning algorithm for sparse states in a known product basis. Specifically, when the basis $B$ in which $\rho$ is $s$-sparse is known, one can easily identify the submatrix of $[\rho]_B$ containing almost all nonzero entries by measuring in $B$. Then, one can apply selected-entry tomography to estimate the entries of this submatrix.

Let $S:=\{x\in\{0,1\}^n:[\rho]_B(x,x)>0\}$. Since $\rho$ is positive semidefinite, we have
\begin{equation*}
  \lvert[\rho]_B(x,y)\rvert^2\le[\rho]_B(x,x)[\rho]_B(y,y),
\end{equation*}
so every nonzero entry of $[\rho]_B$ lies in $S\times S$. Note that $\lvert S\rvert\le s$. With this setup, the algorithm proceeds as follows:
\begin{itemize}
  \item \textbf{Step 1: Capture bitstrings.}
  Measure
  \begin{equation}
    N_1:=\left\lceil\frac{16}{\varepsilon^2}
    \left(s\log2+\log\frac2\delta\right)\right\rceil
    \label{eq:known-basis-capture-copies}
  \end{equation}
  copies of $\rho$ in $B$ and get output bitstrings. Let $S'$ be the set of unique output bitstrings.
  
  \item \textbf{Step 2: Estimate the selected entries.}
  For each $(x,y)\in S'\times S'$, apply the selected-entry tomography in Theorem~\ref{thm:single-entry-tomography} with accuracy $\varepsilon/4s$ and failure probability $\delta/2s^2$, and denote the resulting estimate by $A_{xy}$.
  
  \item \textbf{Step 3: Truncate and return the output.}
  Retain an estimated entry exactly when its magnitude is strictly greater than $\varepsilon/4s$, defining
  \begin{equation}
    [\widehat\rho]_B(x,y):=
    \begin{cases}
      A_{xy}, & (x,y)\in S'\times S'\text{ and }\lvert A_{xy}\rvert>\varepsilon/4s,\\
      0, & \text{otherwise}.
    \end{cases}
    \label{eq:known-basis-thresholding}
  \end{equation}
  Then, return $\widehat\rho$.
\end{itemize}

\begin{theorem}[Learning sparse states in a known basis]
  \label{thm:known-basis-sparse-learning}
  Suppose that $\rho$ is $s$-sparse in a known product basis $B=(B_1,\dots,B_n)$ and let $0<\varepsilon, \delta <1$. The above algorithm outputs a matrix $\widehat\rho$ that is $s$-sparse in $B$ such that $\lVert\widehat\rho-\rho\rVert_1\le\varepsilon$ with probability at least $1-\delta$. The algorithm uses only single-qubit measurements and requires
  \begin{equation}
    \begin{aligned}
      N&=O\!\left(s^4\varepsilon^{-2}\log\frac{s}{\delta}\right)
      &&\text{copies of $\rho$, and}\\
      T&=O\!\left(ns^4\varepsilon^{-2}\log\frac{s}{\delta}\right)
      &&\text{classical runtime}.
    \end{aligned}
    \label{eq:known-basis-learning-resources}
  \end{equation}
\end{theorem}

\begin{proof}
  To analyze Step~1, note that $S'\subseteq S$, so $s':=\lvert S'\rvert\le s$. We write $p(x) = [\rho]_B(x,x)$ for $x\in \{0,1\}^n$ and $p(U):=\sum_{x\in U}p(x)$ for $U\subseteq \{0,1\}^n$. Let $x^{(1)}, \dots,x^{(N_1)}$ be the measurement bitstring outcomes.
  
  If $p(U)>\varepsilon^2/16$ for $U \subseteq S$, the probability that none of $x^{(1)}, \dots,x^{(N_1)}$ belongs to $U$ is $(1-p(U))^{N_1}\le e^{-N_1\varepsilon^2/16}$. The event $p(S\setminus S')>\varepsilon^2/16$ implies that some such subset $U$ has not been observed, so a union bound over the at most $2^{\lvert S\rvert}$ subsets gives
  \begin{equation}
    \mathbb P\!\left[p(S\setminus S')>\frac{\varepsilon^2}{16}\right]
    \le2^{\lvert S\rvert}e^{-N_1\varepsilon^2/16}
    \le\frac\delta2.
    \label{eq:known-basis-capture-probability}
  \end{equation}
  Thus, with probability at least $1-\delta/2$, the unobserved probability mass $p(S\setminus S')$ is at most $\varepsilon^2/16$.

  For $x\in\{0,1\}^n$, let $\lvert\phi_x\rangle:=\bigotimes_{i=1}^n\lvert\phi^{(i)}_{x_i}\rangle$ where $B_i = \{\lvert\phi^{(i)}_0\rangle, \lvert\phi^{(i)}_1\rangle\}$, and define $P:=\sum_{x\in S'}\lvert\phi_x\rangle\!\langle\phi_x\rvert$. Write the leakage error
  \begin{equation}
    \eta:=\operatorname{Tr}[(I-P)\rho]=p(S\setminus S') \le \varepsilon^2/16.
  \end{equation}
  Applying Lemma~\ref{lem:projection-error-trace} to $P\rho P$ gives
  \begin{equation}
    \lVert\rho-P\rho P\rVert_1
    \le2\sqrt\eta
    \le\varepsilon/2.
    \label{eq:known-basis-gentle-projection}
  \end{equation}

  Now we move on to Step~2. Theorem~\ref{thm:single-entry-tomography} and the union bound give, with probability at least $1-\delta/2$,
  \begin{equation}
    \lvert A_{xy}-[\rho]_B(x,y)\rvert\le\varepsilon/4s
    \qquad\text{for every }(x,y)\in S'\times S'.
    \label{eq:known-basis-entry-accuracy}
  \end{equation}
  There are $(s')^2\le s^2$ entries to estimate, and obtaining the selected-entry tomography estimates requires
  \begin{equation}
    N_2:=s'^2\left\lceil\frac{64s^2}{\varepsilon^2}\log\frac{8s^2}{\delta}\right\rceil
    =O\!\left(s^4\varepsilon^{-2}\log\frac{s}{\delta}\right)
    \label{eq:known-basis-entry-copies}
  \end{equation}
  copies of $\rho$ in total.

  Note that every matrix returned in Step~3 has at most $s$ nonzero entries provided that the premises of Steps~1 and 2 are satisfied with probability at least $1-\delta$. Indeed, if $[\rho]_B(x,y)=0$, then $\lvert A_{xy}\rvert\le\varepsilon/(4s)$, so the entry is discarded in Step~3. Thus, at most $s$ entries survive in $\widehat\rho$. A retained entry has error at most $\varepsilon/(4s)$, whereas a discarded true entry in $S'\times S'$ has magnitude at most $\lvert[\rho]_B(x,y)-A_{xy}\rvert+\lvert A_{xy}\rvert\le\varepsilon/(2s)$. Consequently, $\widehat\rho-P\rho P$ has at most $s$ nonzero entries and acts on the $s'$-dimensional image of $P$, giving
  \begin{equation}
    \begin{aligned}
      \lVert\widehat\rho-P\rho P\rVert_F
      &\le\frac{\varepsilon}{2\sqrt{s}},\\
      \lVert\widehat\rho-P\rho P\rVert_1
      &\le\sqrt{s'}\,\lVert\widehat\rho-P\rho P\rVert_F
      \le\frac{\varepsilon}{2}\sqrt{\frac{s'}s}
      \le\frac\varepsilon2.
    \end{aligned}
    \label{eq:known-basis-reconstruction-error}
  \end{equation}
  Combining Eqs.~\eqref{eq:known-basis-gentle-projection} and~\eqref{eq:known-basis-reconstruction-error} by the triangle inequality yields $\lVert\widehat\rho-\rho\rVert_1\le\varepsilon$.

  The total number of copies used is $N=N_1+N_2=O(s^4\varepsilon^{-2}\log(s/\delta))$. Since processing each measurement outcome takes $O(n)$ time, the total classical runtime is $T=O(ns^4\varepsilon^{-2}\log(s/\delta))$.
\end{proof}

%% file: 05_low_order_pauli_moments.tex
\section{Learning with low-order Pauli moments}
\label{sec:continuous-fixed-sparsity}

In this section, we present Algorithm~1, which uses low-order Pauli moments to learn $s$-sparse states in an arbitrary unknown product basis, as briefly outlined in Section~\ref{sec:algorithm2-overview}. We first give an explicit description of Algorithm~1 and the result. We then analyze the construction of candidate product bases from low-order Pauli moments. Finally, we prove completeness and soundness of reconstruction and bound the number of copies of $\rho$ and the classical runtime.

\subsection{Explicit algorithm and result}
\label{sec:continuous-fixed-sparsity-algorithm}

Algorithm~1 takes independently prepared copies of the unknown $s$-sparse state $\rho$, accuracy of the output $\varepsilon$, and failure probability $\delta$. As outlined in Section~\ref{sec:algorithm2-overview}, Algorithm~1 generates a list of candidate product bases and then, for each candidate basis, measures fresh copies of $\rho$ in that basis to test whether the state is $s$-sparse in that basis. If a candidate basis passes the test, Algorithm~1 performs selected-entry tomography with that basis to reconstruct the state. Here, we present the explicit procedure of Algorithm~1 in detail.

For each nonempty $W\subseteq[n]$ and $\alpha=(\alpha_i)_{i\in W}\in\{X,Y,Z\}^W$, we define the Pauli moment tensor $T_W$ by
\begin{equation}
  [T_W]_\alpha
  :=\operatorname{Tr}\!\left[
    \left(\bigotimes_{i\in W}\alpha_i\right)\rho
    \right],
  \label{eq:product-pauli-tensor}
\end{equation}
where tensor indices are ordered by increasing qubit index. We set
\begin{equation}
  \ell_0:=\min\{n,\lfloor\log_2s\rfloor+1\},
  \qquad
  k_0:=s+\lceil2\log_2s\rceil.
  \label{eq:pauli-moment-learning-cutoffs}
\end{equation}
Here, $\ell_0$ will bound the order of the estimated Pauli moments and $k_0$ will bound the size of each enumerated qubit set $K$ when constructing the list. We also choose the following parameters:
\begin{equation}
  \varepsilon_1
  :=\frac{\varepsilon^2}
  {2^{12}s^{\log_2s+2}3^{\ell_0/2}\sqrt n},
  \qquad
  \varepsilon_2:=\frac{\varepsilon_1}{\ell_0},
  \qquad
  \xi:=\frac{\varepsilon^2}{32},
  \qquad
  \theta:=\frac{\varepsilon}{3s^{3/2}}.
  \label{eq:pauli-moment-learning-final-tolerances}
\end{equation}
With these settings, the algorithm proceeds as follows:

\begin{itemize}
  \item \textbf{Step 1: Generate candidate product bases.}
        \begin{enumerate}[label=(\alph*)]
          \item \textbf{Estimate the low-order Pauli moments.}
                For every nonempty $W$ with $\lvert W\rvert\le\ell_0$, estimate every entry of $T_W$ to target entrywise error $\le\varepsilon_1$. Record all estimates $\widehat T_W$.

          \item \textbf{Enumerate partial basis assignments.}
                Enumerate every subset $K\subseteq[n]$ with $\lvert K\rvert\le k_0$ and every basis assignment $(\widetilde{\mathbf n}_j)_{j\in K}\in\mathcal N(\varepsilon_2)^K$. We call each pair $(K,(\widetilde{\mathbf n}_j)_{j\in K})$ a \emph{partial basis assignment}; it specifies both the qubit set and the proposed Bloch vectors on that set.

          \item \textbf{Recover the remaining bases and output candidate product bases.}
                For each partial basis assignment $(K,(\widetilde{\mathbf n}_j)_{j\in K})$ and each $i\in[n]\setminus K$, scan the sets $W\ni i$ with $\lvert W\rvert\le\ell_0$. For each such $W$, contract $\widehat T_W$ only with the proposed Bloch vectors on $K\cap W$:
                \begin{equation}
                  \widehat T_W
                  \bigtimes_{j\in K\cap W}\widetilde{\mathbf n}_j^{T}.
                  \label{eq:pauli-moment-learning-procedure-contraction}
                \end{equation}
                Flatten this contracted tensor in mode $i$, with rows indexed by $\alpha_i$ and columns indexed by the remaining Pauli labels, and perform SVD. If its leading singular value is $\ge 3/(16s^{\log_2s})$, retain the real unit leading left singular vector $\widetilde{\mathbf n}_i$ and stop the scan for this qubit; otherwise, continue to the next $W$. If no $W$ passes for some $i$, discard the partial basis assignment. Otherwise, complete the full collection $\{\widetilde{\mathbf n}_i\}_{i\in[n]}$ to single-qubit bases, and add the resulting candidate product basis to the list. After processing every partial basis assignment, output the list of candidate product bases.
        \end{enumerate}

  \item \textbf{Step 2: Bitstring test.}
        For each candidate basis, fix that basis before measuring fresh copies of $\rho$ in it. Let $q$ denote the output bitstring distribution in that basis, and estimate it by the empirical distribution $\widehat q$ to target precision
        \begin{equation}
          \lvert\widehat q(T)-q(T)\rvert\le\xi/2
          \qquad\text{for all }T\subseteq\{0,1\}^n
          \text{ with }\lvert T\rvert\le s.
        \end{equation}
        From $\widehat q$, choose a set
        \begin{equation}
          S\in\operatorname*{arg\,max}_{T\subseteq\{0,1\}^n:\,\lvert T\rvert\le s}\widehat q(T).
        \end{equation}
        Keep the candidate basis together with its selected set $S$ when $\widehat q(S^c)\le\xi$; otherwise, discard it.

  \item \textbf{Step 3: Sparsity test with selected-entry tomography.}
        For each candidate basis that passed Step~2, use separate fresh copies of $\rho$ to perform selected-entry tomography on the $S\times S$ submatrix with accuracy $\theta/2$ for each entry. Given the raw estimate $Z$, define the symmetrized matrix $\widetilde M = \frac{1}{2}(Z + Z^\dagger)$. Denote
        \begin{equation}
          K_{\mathrm{ent}}
          :=\left\{(x,y)\in S\times S:
          \lvert[\widetilde M]_{xy}\rvert\ge\theta\right\}.
          \label{eq:pauli-moment-learning-procedure-threshold-set}
        \end{equation}
        The candidate basis passes when $\lvert K_{\mathrm{ent}}\rvert\le s$; otherwise, discard it. Process all candidate bases and return the first accepted candidate basis together with the matrix $M$ obtained by retaining the entries of $\widetilde M$ on $K_{\mathrm{ent}}$ and setting all other entries to zero. If no candidate basis is accepted, report failure.
\end{itemize}

With this procedure, Algorithm~1 achieves the following:

\begin{theorem}[Algorithm~1; arbitrary continuous product basis]
  \label{thm:pauli-moment-learning}
  Suppose that $\rho$ is $s$-sparse in an arbitrary unknown product basis, and let $0<\varepsilon, \delta <1$. With probability at least $1-\delta$, Algorithm~1 only uses single-qubit measurements and outputs an $s$-sparse operator $\widehat\rho$ such that $\lVert\widehat\rho-\rho\rVert_1\le\varepsilon$, with
  \begin{equation}
    \begin{aligned}
      N&=2^{O(s(\log s)^2)}\left(\frac n\varepsilon\right)^{O(s)}\polylog\!\left(\frac1\delta\right)
      &&\text{copies of $\rho$, and}\\
      T&=2^{O(s(\log s)^2)}\left(\frac n\varepsilon\right)^{O(s)}\polylog\!\left(\frac1\delta\right)
      &&\text{classical runtime}.
    \end{aligned}
    \label{eq:pauli-moment-learning-resources}
  \end{equation}
\end{theorem}

The proof follows the three steps of Algorithm~1. For Step~1, Section~\ref{sec:sparse-distribution-fourier-analysis} and Section~\ref{sec:exact-basis-recovery} explain how the local bases on a small set of qubits allow the remaining bases to be recovered from low-order Pauli moment tensors. Section~\ref{sec:robust-basis-recovery} shows that Step~1 robustly produces an accurate candidate basis under estimation errors. For Steps~2 and 3, Section~\ref{sec:reconstruction-certification} proves that the output of the algorithm meets the required accuracy. Finally, Section~\ref{sec:pauli-moment-learning-resource-assembly} combines these guarantees and counts the required copies of $\rho$ and classical runtime of all three steps to prove Theorem~\ref{thm:pauli-moment-learning}.

\subsection{Fourier analysis of sparse distributions}
\label{sec:sparse-distribution-fourier-analysis}

Fix the product basis $B$ in which the state is sparse and write $p(x)=[\rho]_B(x,x)$. Then, this diagonal distribution is supported on at most $s$ bitstrings. For $g:\{0,1\}^n\to\mathbb R$ and $W\subseteq[n]$, we denote the unnormalized Fourier coefficient $\widehat g(W)$ as follows:
\begin{equation}
  \widehat g(W):=\sum_{x\in\{0,1\}^n}g(x)\prod_{j\in W}(-1)^{x_j}.
  \label{eq:unnormalized-fourier}
\end{equation}
For $x\in\{0,1\}^n$, $i\in[n]$, and $b\in\{0,1\}$, let $x^{(i\mapsto b)} \in \{0,1\}^n$ be the bitstring obtained by replacing the $i$-th bit with $b$. With this notation, we define the discrete derivatives as follows.
\begin{definition}(Discrete derivative~\cite{odonnell2014analysis})
  For $p:\{0,1\}^n\to\mathbb R$ and $i\in[n]$, the discrete derivative is given by
  \begin{equation}
    (D_ip)(x)
    :=\frac{p(x^{(i\mapsto0)})-p(x^{(i\mapsto1)})}{2}.
  \end{equation}
\end{definition}
We define the set of low-contrast qubits $L$ as qubits with small discrete derivatives:
\begin{equation}
  L:=\{i\in[n]:\lVert D_ip\rVert_1\le1/2\}.
\end{equation}
We will show that $\lvert L\rvert\le2\log_2s$ and that every qubit outside $L$ belongs to a subset $W\subseteq[n]$ with $|W| \le \ell_0$ such that $\lvert\widehat p(W)\rvert > 1/(4s^{\log_2s})$. To this end, we first show the following lemma.

\begin{lemma}[Low-degree Fourier coefficient of a sparse function]
  \label{lem:sparse-low-degree-fourier-coefficient}
  Let $n$ and $s$ be positive integers. If $g:\{0,1\}^n\to\mathbb R$ is supported on at most $s$ bitstrings, then there exists $W\subseteq[n]$ such that
  \begin{equation}
    \lvert W\rvert\le\log_2s,
    \qquad
    \lvert\widehat g(W)\rvert
    \ge\frac{\lVert g\rVert_1}{2s^{\log_2s}}.
    \label{eq:sparse-fourier-bound}
  \end{equation}
\end{lemma}

\begin{proof}
  The claim is immediate when $g=0$. Otherwise, let
  $r:=\lvert\operatorname{supp}(g)\rvert\ge1$ denote the actual support size.
  For any function $h:\{0,1\}^n\to\mathbb R$ and integer $k\ge0$, write
  \begin{equation}
    d:=\lfloor\log_2r\rfloor,
    \qquad
    A_k(h):=\max_{W\subseteq[n]:\,\lvert W\rvert\le k}\lvert\widehat h(W)\rvert.
  \end{equation}
  We prove
  \begin{equation}
    \lVert g\rVert_1\le 2r^d A_d(g)
    \label{eq:sparse-fourier-auxiliary}
  \end{equation}
  by induction on $r$. For the base case $r=1$, we have $d=0$ and $A_0(g)=\lvert\widehat g(\varnothing)\rvert=\lVert g\rVert_1$, so Eq.~\eqref{eq:sparse-fourier-auxiliary} holds.

  For the induction step, let $r\ge2$ and assume that Eq.~\eqref{eq:sparse-fourier-auxiliary} holds for every nonzero real function on $\{0,1\}^n$ with support size strictly smaller than $r$. Let $\Omega=\operatorname{supp}(g)$.
  For each $j\in [n]$, choose a majority bit $b_j = \arg\max_{b\in\{0,1\}} \lvert\{ x\in \Omega: x_j=b \}\rvert$. If $\lvert\{ x\in \Omega: x_j=0 \}\rvert = \lvert\{ x\in \Omega: x_j=1 \}\rvert$, take $b_j=0$. Let
  \begin{equation}
    \Omega_j:=\{x\in\Omega:x_j\ne b_j\}.
  \end{equation}
  Define $h_j:\{0,1\}^n\to\mathbb R$ by setting $h_j(x)=g(x)$ on $\Omega_j$ and $h_j(x)=0$ elsewhere, and write $r_j=\lvert\Omega_j\rvert$.
  If $W\subseteq[n]\setminus\{j\}$, we have
  \begin{equation}
    \begin{aligned}
    \widehat h_j(W)
    &=\sum_x \frac{h_j(x)+(-1)^{x_j + b_j + 1}h_j(x)}{2}\prod_{k\in W}(-1)^{x_k}\\
    &=\frac12\left(
    \widehat g(W)+(-1)^{b_j+1}\widehat g(W\cup\{j\})
    \right).
    \end{aligned}
    \label{eq:sparse-fourier-restriction}
  \end{equation}
  The bit $x_j$ is fixed at $1-b_j$ on the support of $h_j$, so $\widehat h_j(W)=(-1)^{1-b_j}\widehat h_j(W\setminus\{j\})$ when $j\in W$. Applying Eq.~\eqref{eq:sparse-fourier-restriction} to sets that do not contain $j$ therefore bounds every coefficient of $h_j$ of degree at most $d-1$ by $A_d(g)$, giving
  \begin{equation}
    A_{d-1}(h_j)\le A_d(g).
  \end{equation}

  Since $b_j$ is a majority bit, $0\le r_j\le\lfloor r/2\rfloor<r$. If $r_j=0$, then $h_j=0$ and contributes no mass. If $r_j=1$, the exact singleton identity gives
  \begin{equation}
    \lVert h_j\rVert_1
    =A_0(h_j)\le A_{d-1}(h_j)\le A_d(g),
  \end{equation}
  where $d\ge1$ because $r\ge2$. For $r_j\ge2$, set $d_j:=\lfloor\log_2r_j\rfloor$. Then $1\le d_j\le d-1$, and the induction hypothesis applies because $r_j<r$. Using $2r_j\le r$, we obtain
  \begin{equation}
    \begin{aligned}
      \lVert h_j\rVert_1
      &\le 2r_j^{d_j}A_{d_j}(h_j)\\
      &\le (2r_j)^{d_j}A_{d_j}(h_j)\\
      &\le r^{d-1}A_d(g).
    \end{aligned}
  \end{equation}
  Here $2\le2^{d_j}$, $(2r_j)^{d_j}\le r^{d-1}$, and $A_{d_j}(h_j)\le A_{d-1}(h_j)\le A_d(g)$. Since $r^{d-1}\ge1$, the empty and singleton cases obey the same bound. Thus, for every $j\in[n]$,
  \begin{equation}
    \lVert h_j\rVert_1\le r^{d-1}A_d(g).
    \label{eq:sparse-fourier-minority-mass}
  \end{equation}

  Let $x^\star=(b_1,\ldots,b_n)$ be the simultaneous majority bitstring. Suppose first that $x^\star\notin\Omega$. Every $x\in\Omega$ then differs from $x^\star$ in at least one bit position; choose one such position $j(x)$. Since $x_{j(x)}\ne b_{j(x)}$, the bitstring $x$ belongs to $\Omega_{j(x)}$. Thus the set $J:=\{j(x):x\in\Omega\}$ satisfies
  \begin{equation}
    \lvert J\rvert\le r,
    \qquad
    \Omega=\bigcup_{j\in J}\Omega_j.
  \end{equation}
  Using $h_j(x)=g(x)$ on $\Omega_j$ and zero elsewhere, followed by Eq.~\eqref{eq:sparse-fourier-minority-mass}, gives
  \begin{equation}
    \begin{aligned}
      \lVert g\rVert_1
      &=\sum_{x\in\Omega}\lvert g(x)\rvert\\
      &\le\sum_{j\in J}\sum_{x\in\Omega_j}\lvert g(x)\rvert\\
      &=\sum_{j\in J}\lVert h_j\rVert_1\\
      &\le\lvert J\rvert r^{d-1}A_d(g)
      \le r^dA_d(g).
    \end{aligned}
  \end{equation}

  If $x^\star\in\Omega$, it belongs to none of the sets $\Omega_j$, so its contribution must be bounded separately. The same construction covers
  $\Omega\setminus\{x^\star\}$ with at most $r-1$ sets $\Omega_j$. Writing
  \begin{equation}
    R:=\sum_{x\in\Omega\setminus\{x^\star\}}\lvert g(x)\rvert,
  \end{equation}
  we obtain $R\le(r-1)r^{d-1}A_d(g)$. The constant coefficient also gives
  \begin{equation}
    g(x^\star)
    =\widehat g(\varnothing)
    -\sum_{x\in\Omega\setminus\{x^\star\}}g(x),
  \end{equation}
  so $\lvert g(x^\star)\rvert\le A_d(g)+R$. Consequently,
  \begin{equation}
    \lVert g\rVert_1
    \le A_d(g)+2R
    \le 2r^dA_d(g).
  \end{equation}
  Both cases give $\lVert g\rVert_1\le2r^dA_d(g)$, proving Eq.~\eqref{eq:sparse-fourier-auxiliary} at support size $r$ and completing the induction.

  A set $W$ attaining $A_d(g)$ satisfies $\lvert W\rvert\le d\le\log_2r$, and Eq.~\eqref{eq:sparse-fourier-auxiliary} gives
  \begin{equation}
    \lvert\widehat g(W)\rvert=A_d(g)
    \ge\frac{\lVert g\rVert_1}{2r^d}
    \ge\frac{\lVert g\rVert_1}{2r^{\log_2r}}.
  \end{equation}
  Finally, $r\le s$ gives $r^{\log_2r} \le s^{\log_2s}$, and it concludes the proof.
\end{proof}

We now bound the size of the low-contrast set $L$.

\begin{lemma}[Low-contrast-qubit bound]
  \label{lem:low-contrast-qubit-bound}
  Let $p$ be a probability mass function on $\{0,1\}^n$ with support size at
  most $s$.
  The set of low-contrast qubits satisfies
  \begin{equation}
    \lvert L\rvert\le2\log_2s.
    \label{eq:low-contrast-qubits-bound}
  \end{equation}
\end{lemma}

\begin{proof}
  Write $H$ and $h_2$ for Shannon entropy and binary entropy. Let $X=(X_1,\ldots,X_n)\in \{0,1\}^n$ be a random variable drawn from distribution $p$, and let $X_{-i}:=(X_j)_{j\in[n]\setminus\{i\}}$. We have
  \begin{equation}
    \lVert D_ip\rVert_1
    = \sum_{x_{-i}\in\{0,1\}^{n-1}}
    \mathbb P[X_{-i}=x_{-i}]
    \left\lvert
    \mathbb P[X_i=0\mid X_{-i}=x_{-i}]
    -\mathbb P[X_i=1\mid X_{-i}=x_{-i}]
    \right\rvert.
    \label{eq:low-contrast-qubits-conditional-bias}
  \end{equation}
  Note that $h_2(t)\ge1-\lvert2t-1\rvert$ for all $t\in[0,1]$. It then implies
  \begin{equation}
    H(X_i\mid X_{-i})\ge1-\lVert D_ip\rVert_1.
    \label{eq:low-contrast-qubits-entropy-charge}
  \end{equation}

  It remains to control the sum of these conditional entropies.  Write
  $X_{<j}=(X_1,\ldots,X_{j-1})$.  The chain rule and monotonicity under adding
  conditioning variables give, for each $i$,
  \begin{equation}
    \begin{aligned}
      H(X_{-i})
       & =\sum_{j\ne i}
      H\!\left(X_j\mid(X_k)_{k<j,\,k\ne i}\right) \\
       & \ge\sum_{j\ne i}H(X_j\mid X_{<j}).
    \end{aligned}
  \end{equation}
  Summing over $i$ on both sides results in
  \begin{equation}
    \sum_{i=1}^nH(X_{-i})\ge(n-1)H(X).
    \label{eq:low-contrast-qubits-chain-rule}
  \end{equation}
  Because $H(X_i\mid X_{-i})=H(X)-H(X_{-i})$, Eq.~\eqref{eq:low-contrast-qubits-chain-rule}
  implies
  \begin{equation}
    \sum_{i=1}^nH(X_i\mid X_{-i})
    \le H(X)\le\log_2s,
  \end{equation}
  where the last inequality uses that the support of $p$ is at most $s$. Each $i\in L$ has $\lVert D_ip\rVert_1\le1/2$ and therefore contributes at least $1/2$ by Eq.~\eqref{eq:low-contrast-qubits-entropy-charge}. Thus $\lvert L\rvert/2\le\log_2s$, and it concludes the proof.
\end{proof}

Now, we are ready to show the following:

\begin{lemma}[Fourier detection for non-low-contrast qubits]
  \label{lem:fourier-detection}
  If $i\notin L$, there exists $W\subseteq[n]$ with $i\in W$ such that
  \begin{equation}
    \lvert W\rvert\le\ell_0,
    \qquad
    \lvert\widehat p(W)\rvert>
    \frac1{4s^{\log_2s}}.
    \label{eq:fourier-detection-detection}
  \end{equation}
\end{lemma}

\begin{proof}
  Fix $i\notin L$. Define the slice-difference function $g_i:\{0,1\}^{n-1}\to\mathbb R$ on the remaining bits by
  \begin{equation}
    g_i(x_{-i}):=p(x^{(i\mapsto0)})-p(x^{(i\mapsto1)})
  \end{equation}
  for all $x \in \{0,1\}^n$. Note that $g_i(x_{-i})=(2D_ip)(x)$. The support of $g_i$ is contained in the projection of $\operatorname{supp}(p)$ onto the remaining bits, so $\lvert\operatorname{supp}(g_i)\rvert\le s$. Since $D_ip(x)=g_i(x_{-i})/2$, summing over the two values of $x_i$ gives
  \begin{equation}
    \lVert D_ip\rVert_1
    =\sum_{x_{-i}}\sum_{x_i\in\{0,1\}}
    \frac{\lvert g_i(x_{-i})\rvert}{2}
    =\lVert g_i\rVert_1>\frac12.
  \end{equation}

  For $S\subseteq[n]\setminus\{i\}$, expanding the unnormalized Fourier coefficient gives the identity
  \begin{equation}
    \widehat g_i(S) =\widehat p(S\cup\{i\}).
  \end{equation}
  For $n\ge2$, apply Lemma~\ref{lem:sparse-low-degree-fourier-coefficient} to $g_i$ to obtain $S\subseteq[n]\setminus\{i\}$ with $\lvert S\rvert\le\log_2s$ and
  \begin{equation}
    \begin{aligned}
      \left\lvert\widehat p(S\cup\{i\})\right\rvert
      &=\lvert\widehat g_i(S)\rvert\\
      &\ge\frac{\lVert g_i\rVert_1}{2s^{\log_2s}}
      >\frac1{4s^{\log_2s}}.
    \end{aligned}
  \end{equation}
  If $n=1$, the same conclusion holds with $S=\varnothing$, since $g_i$ is a function on the single empty bitstring and $\lvert\widehat g_i(\varnothing)\rvert=\lVert g_i\rVert_1>1/2$.
  Set $W=S\cup\{i\}=\{i\}$. Then Eq.~\eqref{eq:fourier-detection-detection} trivially holds.
\end{proof}

\subsection{Factorization of Pauli moment tensors and the difficult set}
\label{sec:exact-basis-recovery}

Write $B_i=\{\lvert\phi_0^{(i)}\rangle,\lvert\phi_1^{(i)}\rangle\}$, $\mathbf n_i=\mathbf n_{B_i}$, and $\lvert\phi_x\rangle=\bigotimes_i\lvert\phi_{x_i}^{(i)}\rangle$ for $x\in\{0,1\}^n$. The diagonal part of $\rho$ is denoted by
\begin{equation}
  \rho_{\mathrm{diag}}
  :=\sum_xp(x)\lvert\phi_x\rangle\!\langle\phi_x\rvert,
\end{equation}
and we define the diagonal Pauli moment tensor $T_W^{\mathrm{diag}}$ by replacing $\rho$ with $\rho_{\mathrm{diag}}$ in Eq.~\eqref{eq:product-pauli-tensor}. Then, we have the following factorization of $T_W^{\mathrm{diag}}$.

\begin{lemma}[Diagonal tensor factorization]
  \label{lem:diagonal-tensor-factorization}
  For every nonempty $W\subseteq[n]$,
  \begin{equation}
    T_W^{\mathrm{diag}}
    =\widehat p(W)\bigotimes_{i\in W}\mathbf n_i.
    \label{eq:diagonal-tensor-factorization}
  \end{equation}
\end{lemma}

\begin{proof}
  Since $\mathbf n_i$ is the Bloch vector of $\lvert\phi_0^{(i)}\rangle$ in our convention, the orthogonal state $\lvert\phi_1^{(i)}\rangle$ has Bloch vector $-\mathbf n_i$. Thus, with $\boldsymbol\sigma=(X,Y,Z)$, we have
  \begin{equation}
    \langle\phi_{x_i}^{(i)}\rvert\boldsymbol\sigma
    \lvert\phi_{x_i}^{(i)}\rangle
    =(-1)^{x_i}\mathbf n_i.
  \end{equation}
  Therefore, we have
  \begin{equation}
    \begin{aligned}
      T_W^{\mathrm{diag}}
       & =\sum_xp(x) \bigotimes_{i\in W} \langle\phi_{x_i}^{(i)}\rvert\boldsymbol\sigma \lvert\phi_{x_i}^{(i)}\rangle\\
       & =\sum_xp(x)\prod_{i\in W}(-1)^{x_i} \bigotimes_{i\in W}\mathbf n_i\\
       & = \widehat p(W)\bigotimes_{i\in W}\mathbf n_i.
    \end{aligned}
  \end{equation}
\end{proof}

If $\widehat p(W)\ne0$, $T_W^{\mathrm{diag}}$ is a rank-one tensor. For any qubit $i\in W$, we can reshape $T_W^{\mathrm{diag}}$ as a matrix with row index $\alpha_i$ and column index $(\alpha_j)_{j\in W\setminus\{i\}}$ and perform singular value decomposition (SVD). Since the local Bloch vectors have unit norm, its leading singular value is $\lvert\widehat p(W)\rvert$ and its leading left singular vector is $\mathbf n_i$, up to sign.

The tensors recorded in Step~1, however, are formed from the full state $\rho$ and may retain off-diagonal entries. Now, we explain how the contributions of these off-diagonal entries are removed by tensor contraction as follows. For an off-diagonal entry $(x,y)$, let $F(x,y):=\{i\in[n]:x_i\ne y_i\}$. For $W,K\subseteq[n]$, we denote
\begin{equation}
  \Phi_{W,K}(T_W)
  :=T_W\bigtimes_{j\in K\cap W}\mathbf n_j^{T}.
\end{equation}

\begin{lemma}[Cancellation of off-diagonal contributions]
  \label{lem:off-diagonal-cancellation}
  There exists $J_\star\subseteq[n]$ with $\lvert J_\star\rvert\le s$ such that, for every nonempty $W\subseteq[n]$,
  \begin{equation}
    \Phi_{W,J_\star}(T_W)
    =\widehat p(W)
    \bigotimes_{i\in W\setminus J_\star}\mathbf n_i.
    \label{eq:off-diagonal-cancellation-filtering}
  \end{equation}
\end{lemma}

\begin{proof}
  Let $\mathcal F:=\{F(x,y):[\rho]_B(x,y)\ne0,\ x\ne y\}$ be the set of distinct nonempty flip sets arising from nonzero off-diagonal entries. For each $F\in\mathcal F$, choose one representative qubit $j_F\in F$ and define
  \begin{equation}
    J_\star:=\{j_F:F\in\mathcal F\}.
  \end{equation}
  Each distinct flip set arises from one of the at most $s$ nonzero matrix entries, so $\lvert\mathcal F\rvert\le s$. Different flip sets may have the same representative; hence $\lvert J_\star\rvert\le\lvert\mathcal F\rvert\le s$.

  Fix a nonempty $W\subseteq[n]$ and a nonzero off-diagonal entry $[\rho]_B(x,y)$. If $F(x,y)\not\subseteq W$, some qubit $k\notin W$ has $x_k\ne y_k$, so the contribution to $T_W$ vanishes through the factor $\langle\phi_{y_k}^{(k)}\rvert I\lvert\phi_{x_k}^{(k)}\rangle=0$. Otherwise, the chosen representative $j:=j_{F(x,y)}$ belongs to $J_\star\cap W$ and satisfies $x_j\ne y_j$. Contracting the tensor against $\mathbf n_j$ replaces its local Pauli matrix element by
  \begin{equation}
    \begin{aligned}
      &\sum_{\alpha_j\in\{X,Y,Z\}}
      [\mathbf n_j]_{\alpha_j}
      \langle\phi_{y_j}^{(j)}\rvert\alpha_j
      \lvert\phi_{x_j}^{(j)}\rangle\\
      &\qquad=\langle\phi_{y_j}^{(j)}\rvert
      \mathbf n_j\cdot\boldsymbol\sigma
      \lvert\phi_{x_j}^{(j)}\rangle=0,
    \end{aligned}
    \label{eq:off-diagonal-cancellation-annihilation}
  \end{equation}
  because $\mathbf n_j\cdot\boldsymbol\sigma$ is diagonal in $B_j$. Thus contraction on $J_\star\cap W$ removes every off-diagonal contribution.

  By linearity, only the diagonal tensor remains. Applying Lemma~\ref{lem:diagonal-tensor-factorization} and using $\mathbf n_j^T\mathbf n_j=1$ gives
  \begin{equation}
    \begin{aligned}
      \Phi_{W,J_\star}(T_W)
      &=\Phi_{W,J_\star}(T_W^{\mathrm{diag}})\\
      &=\widehat p(W)
      \left(\prod_{j\in J_\star\cap W}\mathbf n_j^T\mathbf n_j\right)
      \bigotimes_{i\in W\setminus J_\star}\mathbf n_i\\
      &=\widehat p(W)
      \bigotimes_{i\in W\setminus J_\star}\mathbf n_i.
    \end{aligned}
  \end{equation}
\end{proof}

Define the difficult set $K_\star:=J_\star\cup L$ and call its members difficult qubits. Lemmas~\ref{lem:low-contrast-qubit-bound} and~\ref{lem:off-diagonal-cancellation} give
\begin{equation}
  \lvert K_\star\rvert
  \le\lvert J_\star\rvert+\lvert L\rvert
  \le s+\lceil2\log_2s\rceil.
  \label{eq:difficult-set-size}
\end{equation}
For every nonempty $W\subseteq[n]$ with $\lvert W\rvert\le\ell_0$, further contracting Eq.~\eqref{eq:off-diagonal-cancellation-filtering} against $\mathbf n_j$ for $j\in(L\setminus J_\star)\cap W$ and using $\mathbf n_j^T\mathbf n_j=1$ yields
\begin{equation}
  \Phi_{W,K_\star}(T_W)
  =\widehat p(W)
  \bigotimes_{i\in W\setminus K_\star}\mathbf n_i.
  \label{eq:difficult-set-contracted-signal}
\end{equation}

For every $i\notin K_\star$, Lemma~\ref{lem:fourier-detection} ensures the existence of a set $W\subseteq[n]$ containing $i$ with $\lvert W\rvert\le\ell_0$ and $\lvert\widehat p(W)\rvert>1/(4s^{\log_2s})$. Again, reshaping the contracted tensor in Eq.~\eqref{eq:difficult-set-contracted-signal} as a matrix with row index $\alpha_i$ and column index $(\alpha_j)_{j\in W\setminus(K_\star\cup\{i\})}$ gives a rank-one matrix. Its leading singular value is $\lvert\widehat p(W)\rvert$ and its leading left singular vector is $\mathbf n_i$, up to sign. Note that every qubit not in $K_\star$ has a large singular value due to Lemma~\ref{lem:fourier-detection}. The next section shows how we can robustly recover its basis vector via SVD under estimatation errors of Pauli moment tensors and approximate erors of basis vectors on the difficult qubits.

\subsection{Robust basis recovery}
\label{sec:robust-basis-recovery}

The exact factorization in Eq.~\eqref{eq:difficult-set-contracted-signal} must remain useful when both the Pauli moment tensors and the basis vectors of difficult qubits are approximate. Specifically, the algorithm has access to estimated tensors $\widehat T_W$ and approximate Bloch vectors $\widetilde{\mathbf n}_j$ for $j\in K_\star\cap W$ drawn from an $\varepsilon_2$-net. Throughout the following analysis, $\varepsilon_1$, $\varepsilon_2$, $\xi$, and $\theta$ retain their fixed values in Eq.~\eqref{eq:pauli-moment-learning-final-tolerances}, with $0<\varepsilon<1$, $0<\delta<1$, and positive integers $s,n$.

We first bound the number of copies needed to estimate the Pauli moment tensors. Let $\mathcal E_1$ denote the event that every estimated Pauli moment of order at most $\ell_0$ has error at most $\varepsilon_1$, and let $N_{\mathrm{obs}}$ count the Pauli observables of order at most $\ell_0$. Measuring each nonidentity Pauli observable on fresh copies of $\rho$ gives outcomes in $\{-1,1\}$, so Hoeffding's inequality and a union bound give $\mathbb P(\mathcal E_1^c)\le\delta/3$ using
\begin{equation}
  N_1 :=N_{\mathrm{obs}} \left\lceil \frac{2}{\varepsilon_1^2} \ln\frac{6N_{\mathrm{obs}}}{\delta} \right\rceil.
\label{eq:pauli-moment-learning-low-order-copies}
\end{equation}
copies of $\rho$. Therefore, with probability at least $1-\delta/3$, we have
\begin{equation}
  \lvert[\widehat T_W]_\alpha-[T_W]_\alpha\rvert\le\varepsilon_1
  \label{eq:contraction-stability-tensor-error}
\end{equation}
for every $\alpha\in\{X,Y,Z\}^W$ and every nonempty $W\subseteq[n]$ with $\lvert W\rvert\le\ell_0$.

For nonempty $W\subseteq[n]$ and $K\subseteq[n]$, write the estimated contraction as
\begin{equation}
  \widehat\Phi_{W,K}\!\left(
  \widehat T_W;\{\widetilde{\mathbf n}_j\}_{j\in K\cap W}
  \right)
  :=\widehat T_W
  \bigtimes_{j\in K\cap W}\widetilde{\mathbf n}_j^{T}.
\end{equation}
Only proposed Bloch vectors on $K\cap W$ enter this contraction; recovered vectors on $[n]\setminus K$ do not enter later contractions. The basis comparison uses angular distance, which ignores label swaps, whereas the vector differences in the perturbation bound require sign alignment.

\begin{lemma}[Stability of contracted tensor]
  \label{lem:contraction-stability}
  Let $\varnothing\ne W\subseteq[n]$ satisfy $\lvert W\rvert\le\ell_0$, and suppose
  \begin{equation}
    \lvert[\widehat T_W]_\alpha-[T_W]_\alpha\rvert\le\varepsilon_1
  \end{equation}
  for every $\alpha\in\{X,Y,Z\}^W$, with $\varepsilon_1$ as in Eq.~\eqref{eq:pauli-moment-learning-final-tolerances}. Assume that the proposed unit Bloch vectors satisfy $\lVert\widetilde{\mathbf n}_j-\mathbf n_j\rVert_2\le\varepsilon_2$ for every $j\in K_\star\cap W$, with $\varepsilon_2$ as in Eq.~\eqref{eq:pauli-moment-learning-final-tolerances}. Then
  \begin{equation}
    \begin{aligned}
      \biggl\lVert
      \widehat\Phi_{W,K_\star}\!\left(
      \widehat T_W;\{\widetilde{\mathbf n}_j\}_{j\in K_\star\cap W}
      \right)
      -\Phi_{W,K_\star}(T_W)
      \biggr\rVert_F
      \le2\cdot3^{\ell_0/2}\varepsilon_1.
    \end{aligned}
    \label{eq:contraction-stability-bound}
  \end{equation}
\end{lemma}

\begin{proof}
  Write $K_\star\cap W=\{j_1,\ldots,j_m\}$, where $m\le\lvert W\rvert\le\ell_0$. We first show the inequality used for each contraction. For a tensor $A$, a vector $v\in\mathbb R^3$, and a mode $j$, let $\beta$ collect all tensor indices other than $\alpha_j$. Cauchy--Schwarz gives
  \begin{equation}
    \begin{aligned}
      \lVert A\bigtimes_j v^T\rVert_F^2
      &=\sum_\beta\left\lvert\sum_{\alpha_j}
        A_{\beta,\alpha_j}v_{\alpha_j}\right\rvert^2\\
      &\le\sum_\beta
        \left(\sum_{\alpha_j}\lvert A_{\beta,\alpha_j}\rvert^2\right)
        \left(\sum_{\alpha_j}\lvert v_{\alpha_j}\rvert^2\right)\\
      &=\lVert A\rVert_F^2\lVert v\rVert_2^2.
    \end{aligned}
    \label{eq:contraction-stability-contraction}
  \end{equation}
  The entrywise error assumption and $-1 \le \lvert[T_W]_\alpha \le 1$ for all Paulis $\alpha$ imply
  \begin{equation}
    \begin{aligned}
      \lVert\widehat T_W-T_W\rVert_F^2
      &=\sum_{\alpha\in\{X,Y,Z\}^W}
        \lvert[\widehat T_W]_\alpha-[T_W]_\alpha\rvert^2
        \le3^{\lvert W\rvert}\varepsilon_1^2,\\
      \lVert T_W\rVert_F^2
      &=\sum_{\alpha\in\{X,Y,Z\}^W}\lvert[T_W]_\alpha\rvert^2
        \le3^{\lvert W\rvert}.
    \end{aligned}
  \end{equation}
  Repeatedly applying Eq.~\eqref{eq:contraction-stability-contraction} gives the following bound:
  \begin{equation}
    \begin{aligned}
      \left\lVert(\widehat T_W-T_W)
        \bigtimes_{j\in K_\star\cap W}\widetilde{\mathbf n}_j^T\right\rVert_F
      &\le\lVert\widehat T_W-T_W\rVert_F
        \prod_{j\in K_\star\cap W}\lVert\widetilde{\mathbf n}_j\rVert_2\\
      &=\lVert\widehat T_W-T_W\rVert_F
        \le3^{\lvert W\rvert/2}\varepsilon_1.
    \end{aligned}
  \end{equation}

  To bound the effect of the proposed Bloch vectors, define the intermediate tensors
  \begin{equation}
    H_a:=T_W
      \bigtimes_{b=1}^{a}\widetilde{\mathbf n}_{j_b}^T
      \bigtimes_{b=a+1}^{m}\mathbf n_{j_b}^T,
    \qquad 0\le a\le m,
  \end{equation}
  where the factor indexed by $b$ contracts mode $j_b$. Thus $H_0=\Phi_{W,K_\star}(T_W)$, while $H_m$ contracts $T_W$ with $\{\widetilde{\mathbf n}_j\}_{j\in K_\star\cap W}$. By multilinearity, for $1\le a\le m$,
  \begin{equation}
    H_a-H_{a-1}
    =T_W\bigtimes_{b=1}^{a-1}\widetilde{\mathbf n}_{j_b}^T
      \bigtimes_{j_a}(\widetilde{\mathbf n}_{j_a}-\mathbf n_{j_a})^T
      \bigtimes_{b=a+1}^{m}\mathbf n_{j_b}^T.
  \end{equation}
  Applying Eq.~\eqref{eq:contraction-stability-contraction} to every factor yields
  \begin{equation}
    \begin{aligned}
      \lVert H_a-H_{a-1}\rVert_F
      &\le\lVert T_W\rVert_F
        \left(\prod_{b<a}\lVert\widetilde{\mathbf n}_{j_b}\rVert_2\right)
        \lVert\widetilde{\mathbf n}_{j_a}-\mathbf n_{j_a}\rVert_2
        \left(\prod_{b>a}\lVert\mathbf n_{j_b}\rVert_2\right)\\
      &=\lVert T_W\rVert_F
        \lVert\widetilde{\mathbf n}_{j_a}-\mathbf n_{j_a}\rVert_2
        \le3^{\lvert W\rvert/2}\varepsilon_2.
    \end{aligned}
  \end{equation}

  Note that the difference between the estimated contraction and the ideal contraction can be expressed as
  \begin{equation}
    \begin{aligned}
      &\widehat\Phi_{W,K_\star}\!\left(
        \widehat T_W;\{\widetilde{\mathbf n}_j\}_{j\in K_\star\cap W}\right)
        -\Phi_{W,K_\star}(T_W)\\
      &\qquad=(\widehat T_W-T_W)
        \bigtimes_{j\in K_\star\cap W}\widetilde{\mathbf n}_j^T
        +\sum_{a=1}^{m}(H_a-H_{a-1}).
    \end{aligned}
  \end{equation}
  The preceding bounds and the triangle inequality now give
  \begin{equation}
    \begin{aligned}
      &\left\lVert\widehat\Phi_{W,K_\star}\!\left(
        \widehat T_W;\{\widetilde{\mathbf n}_j\}_{j\in K_\star\cap W}\right)
        -\Phi_{W,K_\star}(T_W)\right\rVert_F\\
      &\qquad\le3^{\lvert W\rvert/2}\varepsilon_1
        +\sum_{a=1}^{m}3^{\lvert W\rvert/2}\varepsilon_2\\
      &\qquad=3^{\lvert W\rvert/2}(\varepsilon_1+m\varepsilon_2)\\
      &\qquad\le3^{\ell_0/2}(\varepsilon_1+\ell_0\varepsilon_2)
        =2\cdot3^{\ell_0/2}\varepsilon_1.
    \end{aligned}
  \end{equation}
  The last equality uses the fixed relation $\ell_0\varepsilon_2=\varepsilon_1$ from Eq.~\eqref{eq:pauli-moment-learning-final-tolerances}, concluding the proof.
\end{proof}

We now show that Step~1 of the algorithm outputs a list of bases containing a sufficiently accurate candidate product basis.

\begin{lemma}[The list of candidate bases contains an accurate basis]
  \label{lem:accurate-candidate-basis}
  Conditioned on the event $\mathcal E_1$, Step~1 of Algorithm~1 produces a list containing a candidate basis $B'=(B_1',\ldots,B_n')$, represented by $\{\widetilde{\mathbf n}_j\}_{j\in [n]}$, such that
  \begin{equation}
    d_{\mathrm{ang}}(B_j',B_j)\le \gamma
  \end{equation}
  for all $j\in [n]$, where
  \begin{equation}
    \gamma:=\frac{\varepsilon^2}{128s^2\sqrt n}.
    \label{eq:pauli-moment-learning-angular-target}
  \end{equation}
\end{lemma}

\begin{proof}
  The qubit-set enumeration includes $K_\star$ by Eq.~\eqref{eq:difficult-set-size}. The product epsilon-net includes a partial basis assignment whose local bases on the difficult qubits are within angular distance $\varepsilon_2$ of the corresponding local bases of $B$. Aligning the signs of their Bloch vectors gives, for each $j\in K_\star$,
  \begin{equation}
    \lVert\widetilde{\mathbf n}_j-\mathbf n_j\rVert_2
    =2\sin\!\left(
    \frac{d_{\mathrm{ang}}(B_j',B_j)}2
    \right)
    \le\varepsilon_2.
    \label{eq:accurate-candidate-proposed-vector-error}
  \end{equation}
  
  For each $W\subseteq[n]$ with $\lvert W\rvert\le\ell_0$ and each $i\in W\setminus K_\star$, define $M_0$ and $\widehat M$ as the mode-$i$ flattenings of the exact and estimated contracted tensors for this assignment. Specifically, their row index is $\alpha_i\in\{X,Y,Z\}$ and their column index is $\beta\in\{X,Y,Z\}^{W\setminus(K_\star\cup\{i\})}$, with entries
  \begin{equation}
    \begin{aligned}
      [M_0]_{\alpha_i,\beta}
      &:=\bigl[\Phi_{W,K_\star}(T_W)\bigr]_{\alpha_i,\beta},\\
      [\widehat M]_{\alpha_i,\beta}
      &:=\biggl[\widehat\Phi_{W,K_\star}\!\left(
        \widehat T_W;\{\widetilde{\mathbf n}_j\}_{j\in K_\star\cap W}
        \right)\biggr]_{\alpha_i,\beta}.
    \end{aligned}
  \end{equation}
  Thus $M_0,\widehat M\in\mathbb R^{3\times3^{\lvert W\setminus K_\star\rvert-1}}$. Flattening preserves the Frobenius norm, so Lemma~\ref{lem:contraction-stability} yields
  \begin{equation}
    \lVert\widehat M-M_0\rVert_{\mathrm{op}}
    \le\lVert\widehat M-M_0\rVert_F
    \le 2\cdot3^{\ell_0/2}\varepsilon_1 = \frac{\gamma}{16s^{\log_2s}}
    <\frac{1}{2048s^{\log_2s}}.
    \label{eq:accurate-candidate-basis-total-perturbation}
  \end{equation}
  Recall that, for each $i\notin K_\star$, Step~1(c) scans the sets $W\subseteq[n]$ with $i\in W$ and $\lvert W\rvert\le\ell_0$ in the fixed order. It retains the leading left singular vector of the first estimated flattening whose leading singular value is at least $3/(16s^{\log_2s})$, and discards the assignment if any qubit has no passing flattening. Lemma~\ref{lem:fourier-detection} and Eq.~\eqref{eq:difficult-set-contracted-signal} guarantee that, for each such $i$, some $W$ gives an ideal flattening $M_0$ with singular value at least $1/(4s^{\log_2s})$. By Weyl's inequality and Eq.~\eqref{eq:accurate-candidate-basis-total-perturbation}, its estimated leading singular value is at least
  \begin{equation}
    \frac1{4s^{\log_2s}}-\frac1{2048s^{\log_2s}}
    =\frac{511}{2048s^{\log_2s}}
    >\frac3{16s^{\log_2s}},
  \end{equation}
  so every remaining qubit has a passing flattening and Step~1(c) produces a candidate basis from this assignment.

  For the flattening selected by the scan, let $\widehat\lambda\ge3/(16s^{\log_2s})$ be its leading singular value. The retained real unit left singular vector $\widetilde{\mathbf n}_i$ and a corresponding unit right singular vector $\widehat v$ satisfy $\widehat\lambda\widetilde{\mathbf n}_i=\widehat M\widehat v$. Equation~\eqref{eq:difficult-set-contracted-signal} gives $(I-\mathbf n_i\mathbf n_i^T)M_0=0$, so
  \begin{equation}
    \widehat\lambda(I-\mathbf n_i\mathbf n_i^T)\widetilde{\mathbf n}_i
    =(I-\mathbf n_i\mathbf n_i^T)(\widehat M-M_0)\widehat v.
  \end{equation}
  Taking norms and using $\lVert\widehat v\rVert_2=1$ yields
  \begin{equation}
    \begin{aligned}
      \sqrt{1-\lvert\widetilde{\mathbf n}_i^T\mathbf n_i\rvert^2}
      &=\lVert(I-\mathbf n_i\mathbf n_i^T)\widetilde{\mathbf n}_i\rVert_2\\
      &\le\frac{\lVert\widehat M-M_0\rVert_{\mathrm{op}}}
      {\widehat\lambda}.
    \end{aligned}
    \label{eq:accurate-candidate-basis-singular-vector}
  \end{equation}
  The recovered local basis $B_i'$ is the eigenbasis of $\widetilde{\mathbf n}_i\cdot\boldsymbol{\sigma}$. Equations
  \eqref{eq:accurate-candidate-basis-total-perturbation} and \eqref{eq:accurate-candidate-basis-singular-vector} imply
  \begin{equation}
    \sin d_{\mathrm{ang}}(B_i',B_i)
    \le
    \frac{\gamma/(16s^{\log_2s})}
    {3/(16s^{\log_2s})}
    =\frac{\gamma}{3}.
  \end{equation}
  Since this angle lies in $[0,\pi/2]$, $d_{\mathrm{ang}}(B_i',B_i)< \gamma$.
\end{proof}

We can show that the accurate candidate basis $B'$ from Lemma~\ref{lem:accurate-candidate-basis} is sufficiently close to the fixed product basis $B$ so that $\rho$ is approximately sparse in $B'$.

\begin{lemma}[Stability under product-basis errors]
  \label{lem:product-basis-stability}
  Let $B'=(B_1',\ldots,B_n')$ be the accurate candidate basis supplied by Lemma~\ref{lem:accurate-candidate-basis} where $\varphi_i:=d_{\mathrm{ang}}(B_i,B_i')\le\gamma$. Then there exists a state $\sigma$ that is $s$-sparse in $B'$ and satisfies
  \begin{equation}
    \lVert\sigma-\rho\rVert_1
    \le\sqrt{s\sum_{i=1}^n\varphi_i^2}
    \le\sqrt{sn}\,\gamma.
    \label{eq:product-basis-stability-bound}
  \end{equation}
\end{lemma}

\begin{proof}
  For each $i\in[n]$, write $\lvert\phi_b^{(i)}\rangle$ and $\lvert\phi_b^{\prime(i)}\rangle$, $b\in\{0,1\}$, for the basis vectors of $B_i$ and $B_i'$. Relabel the vectors in $B_i'$, if necessary, so that $\mathbf n_i\cdot\widetilde{\mathbf n}_i=\cos\varphi_i$. Then we have
  \begin{equation}
    \left\lvert\langle\phi_b^{(i)}\mid\phi_b^{\prime(i)}\rangle\right\rvert^2
    =\frac{1+\mathbf n_i\cdot\widetilde{\mathbf n}_i}{2}
    =\frac{1+\cos\varphi_i}{2}
    =\cos^2\frac{\varphi_i}{2}.
  \end{equation}
  We choose the phase of each $\lvert\phi_b^{\prime(i)}\rangle$ so that its inner product with $\lvert\phi_b^{(i)}\rangle$ is real and nonnegative. With this choice,
  \begin{equation}
    \langle\phi_b^{(i)}\mid\phi_b^{\prime(i)}\rangle
    =\cos\frac{\varphi_i}{2},
    \qquad b\in\{0,1\},\quad i\in[n].
  \end{equation}
  Define the single-qubit unitaries by $V_i\lvert b\rangle=\lvert\phi_b^{(i)}\rangle$ and $V_i'\lvert b\rangle=\lvert\phi_b^{\prime(i)}\rangle$, and set $V:=\bigotimes_{i=1}^nV_i$ and $V':=\bigotimes_{i=1}^nV_i'$. Write $C:=[\rho]_B$ and define
  \begin{equation}
    \sigma:=V'CV'^\dagger.
  \end{equation}
  Note that $\rho=VCV^\dagger$. Since $C\succeq0$, $\operatorname{Tr}(C)=1$, and $[\sigma]_{B'}=C$, the matrix $\sigma$ is a state that is $s$-sparse in $B'$.

  Let $R:=\{x\in\{0,1\}^n:C_{xx}>0\}$ and $\Pi:=\sum_{x\in R}\lvert x\rangle\!\langle x\rvert$. Then $\lvert R\rvert\le s$. The product structure of $V$ and $V'$ yields
  \begin{equation}
    \begin{aligned}
      \lVert(V'-V)\Pi\rVert_F^2
      &=\sum_{x\in R}\lVert(V'-V)\lvert x\rangle\rVert_2^2\\
      &=2\sum_{x\in R}\left(
        1-\operatorname{Re}\langle x\rvert V^\dagger V'\lvert x\rangle
        \right)\\
      &=2\lvert R\rvert\left(1-\prod_{i=1}^n\cos\frac{\varphi_i}{2}\right).
    \end{aligned}
  \end{equation}
  Consequently,
  \begin{equation}
    \lVert(V'-V)\Pi\rVert_{\mathrm{op}}
    \le\lVert(V'-V)\Pi\rVert_F
    \le
    \sqrt{2s\left(
      1-\prod_{i=1}^n\cos\frac{\varphi_i}{2}
      \right)}.
    \label{eq:product-basis-stability-restricted-unitary}
  \end{equation}
  Since $0\le\cos(\varphi_i/2)\le1$, we have
  \begin{equation}
    \begin{aligned}
      1-\prod_{i=1}^n\cos\frac{\varphi_i}{2}
      &= \sum_{i=1}^n \left(1-\cos\frac{\varphi_i}{2}\right)\prod_{j<i}\cos\frac{\varphi_j}{2}\\
      &\le\sum_{i=1}^n\left(1-\cos\frac{\varphi_i}{2}\right)\\
      &=2\sum_{i=1}^n\sin^2\frac{\varphi_i}{4}
      \le\frac18\sum_{i=1}^n\varphi_i^2.
    \end{aligned}
    \label{eq:product-basis-stability-cosine-product}
  \end{equation}
  Finally, $C=\Pi C\Pi$ and $\lVert C\rVert_1=1$ give
  \begin{equation}
    \begin{aligned}
      \lVert\sigma-\rho\rVert_1
      &\le\lVert(V'-V)CV'^\dagger\rVert_1
        +\lVert VC(V'^\dagger-V^\dagger)\rVert_1\\
      &\le2\lVert(V'-V)\Pi\rVert_{\mathrm{op}}\lVert C\rVert_1\\
      &\le\sqrt{s\sum_{i=1}^n\varphi_i^2}
        \le\sqrt{sn}\,\gamma.
    \end{aligned}
  \end{equation}
\end{proof}

\subsection{Completeness and soundness of reconstruction}
\label{sec:reconstruction-certification}

Conditioning on $\mathcal E_1$, Step~1 produces a list of candidate bases containing at least one accurate candidate basis. Then, the subsequent steps of Algorithm~1 test each candidate basis using fresh copies of $\rho$. Step~2 tests if there exists $S\subseteq\{0,1\}^n$ of at most $s$ bitstrings that captures nearly all of the probability of measuring $\rho$ in each candidate basis. The corresponding $S\times S$ submatrix can still contain $s^2$ entries, so Step~3 estimates this submatrix and retains only entries of magnitude at least $\theta$. The candidate basis is accepted only if at most $s$ matrix entries remain. Here, we prove completeness and soundness on the stated estimation events: the accurate candidate basis $B'$ from Section~\ref{sec:robust-basis-recovery} passes both tests, and every accepted matrix $M$ satisfies $\lVert\rho-M\rVert_1\le\varepsilon$. All lemmas in this section assume that the event $\mathcal E_1$ holds.

\begin{lemma}[Completeness for the accurate candidate basis]
  \label{lem:reconstruction-completeness}
  Let $B'$ be the accurate candidate basis supplied by Lemma~\ref{lem:accurate-candidate-basis}, and $q$ and $\widehat q$ be the true and empirical probability distributions of $\rho$ in the basis $B'$. With $\xi$ and $\theta$ as in Eq.~\eqref{eq:pauli-moment-learning-final-tolerances}, suppose Step~2 satisfies
  \begin{equation}
    \lvert\widehat q(T)-q(T)\rvert\le\xi/2
    \qquad\text{for all }T\subseteq\{0,1\}^n
    \text{ with }\lvert T\rvert\le s.
    \label{eq:reconstruction-completeness-uniform-event}
  \end{equation}
  Also assume that the subsequent entrywise estimation of the $S\times S$ submatrix of $\rho$ in Step~3 satisfies
  \begin{equation}
    \left\lvert
    [\widetilde M]_{xy}-[\rho]_{B'}(x,y)
    \right\rvert
    \le \theta/2 \qquad\text{for all }x,y\in S,
    \label{eq:reconstruction-completeness-entrywise-event}
  \end{equation}
  Then the accurate candidate basis $B'$ is accepted by Steps~2 and~3, and its accepted matrix has at most $s$ nonzero entries in the reported basis.
\end{lemma}

\begin{proof}
  Lemma~\ref{lem:product-basis-stability} supplies a state $\sigma$ with at most $s$ nonzero matrix entries in $B'$ and
  \begin{equation}
    \lVert\rho-\sigma\rVert_1\le b,
    \qquad
    b:=\sqrt{sn}\,\gamma
    =\frac{\varepsilon^2}{128s^{3/2}}.
    \label{eq:reconstruction-completeness-accurate-comparator}
  \end{equation}
  Note that $b<\xi/2$ and $b<\theta/2$. Indeed,
  \begin{equation}
    \frac{b}{\xi/2}=\frac{1}{2s^{3/2}}\le\frac12<1,
    \qquad
    \frac{b}{\theta/2}=\frac{3\varepsilon}{64}<\frac3{64}<1.
    \label{eq:reconstruction-completeness-completeness-margins}
  \end{equation}
  
  To prove that $B'$ passes Step~2, consider the diagonal support of $\sigma$ in this basis:
  \begin{equation}
    R:=\{x:[\sigma]_{B'}(x,x)>0\}.
  \end{equation}
  Since $\sigma$ has at most $s$ nonzero matrix entries, $\lvert R\rvert\le s$. Its measurement distribution vanishes outside $R$, whereas the distribution $q$ of $\rho$ has only a small probability there. Therefore, we have
  \begin{equation}
    \begin{aligned}
      q(R^c)
      &=\sum_{x\notin R}\left\lvert[\rho]_{B'}(x,x)-[\sigma]_{B'}(x,x)\right\rvert\\
      &\le\sum_x\left\lvert[\rho]_{B'}(x,x)-[\sigma]_{B'}(x,x)\right\rvert\\
      &\le\lVert\rho-\sigma\rVert_1\le b.
    \end{aligned}
    \label{eq:reconstruction-completeness-comparator-tail}
  \end{equation}
  Eq.~\eqref{eq:reconstruction-completeness-uniform-event} applies to $R$ as $\lvert R\rvert\le s$. Since $q$ and $\widehat q$ are probability distributions, it also bounds the error on the complement:
  \begin{equation}
    \left\lvert\widehat q(R^c)-q(R^c)\right\rvert
    =\left\lvert\widehat q(R)-q(R)\right\rvert
    \le\xi/2.
  \end{equation}
  Step~2 chooses $S$ to maximize empirical probability over all sets of at most $s$ bitstrings. The set $R$ is one such choice, so $\widehat q(S)\ge\widehat q(R)$ and hence
  \begin{equation}
    \begin{aligned}
      \widehat q(S^c)
      &\le \widehat q(R^c)\\
      &\le q(R^c)+\xi/2\\
      &\le b+\xi/2<\xi.
    \end{aligned}
    \label{eq:reconstruction-completeness-capture-test}
  \end{equation}
  Thus the accurate candidate basis $B'$ passes the test in Step~2.

  For Step~3, we show that every retained matrix entry corresponds to a nonzero entry of $\sigma$ in $B'$. Let $x,y\in S$ satisfy $[\sigma]_{B'}(x,y)=0$. Then
  \begin{equation}
    \begin{aligned}
      \lvert[\rho]_{B'}(x,y)\rvert
      &=\lvert[\rho-\sigma]_{B'}(x,y)\rvert\\
      &\le\lVert\rho-\sigma\rVert_{\mathrm{op}}
      \le\lVert\rho-\sigma\rVert_1
      \le b.
    \end{aligned}
  \end{equation}
  The triangle inequality gives
  \begin{equation}
    \begin{aligned}
      \lvert[\widetilde M]_{xy}\rvert
      &\le\left\lvert[\widetilde M]_{xy}-[\rho]_{B'}(x,y)\right\rvert
        +\lvert[\rho]_{B'}(x,y)\rvert\\
      &\le\theta/2+b<\theta,
    \end{aligned}
  \end{equation}
  Step~3 therefore discards every entry whose position is outside the support of $[\sigma]_{B'}$. It follows that
  \begin{equation}
    \begin{aligned}
      K_{\mathrm{ent}}
      &\subseteq\{(x,y)\in S\times S:[\sigma]_{B'}(x,y)\ne0\},\\
      \lvert K_{\mathrm{ent}}\rvert
      &\le\lvert\operatorname{supp}([\sigma]_{B'})\rvert
      \le s.
    \end{aligned}
  \end{equation}
  Thus $B'$ also passes the entry-count test in Step~3, proving completeness.
\end{proof}

Completeness ensures that at least one candidate basis passes both tests on the stated estimation events. We now prove soundness for every accepted output matrix of Algorithm~1, showing that its trace-norm error relative to the true state $\rho$ is at most $\varepsilon$.

\begin{lemma}[Soundness of accepted reconstructions]
  \label{lem:reconstruction-soundness}
  Fix any candidate basis $\widetilde B$ that passes both tests of Algorithm~1, with $\xi$ and $\theta$ as in Eq.~\eqref{eq:pauli-moment-learning-final-tolerances}. Suppose
  \begin{equation}
    \lvert\widehat q(T)-q(T)\rvert\le\xi/2
    \qquad\text{for all }T\subseteq\{0,1\}^n
    \text{ with }\lvert T\rvert\le s
  \end{equation}
  and
  \begin{equation}
    \left\lvert
    [\widetilde M]_{xy}-[\rho]_{\widetilde B}(x,y)
    \right\rvert
    \le \theta/2 \qquad \text{for all $x,y\in S$.}
  \end{equation}
  Then its accepted matrix $M$ satisfies
  \begin{equation}
    \lVert\rho-M\rVert_1
    \le \varepsilon.
    \label{eq:reconstruction-soundness-certificate}
  \end{equation}
\end{lemma}

\begin{proof}
  Write $t:=\lvert S\rvert\le s$, and let $\Pi:=\Pi_S$ be the orthogonal projector onto the product-basis vectors of $\widetilde B$ indexed by $S$. Acceptance of the test in Step~2 gives
  \begin{equation}
    \begin{aligned}
      \eta
       & :=\operatorname{Tr}[(I-\Pi)\rho]
      =q(S^c)                             \\
       & \le\widehat q(S^c)+\xi/2
      \le\frac32\xi.
    \end{aligned}
    \label{eq:reconstruction-soundness-population-tail}
  \end{equation}
  Applying Lemma~\ref{lem:projection-error-trace} gives
  \begin{equation}
    \lVert\rho-\Pi_S\rho\Pi_S\rVert_1
    \le2\sqrt\eta
    \le2\sqrt{3\xi/2},
    \label{eq:reconstruction-soundness-projection-bound}
  \end{equation}
  which bounds the error from projecting onto the selected subspace.

  Next, we set
  \begin{equation}
    D:=\widetilde M-\Pi_S\rho\Pi_S.
  \end{equation}
  The matrix $D$ acts on a $t$-dimensional subspace and has $t^2$ entries of
  magnitude at most $\theta/2$. Thus
  \begin{equation}
    \lVert D\rVert_F\le \frac{t\theta}{2},
    \qquad
    \operatorname{rank}(D)\le t.
  \end{equation}
  Cauchy--Schwarz then gives
  \begin{equation}
    \lVert\widetilde M-\Pi_S\rho\Pi_S\rVert_1
    \le\sqrt t\,\lVert D\rVert_F
    \le \frac{t^{3/2}\theta}{2}
    \le \frac{s^{3/2}\theta}{2}.
    \label{eq:reconstruction-soundness-tomography-bound}
  \end{equation}

  Recall that $M$ retains only the entries of $\widetilde M$ whose magnitudes are at least $\theta$. Therefore,
  \begin{equation}
    \lVert \widetilde M-M\rVert_F\le t\theta,
    \qquad
    \operatorname{rank}(\widetilde M-M)\le t,
  \end{equation}
  and it gives
  \begin{equation}
    \lVert\widetilde M-M\rVert_1
    \le\sqrt t\,\lVert \widetilde M-M\rVert_F
    \le t^{3/2}\theta
    \le s^{3/2}\theta.
    \label{eq:reconstruction-soundness-thresholding-bound}
  \end{equation}
  
  Combining Eqs.~\eqref{eq:reconstruction-soundness-projection-bound},
  \eqref{eq:reconstruction-soundness-tomography-bound}, and~\eqref{eq:reconstruction-soundness-thresholding-bound}
  by the triangle inequality gives
  \begin{equation}
    \begin{aligned}
      \lVert\rho-M\rVert_1
      &\le\lVert\rho-\Pi_S\rho\Pi_S\rVert_1
        +\lVert\Pi_S\rho\Pi_S-\widetilde M\rVert_1
        +\lVert\widetilde M-M\rVert_1\\
      &\le2\sqrt{3\xi/2}
        +\frac12s^{3/2}\theta+s^{3/2}\theta\\
      &=\left(\frac{\sqrt3}{4}+\frac12\right)\varepsilon
        <\varepsilon.
    \end{aligned}
  \end{equation}
\end{proof}

\paragraph{Resource bounds.}
The preceding lemmas establish completeness and soundness whenever the required estimation bounds hold:
\begin{equation}
  \lvert\widehat q(T)-q(T)\rvert
  \le\xi/2 \qquad \text{for all }T\subseteq\{0,1\}^n
      \text{ with }\lvert T\rvert\le s
  \label{eq:reconstruction-sampling-capture-event}
\end{equation}
and
\begin{equation}
  \left\lvert
    [\widetilde M]_{xy}-[\rho]_{\widetilde B}(x,y)
    \right\rvert
  \le\theta/2
  \qquad \text{for all }x,y\in S.
  \label{eq:reconstruction-sampling-estimation-error}
\end{equation}
We now bound the number of copies of $\rho$ needed to ensure these bounds simultaneously across the candidate bases. Throughout, $\xi$ and $\theta$ retain their fixed values in Eq.~\eqref{eq:pauli-moment-learning-final-tolerances}. Write $G$ for the number of partial basis assignments $(K,(\widetilde{\mathbf n}_j)_{j\in K})$, which bounds the size of the candidate basis list because each assignment produces at most one candidate basis.

Let $\mathcal E_2$ denote the event in Eq.~\eqref{eq:reconstruction-sampling-capture-event} simultaneously for every candidate basis. For each candidate basis, let $x^{(1)},\ldots,x^{(N_2)}\in\{0,1\}^n$ be the outcomes from measuring $N_2$ fresh copies of $\rho$ in that basis, with ideal distribution $q$. Its empirical distribution is given by
\begin{equation}
  \widehat q(x)
  :=\frac{\bigl\lvert\{k\in\{1,\ldots,N_2\}:x^{(k)}=x\}\bigr\rvert}{N_2},
  \qquad x\in\{0,1\}^n.
\end{equation}
For a fixed candidate basis and each bitstring $x$, the empirical frequency averages $N_2$ independent Bernoulli variables with mean $q(x)$. Hoeffding's bound holds uniformly over the Step~1 outcomes, so averaging conditional on $\mathcal E_1$ gives
\begin{equation}
  \mathbb P\!\left[
    \lvert\widehat q(x)-q(x)\rvert>\frac{\xi}{2s}
    \,\middle|\,\mathcal E_1\right]
  \le2\exp\!\left(-\frac{N_2\xi^2}{2s^2}\right).
\end{equation}
Here, we assumed that Step~1 succeeded so the conditioning on $\mathcal E_1$ appears here.

The event $\mathcal E_2$ is guaranteed if $\lvert\widehat q(x)-q(x)\rvert\le\xi/(2s)$ for every bitstring $x$ in every candidate basis. Therefore, a union bound over $2^n$ bitstrings and at most $G$ candidate bases gives
\begin{equation}
  \mathbb P\!\left[
    \mathcal E_2^c \,\middle|\,\mathcal E_1\right]
  \le 2^{n+1}G\exp\!\left(-\frac{N_2\xi^2}{2s^2}\right).
  \label{eq:reconstruction-sampling-capture-failure}
\end{equation}
We set
\begin{equation}
  N_2
  :=\left\lceil\frac{2s^2}{\xi^2}\left[
    n\log2+\log\frac{6G}{\delta}
    \right]\right\rceil
  =O\!\left(
    s^2\varepsilon^{-4}\left[n+\log\frac G\delta\right]
    \right),
  \label{eq:pauli-moment-learning-capture-copies}
\end{equation}
where we used the fixed value $\xi=\varepsilon^2/32$. Substituting this into Eq.~\eqref{eq:reconstruction-sampling-capture-failure} gives
\begin{equation}
  \mathbb P\!\left[
    \mathcal E_2^c \,\middle|\,\mathcal E_1\right]
  \le 2^{n+1}G\exp\!\left(-n\log2-\log\frac{6G}{\delta}\right)
  =\frac\delta3.
\end{equation}

For each candidate basis passing Step~2, fix $S$ before estimating its submatrix on separate fresh copies of $\rho$. Let $\mathcal E_3$ be the event that Eq.~\eqref{eq:reconstruction-sampling-estimation-error} holds for all candidate bases that pass Step~2. For failure probability $\delta/(3Gs^2)$ per matrix entry, Theorem~\ref{thm:single-entry-tomography} requires at most
\begin{equation}
  \left\lceil4(\theta/2)^{-2}
    \log\frac{12Gs^2}{\delta}\right\rceil
\end{equation}
copies of $\rho$, and there are at most $s^2$ entries per candidate basis. Consequently, a sufficient number of copies of $\rho$ per candidate basis is
\begin{align}
  N_3
  &:=s^2\left\lceil4(\theta/2)^{-2}
    \log\frac{12Gs^2}{\delta}\right\rceil\nonumber\\
  &=O\!\left(
    s^5\varepsilon^{-2}\log\frac{sG}{\delta}
    \right).
  \label{eq:pauli-moment-learning-entry-copies}
\end{align}
The number of candidate bases is at most $G$, so a union bound over at most $Gs^2$ entries gives
\begin{equation}
  \mathbb P\!\left[
    \mathcal E_3^c
    \,\middle|\,\mathcal E_2
    \right]
  \le Gs^2\frac{\delta}{3Gs^2}
  =\frac\delta3.
\end{equation}

\subsection{Proof of Theorem~\ref{thm:pauli-moment-learning}}
\label{sec:pauli-moment-learning-proof}
\label{sec:pauli-moment-learning-resource-assembly}

On $\mathcal E_1\cap\mathcal E_2\cap\mathcal E_3$, Lemma~\ref{lem:accurate-candidate-basis} supplies an accurate candidate basis $B'$, and Lemma~\ref{lem:reconstruction-completeness} guarantees that it passes both tests if reached. If the algorithm stops earlier, it has already returned an accepted matrix. In either case, the test in Step~3 gives at most $s$ nonzero matrix entries in the reported product basis, and Lemma~\ref{lem:reconstruction-soundness} bounds the trace-norm error by $\varepsilon$.

The Pauli moment estimation guarantees in Section~\ref{sec:robust-basis-recovery} and the reconstruction sampling guarantees in Section~\ref{sec:reconstruction-certification} give
\begin{equation}
  \mathbb P(\mathcal E_1^c)\le\frac\delta3,
  \qquad
  \mathbb P(\mathcal E_2^c\mid\mathcal E_1)\le\frac\delta3,
  \qquad
  \mathbb P(\mathcal E_3^c\mid\mathcal E_2)\le\frac\delta3.
  \label{eq:pauli-moment-learning-step-failures}
\end{equation}
A union bound over $\mathcal E_1^c$, $\mathcal E_1\cap\mathcal E_2^c$, and $\mathcal E_1 \cap \mathcal E_2\cap\mathcal E_3^c$ therefore bounds total failure probability by $\delta/3+\delta/3+\delta/3=\delta$.

It remains to bound the resources. Each partial basis assignment chooses a set of at most $k_0$ qubits and one epsilon-net point for each chosen qubit. The fixed resolution in Eq.~\eqref{eq:pauli-moment-learning-final-tolerances} satisfies
\begin{equation}
  \varepsilon_2^{-2}
  =2^{24}\ell_0^2s^{2\log_2s+4}3^{\ell_0}n\varepsilon^{-4}.
\end{equation}
Using $k_0=O(s)$ and $\ell_0=O(1+\log s)$, the number $G$ of partial basis assignments therefore satisfies
\begin{equation}
  G
  \le\sum_{j=0}^{\min\{n,k_0\}}
    \binom nj\left(\frac C{\varepsilon_2}\right)^{2j}
  \le2^{O(s(\log s)^2)}
    \left(\frac n\varepsilon\right)^{O(s)}.
  \label{eq:pauli-moment-learning-assignment-count}
\end{equation}
Each assignment produces at most one candidate product basis, so the list of candidate product bases also has size at most $G$.

For each nonempty $W\subseteq[n]$ with $\lvert W\rvert\le\ell_0$, Step~1 estimates all $3^{\lvert W\rvert}$ Pauli observables. Then it gives
\begin{equation}
  N_{\mathrm{obs}}
  :=\sum_{j=0}^{\ell_0}\binom nj3^j
  \le(\ell_0+1)(3n)^{\ell_0}.
  \label{eq:pauli-moment-learning-tensor-count}
\end{equation}
This bounds $N_1$ in Eq.~\eqref{eq:pauli-moment-learning-low-order-copies}.

We now combine the sample complexity counts from Eqs.~\eqref{eq:pauli-moment-learning-low-order-copies}, \eqref{eq:pauli-moment-learning-capture-copies}, and~\eqref{eq:pauli-moment-learning-entry-copies}. While Step~1 uses $N_1$ copies of $\rho$ to estimate the low-order Pauli moments, Step~2 and Step~3 use $N_2$ and $N_3$ copies per candidate basis. A sufficient total number of copies is therefore
\begin{equation}
  \begin{aligned}
    N&:=N_1+G(N_2+N_3)\\
     &\le2^{O(s(\log s)^2)}
       \left(\frac n\varepsilon\right)^{O(s)}
       \polylog\!\left(\frac1\delta\right).
  \end{aligned}
  \label{eq:pauli-moment-learning-total-copies}
\end{equation}

We then count classical runtime. Step~1 processes the Pauli moment tensors in $O(nN_1)$ time and enumerates the $G$ partial basis assignments in $O(Gn)$ time. For each assignment, it scans at most $nN_{\mathrm{obs}}$ flattenings of the Pauli moment tensors, each with at most $3^{\ell_0}$ entries. Each contraction and SVD costs $\poly(3^{\ell_0})$ time, so the basis candidate construction cost is $O(GnN_{\mathrm{obs}}\poly(3^{\ell_0}))$. In Step~2, counting the observed bitstrings takes $O(nN_2)$ time per candidate basis. Selecting up to $s$ most frequent bitstrings and applying the capture test take $O(N_2)$ time. Thus Step~2 costs at most $O(GnN_2)$ time. In Step~3, each selected-entry tomography outcome takes $O(n)$ time, giving $O(nN_3)$ per candidate basis that passes Step~2. Processing the output matrix takes $O(s^2)$ time per candidate basis. Thus Step~3 costs at most $O(G(nN_3+s^2))$ time. Summing the three steps gives
\begin{equation}
  \begin{aligned}
    T&\le O\!\left(
      nN_1+GnN_{\mathrm{obs}}\poly(3^{\ell_0})
      +GnN_2+G(nN_3+s^2)
      \right)\\
     &\le2^{O(s(\log s)^2)}
       \left(\frac n\varepsilon\right)^{O(s)}
       \polylog\!\left(\frac1\delta\right).
  \end{aligned}
  \label{eq:pauli-moment-learning-classical-runtime}
\end{equation}
These bounds complete the proof of Theorem~\ref{thm:pauli-moment-learning}.

%% file: 06_tree_merging.tex
\section{Learning with tree-structured merging}
\label{sec:algorithm1-results}

In this section, we present Algorithm~2, briefly outlined in Section~\ref{sec:algorithm1-overview}, which uses tree-structured merging to learn $s$-sparse states. We first give an explicit description of Algorithm~2 the results. We then analyze Algorithm~2 for the case where the unknown state is sparse in an arbitrary product basis using $\epsilon$-nets. We then analyze Algorithm~2 for the case where the unknown state is sparse in a product of a known finite set of product bases.

\subsection{Explicit algorithm and results}
\label{sec:tree-hierarchy}

\begin{figure}
  \centering
  \includegraphics[width=0.7\textwidth]{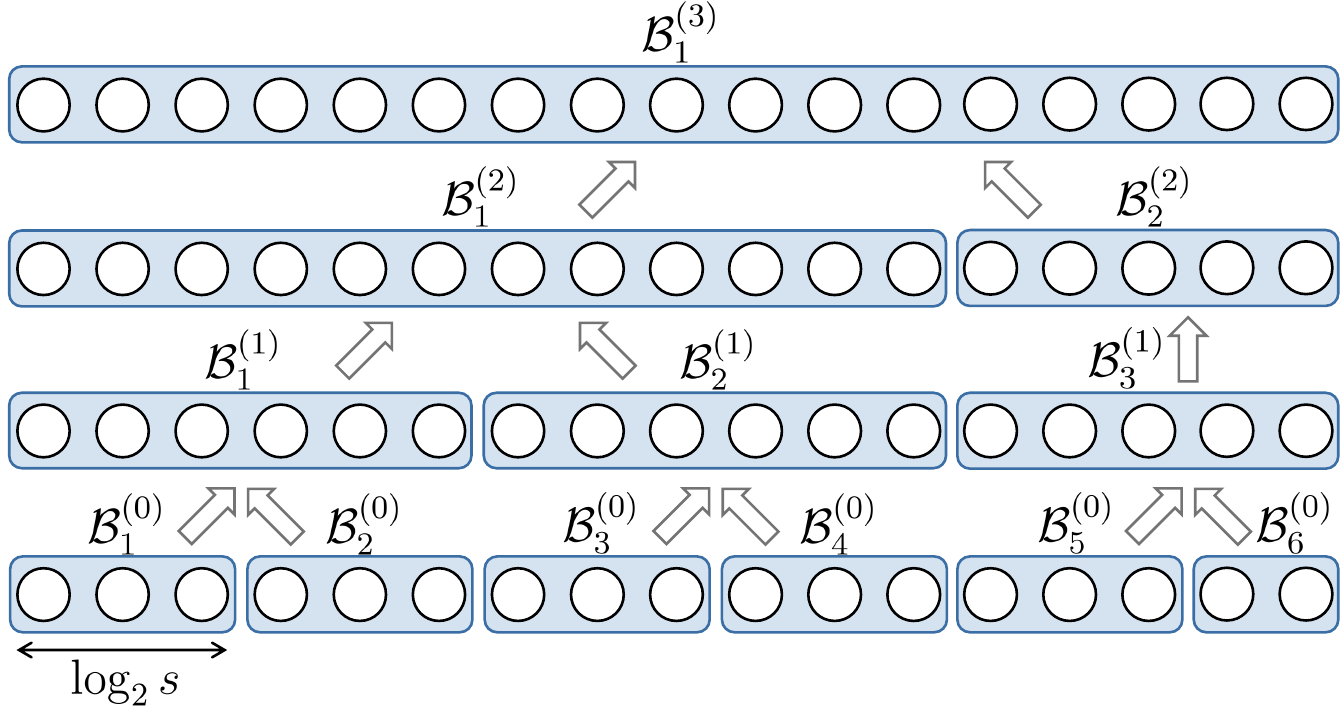}
  \caption{Example of the tree-merging procedure for $n=17$ and $m=3$. At level $0$, all blocks have size $m=3$, except for the last block. At level $2$, there is an odd number of blocks, and so the last block is directly passed to level $3$.}
  \label{fig:Tree-in-detail}
\end{figure}

As Algorithm~1 does, Algorithm~2 takes independently prepared copies of the unknown $s$-sparse state $\rho$, accuracy of the output $\varepsilon$, and failure probability $\delta$. In Section~\ref{sec:algorithm1-overview}, we presented Algorithm~2 implicitly assuming $n$ is proportional to $2^D$ for simplicity. Here, we present Algorithm~2 in full generality and specify the required tomography precision at each step for the tree-merging procedure.

Since $t$-sparse states are also $s$-sparse for $s \ge t$, we assume $s \ge 2$ without loss of generality. We define
\begin{equation}
  m:=\lfloor\log_2s\rfloor,
  \qquad
  K_0:=\left\lceil\frac{n}{m}\right\rceil,
  \qquad
  D:=\lceil\log_2K_0\rceil,
  \label{eq:tree-internal-parameters}
\end{equation}
where $m$ is the size of the leaf blocks at level $0$, $K_0$ is the number of blocks at level $0$, and $D$ is the depth of the tree. Specifically, we partition $[n]$ into $K_0$ consecutive blocks such that the first $K_0-1$ blocks each contain $m$ qubits and the final block contains $n-(K_0-1)m$ qubits, so that
\begin{equation}
  [n]=\bigsqcup_{j=1}^{K_0}\mathcal B_j^{(0)}.
\end{equation}
At level $d$, let $K_d$ be the number of blocks and denote the $j$-th block by $\mathcal B_j^{(d)}$. For $0\leq d<D$, set $K_{d+1}:=\lceil K_d/2\rceil$ and define
\begin{equation}
  \mathcal B_j^{(d+1)}
  :=
  \begin{cases}
    \mathcal B_{2j-1}^{(d)}\sqcup\mathcal B_{2j}^{(d)},
     & 1\leq j\leq\lfloor K_d/2\rfloor,  \\
    \mathcal B_{K_d}^{(d)},
     & j=K_{d+1}\text{ if $K_d$ is odd}.
  \end{cases}
  \label{eq:tree-blocks}
\end{equation}
Thus a merge combines two adjacent child blocks, while an unmatched final block is transferred to the next level unchanged. Proceeding until level $D$, we obtain the root block $\mathcal B_1^{(D)} = [n]$; see Fig.~\ref{fig:Tree-in-detail}.

With this tree structure, we set the parameters $\varepsilon_1$, $\varepsilon_2$, and $\delta_1$ as follows:
\begin{equation}
  \varepsilon_1:=\frac{\varepsilon^2}{64s^2n^3},
  \qquad
  \varepsilon_2:=\frac{\varepsilon^2}{64sn^3},
  \qquad
  \delta_1:=\frac{\delta}{n}.
  \label{eq:algorithm1-merge-parameters}
\end{equation}
With this setup, Algorithm~2 proceeds as follows:

\begin{itemize}
  \item \textbf{Step 1: Initialize the leaves.}
        For every leaf block $\mathcal B_j^{(0)}$, set
        \begin{equation}
          P_j^{(0)}:=I_{\mathcal B_j^{(0)}}.
          \label{eq:tree-leaf-projector}
        \end{equation}
        Choose an arbitrary product basis on the leaf and retain all of its basis vectors. Since $m=\lfloor\log_2s\rfloor$, the number of basis vectors is at most $s$. No copies of $\rho$ are measured at this initialization.

  \item \textbf{Step 2: Recursive merging.}
        From $d=0$ to $D-1$, we repeat the following procedure:
        \begin{enumerate}[label=(\alph*)]
          \item \textbf{Tomography on temporary projectors.}
                For $j=1, \dots, \lfloor K_d/2 \rfloor$, define the temporary projector $A_j^{(d+1)}$ by
                \begin{equation}
                  A_j^{(d+1)}
                  :=P_{2j-1}^{(d)}\otimes P_{2j}^{(d)}.
                  \label{eq:compressed-parent}
                \end{equation}
                Let $B_L$ and $B_R$ denote the bases for $\mathcal B_{2j-1}^{d}$ and $\mathcal B_{2j}^{d}$, respectively. The temporary projector spans at most $s^2$ basis vectors in $B_L \times B_R$ and has rank at most $s^2$. Run the selected-entry tomography in Corollary~\ref{cor:selected-submatrix-tomography} with an approximation error $\varepsilon_1$ and failure probability $\delta_1$ to estimate the submatrix of $\rho_{\mathcal B_j^{(d+1)}}$ supported on $A_j^{(d+1)}$. The size of each submatrix is at most $s^2\times s^2$. This tomography can be done in parallel for all $j$ because the blocks $\mathcal B_j^{(d+1)}$ are disjoint. This returns a Hermitian matrix $\widehat\sigma_j^{(d+1)}$ that approximates
                \begin{equation}
                  \sigma_j^{(d+1)}:=A_j^{(d+1)}\rho_{\mathcal B_j^{(d+1)}}A_j^{(d+1)}.
                  \label{eq:compressed-matrix-estimate}
                \end{equation}
          \item \textbf{Compress sparsity.}
                If $d<D-1$, then for each $j$ we find the projector $P_j^{(d+1)}$ consisting of at most $s$ product basis vectors that minimizes
                \begin{equation}
                  \widehat\tau(P):=\operatorname{Tr}[(I-P)\widehat\sigma_j^{(d+1)}]
                \end{equation}
                over product bases. It returns $P_j^{(d+1)}$, the corresponding product basis, and its retained basis vectors. We search for the best product basis over the product of a single-qubit $\varepsilon_2$-net when the input state $\rho$ is $s$-sparse in an arbitrary product basis. If the input state $\rho$ is $s$-sparse in $\mathcal{A}^n$ for some known set of single-qubit bases $\mathcal{A}$, the search space is specified in Section~\ref{sec:finite-set-result}. This step is done by classical computation and no copies of $\rho$ are measured. If $K_d$ is odd, carry the last block $\mathcal{B}_{K_{d}}^{(d)}$, its projector, and its product basis unchanged to the next level.
        \end{enumerate}

  \item \textbf{Step 3: Reconstruct and return the output.}
        The final merge in Step~2 produces the estimate $\widehat\sigma_1^{(D)}$, supported on the temporary root projector $A_1^{(D)}$. If we compact the sparsity as in Step~2(b), the output can be up to $s^2$ sparse. We proceed as follows instead to obtain an $s$-sparse matrix. For each basis $B$, form $\widehat\rho_B$ by retaining the at most $s$ entries of $[\widehat\sigma_1^{(D)}]_B$ having largest magnitude and truncate all other entries to zero. Find the product basis $\widehat B$ satisfying
        \begin{equation}
          \left\|\widehat\rho_B-\widehat\sigma_1^{(D)}\right\|_F
          \leq\frac{\varepsilon}{4\sqrt{s}}.
        \end{equation}
        As in Step~2(b), we search for such a basis among the product single-qubit $\varepsilon_2$-net if the input state $\rho$ is $s$-sparse in an arbitrary product basis. If the input state $\rho$ is $s$-sparse in $\mathcal{A}^n$ for some known set of single-qubit bases $\mathcal{A}$, the search space is specified in Sec.\ref{sec:finite-set-result}. Return the basis as $\widehat B$ and the truncated matrix $\widehat\rho_B$ as the final output $\widehat\rho$. This step is done by classical computation and consumes no copies of $\rho$.
\end{itemize}

With this procedure, Algorithm~2 first achieves polynomial sample complexity on an arbitrary unknown product basis while leaving the classical computational cost of the product-basis search exhaustive.

\begin{theorem}[Algorithm~2; arbitrary continuous product basis]
  \label{thm:tree-continuous-basis}
  Suppose that $\rho$ is $s$-sparse in an arbitrary unknown product basis, and let $0<\varepsilon, \delta <1$. With probability at least $1-\delta$, Algorithm~2 only uses single-qubit measurements and outputs an $s$-sparse operator $\widehat\rho$ such that $\lVert\widehat\rho-\rho\rVert_1\leq\varepsilon$, with
    \begin{equation}
    \begin{aligned}
      N&=O\!\left(n^6s^{12}\varepsilon^{-4}\log^2\!\left(\frac{ns}{\delta}\right)\right)
      &&\text{copies of $\rho$, and}\\
      T&=\left(\frac {ns}\varepsilon\right)^{O(n)}{\rm polylog}\left(\frac{1}{\delta}\right)
      &&\text{classical runtime}.
    \end{aligned}
    \label{eq:tree-continuous-basis-resources}
  \end{equation}
\end{theorem}

We next assume that every local basis is drawn from a known finite set of single-qubit bases $\mathcal A$, i.e., $\rho$ is $s$-sparse in some unknown product basis $B=(B_1,\dots,B_n)$ with $B_i\in\mathcal A$. Although this gives $\lvert\mathcal A\rvert^n$ product-basis assignments, our analysis using the Donoho--Stark uncertainty principle reduces the exponential search space to $n^{O(\log s)}$ candidates. To be specific, we write $B=\{\lvert 0_B\rangle,\lvert 1_B\rangle\}$ for each $B\in\mathcal A$, and define
\begin{equation}
  \mu:=
  \max_{\substack{B,B'\in\mathcal A\\B\ne B'}}
  \max_{u,v\in\{0,1\}}
  \lvert\langle u_B\mid v_{B'}\rangle\rvert,
  \qquad
  \kappa:=-\log\mu.
  \label{eq:finite-set-separation}
\end{equation}
Then $\mu<1$, so $\kappa>0$ quantifies the minimum pairwise separation of distinct bases in $\mathcal A$. With that, we have the following result.

\begin{theorem}[Algorithm~2; finite set of single-qubit bases]
  \label{thm:tree-finite-bases}
  Let $\mathcal A$ be a known finite set of single-qubit bases with $\lvert\mathcal A\rvert>1$ and minimum pairwise separation $\kappa>0$. Suppose that $\rho$ is $s$-sparse in an unknown product basis $B=(B_1,\ldots,B_n)\in\mathcal A^n$, and let $0<\varepsilon, \delta <1$. With probability at least $1-\delta$, Algorithm~2 only uses single-qubit measurements and outputs a product basis $\widehat B\in\mathcal A^n$ and an $s$-sparse operator $\widehat\rho$ in that basis such that $\lVert\widehat\rho-\rho\rVert_1\leq\varepsilon$, with
    \begin{equation}
    \begin{aligned}
      N&=O\!\left(n^6s^{12}\varepsilon^{-4}\log^2\!\left(\frac{ns}{\delta}\right)\right)
      &&\text{copies of $\rho$, and}\\
      T&=\left(\lvert\mathcal A\rvert n\right)^{O(\log(s)/\kappa)}\varepsilon^{-4}\operatorname{polylog}\!\left(\delta^{-1}\right)
      &&\text{classical runtime}.
    \end{aligned}
    \label{eq:tree-finite-bases-resources}
  \end{equation}
\end{theorem}

The rest of this section is devoted to the proofs of Theorem~\ref{thm:tree-continuous-basis} and Theorem~\ref{thm:tree-finite-bases}. Note that Algorithm~2 is highly adaptive: the measurement bases at each level are chosen based on the previous tomography results. Therefore, the main challenge to prove the results is that the tomography errors in the previous steps will propagate to the next level, which can be amplified. Section~\ref{sec:ideal-merge} addresses this challenge by carefully bounding the errors at each merge step, and this will be adapted to the proofs of Theorem~\ref{thm:tree-continuous-basis} in Section~\ref{sec:polynomial-copy-baseline} and Theorem~\ref{thm:tree-finite-bases} in Section~\ref{sec:finite-set-result}.

\subsection{Leakage propagation through merges}
\label{sec:ideal-merge}

To establish the bound on the error accumulation through the adaptive tomography procedure, we analyze how leakage errors evolve in each merge in Step~2. Specifically, for each level $d$, let $\eta_d$ denote a common upper bound on the leakage at that level: for every $j$,
\begin{equation}
  \operatorname{Tr}[(I-P_j^{(d)})\rho_{\mathcal B_j^{(d)}}] \leq \eta_d.
\end{equation}
By construction, $\eta_0=0$.

To see how the leakage in the previous step propagates to the next level, consider a merge of two adjacent child blocks $L=\mathcal B_{2j-1}^{(d)}$ and $R=\mathcal B_{2j}^{(d)}$ at level $d$, which form $Q=\mathcal B_j^{(d+1)}$. For notational simplicity, let us denote
\begin{equation}
  P_L:=P_{2j-1}^{(d)},
  \quad
  P_R:=P_{2j}^{(d)},
  \quad
  A:=A_j^{(d+1)},
\end{equation}
so that $A = P_L \otimes P_R$ is the temporary projector on $Q$. We then have the projected state $\sigma_j^{(d+1)} = A \rho_Q A$ and its tomography result $\widehat\sigma_j^{(d+1)}$ from Eq.~\eqref{eq:compressed-parent}. For any candidate product-basis projector $P$ on $Q$, we define the exact leakage $\tau_j^{(d+1)}$ and the empirical leakage $\widehat\tau_j^{(d+1)}$ as
\begin{equation}
  \tau_j^{(d+1)}(P)
  :=\operatorname{Tr}[(I-P)\sigma_j^{(d+1)}],
  \qquad
  \widehat\tau_j^{(d+1)}(P)
  :=\operatorname{Tr}[(I-P)\widehat\sigma_j^{(d+1)}].
  \label{eq:tree-losses}
\end{equation}

Let $P_{j,\star}^{(d+1)}$ be the projector onto the diagonal support of $\rho_Q$ in the hidden sparse basis. By positivity and inherited $s$-sparsity, it has rank at most $s$ and satisfies $P_{j,\star}^{(d+1)}\rho_QP_{j,\star}^{(d+1)}=\rho_Q$, so it is the ideal target of Step~2(b). Three effects can increase the leakage of the projector returned:
\begin{itemize}
  \item First, leakage accumulated at the child blocks means that the compressed operator $\sigma_j^{(d+1)}=A_j^{(d+1)}\rho_QA_j^{(d+1)}$ need not equal $\rho_Q$, so even $P_{j,\star}^{(d+1)}$ can have nonzero leakage on $\sigma_j^{(d+1)}$.
  \item Second, the compaction minimizes the empirical leakage $\widehat\tau_j^{(d+1)}(P)$ computed from $\widehat\sigma_j^{(d+1)}$, which differs from the true leakage $\tau_j^{(d+1)}(P)$ because of tomography error.
  \item Third, for the arbitrary basis case in Theorem~\ref{thm:tree-continuous-basis}, the limited search space of product bases introduces an additional error. Specifically, the optimal product basis is replaced by a nearby one from the product of single-qubit $\varepsilon_2$-nets, which is controlled by $\xi$ in Lemma~\ref{lem:generic-one-merge-stability} below. For the fixed basis case in Theorem~\ref{thm:tree-finite-bases}, we do not have this error because the search space contains the true product basis.
\end{itemize}

\begin{lemma}[Leakage propagation]
  \label{lem:generic-one-merge-stability}
  Suppose both child projectors $P_L$ and $P_R$ have leakage at most $\eta_d$ and the tomography result obeys $\lVert\widehat\sigma_j^{(d+1)}-\sigma_j^{(d+1)}\rVert_F\leq\varepsilon_1$. Also, assume the search space of the product bases considered in Step~2(b) contains $\overline P$ satisfying
  \begin{equation}
    \tau_j^{(d+1)}(\overline P)
    \leq\tau_j^{(d+1)}(P_{j,\star}^{(d+1)})+\xi,
    \label{eq:generic-comparator-gap}
  \end{equation}
  for some $\xi\geq0$. If $\widehat P$ is returned by the empirical leakage minimization in Step~2(b), then
  \begin{equation}
    \operatorname{Tr}[(I-\widehat P)\rho_Q]
    \leq8\eta_d+4s\varepsilon_1+2\xi.
    \label{eq:generic-merge-recurrence}
  \end{equation}
\end{lemma}

\begin{proof}
  For readability, we write the projected state and its tomography result as
  \begin{equation}
    \sigma:=\sigma_j^{(d+1)},
    \quad
    \widehat\sigma:=\widehat\sigma_j^{(d+1)}.
    \label{eq:merge-leakage-proof-notation}
  \end{equation}
  Also write $\tau:=\tau_j^{(d+1)}$, $\widehat\tau:=\widehat\tau_j^{(d+1)}$, and $P_\star:=P_{j,\star}^{(d+1)}$.

  First, since $A=P_L\otimes P_R$, we have
  \begin{equation}
    \begin{aligned}
      I-A
       & =(I-P_L)\otimes I+P_L\otimes(I-P_R)      \\
       & \preceq(I-P_L)\otimes I+I\otimes(I-P_R).
    \end{aligned}
    \label{eq:merge-temporary-projector}
  \end{equation}
  Taking the expectation in $\rho_Q$ gives
  \begin{equation}
    \operatorname{Tr}[(I-A)\rho_Q]\le \operatorname{Tr}[(I-P_L)\rho_L]+\operatorname{Tr}[(I-P_R)\rho_R] \le 2\eta_d.
    \label{eq:merge-temporary-leakage}
  \end{equation}
  Because $\sigma$ and $\widehat\sigma$ are supported on $A$, every candidate $P$ satisfies
  \begin{equation}
    \widehat\tau(P)-\tau(P)
    =\operatorname{Tr}[A(I-P)A(\widehat\sigma-\sigma)].
    \label{eq:merge-loss-difference}
  \end{equation}
  Therefore, by the Cauchy-Schwarz inequality,
  \begin{equation}
    \lvert\widehat\tau(P)-\tau(P)\rvert
    \leq\lVert A(I-P)A\rVert_F \lVert\widehat\sigma-\sigma\rVert_F
    \leq s\varepsilon_1.
    \label{eq:merge-uniform-loss-error}
  \end{equation}
  Here, we used that $A(I-P)A$ has eigenvalues at most one and rank at most $s^2$. Meanwhile, the ideal projector $P_\star$ has $\operatorname{Tr}[(I-P_\star)\rho_Q]=0$, so
  \begin{equation}
    \begin{aligned}
      \tau(P_\star)
       & =\lVert(I-P_\star)A\rho_Q^{1/2}\rVert_F^2     \\
       & =\lVert(I-P_\star)(I-A)\rho_Q^{1/2}\rVert_F^2 \\
       & \leq\operatorname{Tr}[(I-A)\rho_Q]
      \leq2\eta_d.
    \end{aligned}
    \label{eq:merge-exact-comparison}
  \end{equation}
  Empirical optimality, together with Eqs.~\eqref{eq:generic-comparator-gap},~\eqref{eq:merge-uniform-loss-error}, and~\eqref{eq:merge-exact-comparison}, gives
  \begin{equation}
    \begin{aligned}
      \tau(\widehat P)
       & \leq\widehat\tau(\widehat P)+s\varepsilon_1  \\
       & \leq\widehat\tau(\overline P)+s\varepsilon_1 \\
       & \leq\tau(\overline P)+2s\varepsilon_1        \\
       & \leq2\eta_d+\xi+2s\varepsilon_1.
    \end{aligned}
    \label{eq:generic-selected-comparison}
  \end{equation}
  Finally, the triangle inequality for leakage error in Lemma~\ref{lem:projection-leakage-triangle} gives
  \begin{equation}
    \begin{aligned}
      \sqrt{\operatorname{Tr}[(I-\widehat P)\rho]}
       & \leq\sqrt{\operatorname{Tr}[(I-A)\rho]}
      +\sqrt{\operatorname{Tr}[(I-\widehat P)A\rho A]}                \\
       & =\sqrt{\operatorname{Tr}[(I-A)\rho]}+\sqrt{\tau(\widehat P)} \\
       & \leq\sqrt{2\eta_d}+\sqrt{2\eta_d+\xi+2s\varepsilon_1}.
    \end{aligned}
  \end{equation}
  Squaring this inequality and using $(\sqrt a+\sqrt b)^2\leq2a+2b$ gives
  \begin{equation}
    \begin{aligned}
      \operatorname{Tr}[(I-\widehat P)\rho]
       & \leq\left(\sqrt{2\eta_d}+\sqrt{2\eta_d+\xi+2s\varepsilon_1}\right)^2 \\
       & \leq8\eta_d+4s\varepsilon_1+2\xi.
    \end{aligned}
    \label{eq:generic-projection-transfer}
  \end{equation}
\end{proof}

Applying Lemma~\ref{lem:generic-one-merge-stability} at level $d$ with $\xi=\xi_d$, and using the fact that carried blocks keep their previous leakage, gives
\begin{equation}
  \eta_{d+1}
  \leq8\eta_d+4s\varepsilon_1+2\xi_d.
  \label{eq:generic-level-recurrence}
\end{equation}
Starting from $\eta_0=0$ and iterating Eq.~\eqref{eq:generic-level-recurrence} yields
\begin{equation}
  \eta_{D-1}
  \leq\frac{4s(8^{D-1}-1)}7\varepsilon_1
  +2\sum_{d=0}^{D-2}8^{D-2-d}\xi_d.
  \label{eq:generic-root-iteration}
\end{equation}
Eq.~\eqref{eq:generic-root-iteration} controls the projectors entering the final merge, whose estimate is used in Step~3. It remains to bound $\xi_d$ in the two settings: Section~\ref{sec:polynomial-copy-baseline} uses the single-qubit $\varepsilon_2$-net for Theorem~\ref{thm:tree-continuous-basis}, whereas Section~\ref{sec:finite-set-result} shows $\xi_d=0$ for the fixed basis set in Theorem~\ref{thm:tree-finite-bases}.

\subsection{Arbitrary continuous product bases (Theorem~\ref{thm:tree-continuous-basis})}
\label{sec:polynomial-copy-baseline}

Now we prove Theorem~\ref{thm:tree-continuous-basis}, where the input state $\rho$ is $s$-sparse in an arbitrary product basis. In the arbitrary-basis setting, Step~2(b) searches the product of single-qubit $\varepsilon_2$-nets, which leads to error $\xi_d$ in the leakage recurrence from Section~\ref{sec:ideal-merge}. The recurrence then controls the leakage of the projector entering Step~3. The remaining procedure in Step~3 is handled similarly. We conclude by counting the number of copies of $\rho$, failure probability, and classical runtime.

\paragraph{Leakage with epsilon-net.}

Again, consider a merge from level $d$ to level $d+1$ on $Q:=\mathcal B_j^{(d+1)}$, using the notation of Section~\ref{sec:ideal-merge}. Then, following Lemma~\ref{lem:generic-one-merge-stability}, we have the following recurrence formula:

\begin{corollary}
  \label{cor:net-merge-leakage}
  Suppose both child projectors $P_L$ and $P_R$ have leakage at most $\eta_d$ and the tomography result obeys $\lVert\widehat\sigma_j^{(d+1)}-\sigma_j^{(d+1)}\rVert_F\leq\varepsilon_1$. If $\widehat P$ is returned by the empirical leakage minimization in Step~2(b) searching over the product of single-qubit $\varepsilon_2$-nets, then
  \begin{equation}
    \operatorname{Tr}[(I-\widehat P_j^{(d+1)})\rho_Q]
    \leq8\eta_d+4s\varepsilon_1
    +2\lvert Q\rvert\sin\frac{\varepsilon_2}{2}.
    \label{eq:net-merge-leakage-recurrence}
  \end{equation}
\end{corollary}

\begin{proof}
  Let $P_\star:=P_{j,\star}^{(d+1)}$. Order the qubits in $Q$, write the hidden basis corresponding to $P_\star$ as $B=(B_1,\ldots,B_{\lvert Q\rvert})$, and choose $\overline B = (\overline B_1,\ldots,\overline B_{\lvert Q\rvert})\in\mathcal N(\varepsilon_2)^{\lvert Q\rvert}$ where each $\overline B_i$ is within angular distance $\varepsilon_2$ of $B_i$. For each qubit $i$, let $V_i$ and $V_i'$ be the single-qubit unitaries that map the computational basis to $\overline B_i$ and $B_i$ respectively, i.e.,
  \begin{equation}
    V_i\lvert b\rangle=\lvert\phi_b^{(i)}\rangle,
    \qquad
    V_i'\lvert b\rangle=\lvert\overline\phi_b^{(i)}\rangle,
    \qquad b\in\{0,1\},
  \end{equation}
  where $B_i = \{\lvert\phi_0^{(i)}\rangle,\lvert\phi_1^{(i)}\rangle\}$ and $\overline B_i = \{\lvert\overline\phi_0^{(i)}\rangle,\lvert\overline\phi_1^{(i)}\rangle\}$. Let $S\subseteq\{0,1\}^{\lvert Q\rvert}$ be the set of bitstrings whose basis vectors are retained by $P_\star$. Then, we have
  \begin{equation}
    P_\star=V\Pi_SV^\dagger,
    \qquad
    \overline P:=V'\Pi_SV'^\dagger,
  \end{equation}
  where we denote $\Pi_S:=\sum_{x\in S}\lvert x\rangle\!\langle x\rvert$, $V:=\bigotimes_{i=1}^{\lvert Q\rvert}V_i$, and $V':=\bigotimes_{i=1}^{\lvert Q\rvert}V_i'$.

  For $0\leq r\leq\lvert Q\rvert$, define
  \begin{equation}
    V^{(r)}
    :=\left(\bigotimes_{i=1}^{r}V_i'\right)
    \otimes\left(\bigotimes_{i=r+1}^{\lvert Q\rvert}V_i\right),
    \qquad
    P^{(r)}:=V^{(r)}\Pi_SV^{(r)\dagger}.
  \end{equation}
  Thus $V^{(0)}=V$, $V^{(\lvert Q\rvert)}=V'$, $P^{(0)}=P_\star$, and $P^{(\lvert Q\rvert)}=\overline P$. Let
  \begin{equation}
    \theta_r:=d_{\mathrm{ang}}(\mathbf n_r,\overline{\mathbf n}_r)
    \leq\varepsilon_2.
  \end{equation}
  Only the $r$-th local basis changes between $P^{(r-1)}$ and $P^{(r)}$. Let $S^{(-r)}$ be the set of bitstrings in $S$ obtained by omitting the $r$-th qubit,
  \begin{equation}
    S^{(-r)}:=\{z \in \{0,1\}^{|Q|-1}: (z,0)\in S \text{ or } (z,1)\in S\}.
  \end{equation}
  Here, $(z,b)$ denotes the bitstring obtained by appending $b$ to $z$ on the $r$-th qubit for $b=0,1$.
  \begin{equation}
    P^{(r)}-P^{(r-1)}=\bigoplus_{z\in S^{(-r)}}\left(P_z^{(r)}-P_z^{(r-1)}\right).
  \end{equation}
  The operator norm of an orthogonal block-diagonal operator is the largest
  norm of its blocks, and thus we have:
  \begin{equation}
    \left\lVert P^{(r)}-P^{(r-1)}\right\rVert_{\mathrm{op}}
    =\max_z
    \left\lVert P_z^{(r)}-P_z^{(r-1)}\right\rVert_{\mathrm{op}}.
  \end{equation}
  For each $z\in S^{(-r)}$, if both $(z,0)$ and $(z,1)$ are in $S$, then both $P_z^{(r)}$ and $P_z^{(r-1)}$ act as the identity on qubit $r$, and thus $P_z^{(r)}-P_z^{(r-1)}=0$. If only one of $(z,0)$ and $(z,1)$ is in $S$, we have
  \begin{equation}
    \begin{aligned}
      \left\lVert
      P_z^{(r)}-P_z^{(r-1)}
      \right\rVert_{\mathrm{op}}
       & =\frac12\left\lVert
      (\overline{\mathbf n}_r-\mathbf n_r)
      \cdot\boldsymbol{\sigma}
      \right\rVert_{\mathrm{op}} \\
       & =\frac12
      \left\lVert\overline{\mathbf n}_r-\mathbf n_r\right\rVert_2
      =\sin\frac{\theta_r}{2}.
    \end{aligned}
  \end{equation}
  Here, $\mathbf n_r$ and $\overline{\mathbf n}_r$ are the Bloch vectors of $\lvert\phi_0^{(r)}\rangle$ and $\lvert\overline\phi_0^{(r)}\rangle$, respectively. Therefore, we have
  \begin{equation}
    \left\lVert P^{(r)}-P^{(r-1)}\right\rVert_{\mathrm{op}}
    \leq\sin\frac{\theta_r}{2}
    \leq\sin\frac{\varepsilon_2}{2}.
  \end{equation}
  Telescoping over the local basis changes gives
  \begin{equation}
    \left\lVert\overline P-P_\star\right\rVert_{\mathrm{op}}
    \leq
    \sum_{r=1}^{\lvert Q\rvert}
    \left\lVert P^{(r)}-P^{(r-1)}\right\rVert_{\mathrm{op}}
    \leq
    \lvert Q\rvert\sin\frac{\varepsilon_2}{2}.
  \end{equation}
  Since $\sigma_j^{(d+1)}\succeq0$ and
  $\operatorname{Tr}\sigma_j^{(d+1)}\leq1$,
  \begin{equation}
    \begin{aligned}
      \tau(\overline P)-\tau(P_\star)
       & =
      \operatorname{Tr}\!\left[
                           (P_\star-\overline P)\sigma_j^{(d+1)}
                           \right] \\
       & \leq
      \left\lVert\overline P-P_\star\right\rVert_{\mathrm{op}}
      \operatorname{Tr}\sigma_j^{(d+1)}                         \\
       & \leq
      \lvert Q\rvert\sin\frac{\varepsilon_2}{2}.
    \end{aligned}
    \label{eq:net-merge-leakage-net-comparison}
  \end{equation}
  Thus the assumption of Lemma~\ref{lem:generic-one-merge-stability} holds with $\xi=\lvert Q\rvert\sin(\varepsilon_2/2)$, and applying that lemma concludes the proof.
\end{proof}

The blocks $\mathcal B_j^{(d+1)}$ contain at most $2^{d+1}m$ qubits. Hence
\begin{equation}
  \xi_d\leq2^{d+1}m\sin\frac{\varepsilon_2}{2}.
\end{equation}
Substituting this into Eq.~\eqref{eq:generic-root-iteration} gives
\begin{equation}
  \eta_{D-1}
  \leq
  \frac{4s(8^{D-1}-1)}7\varepsilon_1
  +\frac{2m}{3}(8^{D-1}-2^{D-1})\sin\frac{\varepsilon_2}{2}.
  \label{eq:net-merge-leakage-root-iteration}
\end{equation}
Recall $D=\lceil\log_2K_0\rceil$ with $K_0=\lceil n/m\rceil$. Using the parameterizations in Eq.~\eqref{eq:algorithm1-merge-parameters}, we have
\begin{equation}
  \eta_{D-1}
  <\frac{\varepsilon^2}{112s}
  +\frac{\varepsilon^2}{192s}
  =\frac{19\varepsilon^2}{1344s}.
  \label{eq:tree-continuous-basis-tree-tolerances}
\end{equation}
The projector entering Step~3 has leakage controlled in this way.

\paragraph{Approximating input state with epsilon-net.}

Given the projector $A^{(D)}_1$ from the end of Step~2 with small leakage $\eta_{D-1}$ in Eq.~\eqref{eq:tree-continuous-basis-tree-tolerances}, we reconstruct the $s$-sparse output state $\widehat\rho$ along with a basis $B' \in \mathcal{N}(\varepsilon_2)^n$ in Step~3. To this end, we show that such an output basis exists in $\mathcal{N}(\varepsilon_2)^n$.

\begin{lemma}[Approximation of input state with epsilon-net]
  \label{lem:net-state-approximation}
  Let $B = (B_1,\ldots,B_n)$ be a product basis, and suppose that $\rho$ is $s$-sparse in $B$. Let $B' = (B_1',\ldots,B_n')\in \mathcal{N}(\varepsilon_2)^n$ be a product basis such that $d_{\mathrm{ang}}(B_i,B_i')\leq\varepsilon_2$ for all $i \in \{1,\ldots,n\}$. For $x \in \{0,1\}^n$, let $\lvert\phi_x\rangle$ and $\lvert\phi'_x\rangle$ be the basis vectors in $B$ and $B'$ respectively. Define
  \begin{equation}
    \rho'
    :=\sum_{x,y\in\{0,1\}^n}[\rho]_B(x,y)
    \lvert \phi'_x\rangle\!\langle \phi'_y\rvert.
    \label{eq:net-state-approximation-comparator}
  \end{equation}
  Then, $\rho'$ is an $s$-sparse state in $B'$ such that
  \begin{equation}
    \lVert\rho'-\rho\rVert_F
    \leq4n\sin\frac{\varepsilon_2}{4}.
    \label{eq:net-state-approximation-error}
  \end{equation}
\end{lemma}

\begin{proof}
  Let $V=\bigotimes_{i=1}^n V_i$ and
  $V'=\bigotimes_{i=1}^n V_i'$ map the computational basis to $B$ and $B'$.
  For $\theta_i:=d_{\mathrm{ang}}(B_i,B_i')\le \varepsilon_2$, the eigenvalues of $V_i'V_i^\dagger$ are $\pm e^{\pm i\theta_i}$ up to a global phase. Then, we have
  \begin{equation}
    \lVert V_i'-V_i\rVert_{\mathrm{op}}
    =\lVert V_i'V_i^\dagger-I\rVert_{\mathrm{op}}
    =\max\{\lvert e^{i\theta_i}-1\rvert, \lvert e^{-i\theta_i}-1\rvert\}
    =2\sin\frac{\theta_i}{4}
    \leq2\sin\frac{\varepsilon_2}{4}.
    \label{eq:net-state-approximation-local-unitary-error}
  \end{equation}
  This gives
  \begin{equation}
    \begin{aligned}
      \lVert V'-V\rVert_{\mathrm{op}}
       & \leq\sum_{j=1}^n
      \left\lVert
      \left(\bigotimes_{i<j}V_i'\right)
      (V_j'-V_j)
      \left(\bigotimes_{i>j}V_i\right)
      \right\rVert_{\mathrm{op}}            \\
       & \leq2n\sin\frac{\varepsilon_2}{4}.
    \end{aligned}
    \label{eq:net-state-approximation-product-unitary-error}
  \end{equation}
  Write $C:=[\rho]_B=[\rho']_{B'}$. By construction, we have
  \begin{equation}
    \rho=VCV^\dagger,
    \qquad
    \rho'=V'CV'^\dagger.
    \label{eq:net-state-approximation-same-coefficients}
  \end{equation}
  Since $\lVert C\rVert_F\leq1$, we have
  \begin{equation}
    \begin{aligned}
      \lVert\rho'-\rho\rVert_F
       & \leq\lVert(V'-V)CV'^\dagger\rVert_F
      +\lVert VC(V'^\dagger-V^\dagger)\rVert_F                \\
       & \leq2\lVert V'-V\rVert_{\mathrm{op}}\lVert C\rVert_F \\
       & \leq2\lVert V'-V\rVert_{\mathrm{op}}
    \end{aligned}
    \label{eq:net-state-approximation-conjugation-bound}
  \end{equation}
  Combining this with Eq.~\eqref{eq:net-state-approximation-product-unitary-error} concludes the proof.
\end{proof}

\paragraph{Returning output with resource bounds.}

We are now ready to complete the proof of Theorem~\ref{thm:tree-continuous-basis}. Lemma~\ref{lem:projection-error} and Eq.~\eqref{eq:merge-temporary-leakage} bound the distance between the estimated projected state $\widehat \sigma^{(D)}_1$ and the input state $\rho$ with tomography and leakage errors.
\begin{equation}
  \begin{aligned}
    \left\|\widehat\sigma_1^{(D)}-\rho\right\|_F
     & \leq
    \left\|\widehat\sigma_1^{(D)}-A_1^{(D)}\rho A_1^{(D)}\right\|_F
    +\left\|A_1^{(D)}\rho A_1^{(D)}-\rho\right\|_F \\
     & \leq\varepsilon_1
    +\sqrt{2\operatorname{Tr}[(I-A_1^{(D)})\rho]}  \\
     & \leq\varepsilon_1+2\sqrt{\eta_{D-1}}.
  \end{aligned}
  \label{eq:tree-continuous-basis-root-tomography}
\end{equation}
Due to Eq.~\eqref{eq:tree-continuous-basis-tree-tolerances} and $\varepsilon_1\leq\varepsilon/(512\sqrt{s})$, we have
\begin{equation}
  \left\|\widehat\sigma_1^{(D)}-\rho\right\|_F
  <
  \left(
  \frac1{512}+2\sqrt{\frac{19}{1344}}
  \right)\frac{\varepsilon}{\sqrt{s}}.
  \label{eq:tree-continuous-basis-root-tomography-error}
\end{equation}

We next account for the error due to compressing the sparsity of $\widehat\sigma_1^{(D)}$ by choosing an epsilon-net of product bases and truncating the coefficients. Lemma~\ref{lem:net-state-approximation} ensures that there exists a state $\rho'$ that is $s$-sparse in some basis $B'\in\mathcal{N}(\varepsilon_2)^n$ such that
\begin{equation}
  \left\lVert\rho'-\rho\right\rVert_F
  \le
  4n\sin\frac{\varepsilon_2}{4}
  \le n\varepsilon_2 = \frac{\varepsilon^2}{64sn^2}.
  \label{eq:tree-continuous-basis-reconstruction-targets}
\end{equation}
Recall that in each basis in $\mathcal{N}(\varepsilon_2)^n$, we retain the at most $s$ entries of the final merge estimate having largest magnitude and set all other entries to zero. We accept the basis $\widehat B$ and the resulting truncated state $\widehat\rho$ if
\begin{equation}
  \|\widehat\rho-\widehat\sigma_1^{(D)}\|_F
  \le \frac{\varepsilon}{4\sqrt{s}}.
  \label{eq:tree-continuous-basis-reconstruction-threshold}
\end{equation}
We can see that there exists such a basis. Specifically, let $\widehat\rho_{B'}$ be the truncated state in the basis $B'$ that we obtain from $\widehat\sigma_1^{(D)}$ by retaining the at most $s$ entries with largest magnitude and setting all other entries to zero. Then,
\begin{equation}
  \begin{aligned}
    \left\|\widehat\rho_{B'}-\widehat\sigma_1^{(D)}\right\|_F
     & \leq
    \left\|\rho'-\widehat\sigma_1^{(D)}\right\|_F                                    \\
     & \leq \left\|\rho'-\rho\right\|_F+\left\|\widehat\sigma_1^{(D)}-\rho\right\|_F \\
     & <\frac{\varepsilon^2}{64sn^2}
    +\left(\frac1{512}+2\sqrt{\frac{19}{1344}}\right)
    \frac{\varepsilon}{\sqrt{s}}                                                     \\
     & <\frac{\varepsilon}{4\sqrt{s}},
  \end{aligned}
\end{equation}
where the last inequality uses $0<\varepsilon\leq1$, $s\geq2$, and $n\geq2$ without loss of generality.

Therefore, Algorithm~2 outputs $\widehat{\rho}$ that is $s$-sparse in some basis $\widehat B$ satisfying Eq.~\eqref{eq:tree-continuous-basis-reconstruction-threshold}. Combining this and Eq.~\eqref{eq:tree-continuous-basis-root-tomography-error} gives
\begin{equation}
  \begin{aligned}
    \lVert\widehat\rho-\rho\rVert_F
     & \leq
    \left\|\widehat\rho-\widehat\sigma_1^{(D)}\right\|_F
    +\left\|\widehat\sigma_1^{(D)}-\rho\right\|_F \\
     & \leq \left(
    \frac14+\frac1{512}
    +2\sqrt{\frac{19}{1344}}
    \right)\frac{\varepsilon}{\sqrt{s}}           \\
     & <\frac{\varepsilon}{\sqrt{2s}}.
  \end{aligned}
  \label{eq:tree-continuous-basis-frobenius-closure}
\end{equation}
Finally, since $\rho$ and $\widehat\rho$ have rank at most $s$, the difference $\rho-\widehat\rho$ has rank at most $2s$, so the Cauchy-Schwarz inequality gives
\begin{equation}
  \lVert\widehat\rho-\rho\rVert_1
  \leq\sqrt{2s}\lVert\widehat\rho-\rho\rVert_F
  \leq\varepsilon.
  \label{eq:tree-continuous-basis-trace-closure}
\end{equation}

We now count the resources. At a merge, $\operatorname{rank}(A_j^{(d+1)})\leq s^2$.
Corollary~\ref{cor:selected-submatrix-tomography} therefore returns a Hermitian
estimate of the selected basis submatrix with Frobenius error $\varepsilon_1$ using
\begin{equation}
  O\!\left(
  s^8\varepsilon_1^{-2}\log\frac{ns}{\delta}
  \right)
  \label{eq:tree-continuous-basis-per-level-copies}
\end{equation}
copies of $\rho$ for one node. Since the active parents at one level occupy disjoint blocks, the same fresh copies serve every node at that level. Substituting the fixed value
$\varepsilon_1=\varepsilon^2/(64s^2n^3)$ from Eq.~\eqref{eq:algorithm1-merge-parameters} into
Eq.~\eqref{eq:tree-continuous-basis-per-level-copies} gives
$O(n^6s^{12}\varepsilon^{-4}\log(ns/\delta))$ copies of $\rho$ per level. There are
$O(\log n)$ levels, so we consume
\begin{equation}
  O\left(
  n^6s^{12}\varepsilon^{-4}\log n\log(ns/\delta)\right)
\end{equation}
copies of $\rho$ in total. Note that the tree has $K_0-1 \le n$ tomography steps to merge all $K_0$ leaves in total. Therefore, by a union bound, all tomography steps succeed with probability at least $1 - \delta_1 (K_0 -1) \ge 1-\delta$.

For the classical runtime, the basis search space is the product of $\mathcal N(\varepsilon_2)$. $\mathcal N(\varepsilon_2)$ has size at most $(C/\varepsilon_2)^2$ with a universal constant $C$, which gives
\begin{equation}
  \left\lvert\mathcal N(\varepsilon_2)^n\right\rvert
  \leq\left(\frac{C}{\varepsilon_2}\right)^{2n}
  =\left(\frac{64Csn^3}{\varepsilon^2}\right)^{2n}
  =\left(\frac{sn}{\varepsilon}\right)^{O(n)}.
\end{equation}
All other computations, including processing the tomography result and compressing the sparsity of the matrix, are absorbed by $(sn/\varepsilon)^{O(n)}$, except the factor $\operatorname{polylog}(\delta^{-1})$ which comes from processing the tomography data. Thus, the total classical runtime is
\begin{equation}
  \left(\frac{sn}{\varepsilon}\right)^{O(n)} \operatorname{polylog}(\delta^{-1}).
\end{equation}
This completes the proof of Theorem~\ref{thm:tree-continuous-basis}.

\subsection{Known finite set of single-qubit bases (Theorem~\ref{thm:tree-finite-bases})}
\label{sec:finite-set-result}

We now prove Theorem~\ref{thm:tree-finite-bases}, where the input state $\rho$ is $s$-sparse in a product basis drawn from $\mathcal A^n$ for the known finite set $\mathcal A$. While we search over the epsilon-net when performing the sparsity compression step (Step~2(b) and Step~3) for the arbitrary basis case, we instead search over a much smaller set of product bases in $\mathcal A^n$. We use the Donoho--Stark uncertainty principle to show that the search space is only $n^{O(\log s)}$. Apart from this basis searching, the rest of the algorithm is the same as the arbitrary basis case.

\paragraph{Search space for finite set of bases}
\label{sec:finite-set-search-space}

Consider a merge in Step~2 from $L:=\mathcal B_{2j-1}^{(d)}$ and $R:=\mathcal B_{2j}^{(d)}$ to $Q:=\mathcal B_j^{(d+1)}$, using the notation of Section~\ref{sec:ideal-merge}. Let $B=B_L \times B_R \in \mathcal A^{|Q|}$ be the temporary product basis from the children of $Q$, and denote $B=(B_1,\dots,B_{|Q|})$. We set
\begin{equation}
  H:=\left\lceil
  \frac{5\log(s)+2\log(4/3)}{4\kappa}
  \right\rceil,
  \label{eq:tree-finite-bases-radius}
\end{equation}
and define the basis search space for this merge as
\begin{equation}
  \mathcal L_j^{(d+1)}
  :=\left\{
  B'=(B_1',\dots,B_{|Q|}')\in\mathcal A^{|Q|}:
  \left\lvert\{i\in Q:B_i'\ne B_i\}\right\rvert\leq H
  \right\}.
  \label{eq:tree-finite-bases-search-space}
\end{equation}
Note that the number of product bases in this search space is at most
\begin{equation}
  G:=\sum_{h=0}^{\min\{n,H\}}
  \binom nh(\lvert\mathcal A\rvert-1)^h
  \leq(H+1)(n\lvert\mathcal A\rvert)^H.
  \label{eq:tree-finite-bases-radius-count}
\end{equation}
In the rest of this section, we show that this small search space is enough to find a good basis in Step~2(b).

\paragraph{Donoho--Stark uncertainty principle}
\label{sec:donoho-stark-uncertainty}

The Donoho--Stark uncertainty principle~\cite{donohoUncertaintyPrinciplesSignal1989,eladGeneralizedUncertaintyPrinciple2002,boggiattoTwoAspectsDonoho2016} is a fundamental result in signal processing that bounds the size of the support of a signal in two orthonormal bases. Adapting Theorem~2 of Ref.~\cite{boggiattoTwoAspectsDonoho2016}, we provide a statement of the uncertainty principle with a self-contained proof.

\begin{theorem}[Donoho--Stark uncertainty principle]
  \label{thm:donoho-stark}
  Let $U\in\mathbb C^{N\times N}$ be unitary with $\lVert U\rVert_{\max}:=\max_{a,b}\lvert U_{ab}\rvert$, and let $v\in\mathbb C^N$ satisfy $\lVert v\rVert_2=1$. Let $S, T \subseteq\{1,\ldots,N\}$, and denote $P_S$ (resp. $P_T$) as the orthogonal projector onto $\operatorname{span}\{\mathbf e_i:i\in S\}$ (resp. $\operatorname{span}\{\mathbf e_i:i\in T\}$), where $\mathbf e_1,\ldots, \mathbf e_N$ are the standard basis vectors of $\mathbb C^N$. Suppose
  \begin{equation}
    \lVert v-P_S v\rVert_2\leq e_S,
    \qquad
    \lVert Uv-P_T Uv\rVert_2\leq e_T,
    \qquad
    e_S+e_T<1.
    \label{eq:donoho-stark-tail-assumptions}
  \end{equation}
  Then
  \begin{equation}
    \lvert S\rvert\lvert T\rvert
    \geq\frac{(1-e_S-e_T)^2}{\lVert U\rVert_{\max}^2}.
    \label{eq:donoho-stark-uncertainty}
  \end{equation}
\end{theorem}

\begin{proof}
  The two tail assumptions motivate the restricted comparison
  \begin{equation}
    \begin{aligned}
      \lVert v-P_S U^\dagger P_T Uv\rVert_2
       & \leq\lVert v-P_S v\rVert_2 +\lVert P_S U^\dagger (Uv-P_T Uv)\rVert_2 \\
       & \leq e_S+e_T.
    \end{aligned}
    \label{eq:donoho-stark-restricted-approximation}
  \end{equation}
  The reverse triangle inequality therefore gives
  \begin{equation}
    1-e_S-e_T
    \leq\lVert P_S U^\dagger P_T Uv\rVert_2.
    \label{eq:donoho-stark-reverse-triangle}
  \end{equation}
  Now, note that the matrix $P_S U^\dagger P_T$ has at most $\lvert S\rvert\lvert T\rvert$ nonzero entries, so
  \begin{equation}
    \lVert P_S U^\dagger P_T\rVert_F \leq\lVert U\rVert_{\max}\sqrt{\lvert S\rvert\lvert T\rvert}.
  \end{equation}
  Therefore, we have
  \begin{equation}
    \lVert P_S U^\dagger P_T Uv\rVert_2\leq\lVert P_S U^\dagger P_T\rVert_{\mathrm{op}}\leq\lVert P_S U^\dagger P_T\rVert_F \leq\lVert U\rVert_{\max}\sqrt{\lvert S\rvert\lvert T\rvert}.
    \label{eq:donoho-stark-restricted-transform}
  \end{equation}
  Combining Eqs.~\eqref{eq:donoho-stark-reverse-triangle} and~\eqref{eq:donoho-stark-restricted-transform} concludes the proof.
\end{proof}

As a simple corollary of Theorem~\ref{thm:donoho-stark}, we provide a matrix version of the uncertainty principle.

\begin{corollary}[Matrix Donoho--Stark uncertainty principle]
  \label{cor:matrix-donoho-stark}
  Let $M$ and $N$ be nonzero square matrices of the same size, and suppose
  \begin{equation}
    N=VMV^\dagger
  \end{equation}
  for a unitary matrix $V$ of the same size. Suppose $M'$ (resp. $N'$) is obtained from $M$ (resp. $N$) by keeping at most $s$ (resp. $t$) entries, such that
  \begin{equation}
    \lVert M-M'\rVert_F \leq e_S \lVert M\rVert_F,
    \qquad
    \lVert N-N'\rVert_F \leq e_T \lVert N\rVert_F,
  \end{equation}
  where $e_S+e_T<1$. Define $\lVert V\rVert_{\max}:=\max_{a,b}\lvert V_{ab}\rvert$. Then
  \begin{equation}
    \sqrt{st}\geq
    \frac{1-e_S-e_T}{\lVert V\rVert_{\max}^2}.
    \label{eq:matrix-donoho-stark}
  \end{equation}
\end{corollary}

\begin{proof}
  Let $\operatorname{vec}$ denote vectorization of a square matrix, i.e., for an $N\times N$ matrix $A$,
  \begin{equation}
    \operatorname{vec}(A)\coloneqq(A_{11},\dots,A_{1N},A_{21},\dots,A_{NN})\in\mathbb C^{N^2}.
  \end{equation}
  Let $v:=\operatorname{vec}(M)/\lVert M\rVert_F$. Since $\lVert M\rVert_F = \lVert VMV^\dagger\rVert_F$, we have
  \begin{equation}
    \frac{\operatorname{vec}(N)}{\lVert N\rVert_F}
    =\left(V\otimes V^*\right)v.
  \end{equation}
  Let $S$ and $T$ be the supports of $\operatorname{vec}(M')$ and $\operatorname{vec}(N')$, respectively. Then
  $\lvert S\rvert\leq s$,
  $\lvert T\rvert\leq t$, and
  \begin{equation}
    \lVert v-P_Sv\rVert_2\leq e_S,
    \qquad
    \left\lVert
    (V\otimes V^*)v-P_T(V\otimes V^*)v
    \right\rVert_2\leq e_T.
  \end{equation}
  Moreover,
  $V\otimes V^*$ is unitary and
  \begin{equation}
    \left\lVert V\otimes V^*\right\rVert_{\max}
    =\lVert V\rVert_{\max}^2.
  \end{equation}
  Applying Theorem~\ref{thm:donoho-stark} to $S$ and $T$ yields
  \begin{equation}
    st\geq\lvert S\rvert\lvert T\rvert
    \geq\frac{(1-e_S-e_T)^2}{\lVert V\rVert_{\max}^4}.
  \end{equation}
\end{proof}

\paragraph{Leakage with small search space.}
\label{sec:finite-basis-leakage}

Now we use the Donoho--Stark uncertainty principle to show that the search space $\mathcal L_j^{(d+1)}$ defined in Eq.~\eqref{eq:tree-finite-bases-search-space} contains a good enough basis so that the leakage can be well controlled throughout the recursion.

\begin{lemma}
  \label{lem:finite-basis-leakage}
  Suppose that every tomography result in Step~2(a) satisfies
  \begin{equation}
    \left\lVert
    \widehat\sigma_j^{(d+1)}-\sigma_j^{(d+1)}
    \right\rVert_F\leq\varepsilon_1,
    \label{eq:finite-basis-leakage-tomography-assumption}
  \end{equation}
  for all $d\le D-1$, and all sparsity compression steps in Step~2(b) are performed over the search space $\mathcal L_j^{(d+1)}$. Define the level-wise leakage bounds recursively
  \begin{equation}
    \eta_0:=0,
    \qquad
    \eta_{d+1}:=8\eta_d+4s\varepsilon_1,
    \qquad 0\leq d<D-1.
    \label{eq:finite-basis-leakage-leakage-recursion}
  \end{equation}
  Then the projector returned by Step~2(b) obeys
  \begin{equation}
    \operatorname{Tr}[(I-\widehat P_j^{(d+1)})\rho_Q]
    \leq8\eta_d+4s\varepsilon_1
    =\eta_{d+1}.
    \label{eq:finite-basis-leakage-recurrence}
  \end{equation}
\end{lemma}

\begin{proof}
  Solving the recursion in Eq.~\eqref{eq:finite-basis-leakage-leakage-recursion} gives
  \begin{equation}
    \eta_d=\frac{4s(8^d-1)}7\varepsilon_1,
    \qquad 0\leq d\leq D-1.
    \label{eq:finite-basis-leakage-root-iteration}
  \end{equation}
  It suffices to consider $0<\varepsilon\leq1$. For every
  $0\leq d\leq D-1$, Eq.~\eqref{eq:finite-basis-leakage-root-iteration}, together with
  $8^d\leq8^{D-1}<K_0^3$ and $K_0\leq n$, gives
  \begin{equation}
    \eta_d
    <\frac{4}{7}sn^3\varepsilon_1
    =\frac{\varepsilon^2}{112s}
    \label{eq:tree-finite-bases-tree-tolerance}
  \end{equation}

  We show Eq.~\eqref{eq:finite-basis-leakage-recurrence} by induction. For $d= 0, \dots, d'-1$, let Eq.~\eqref{eq:finite-basis-leakage-recurrence} hold. Consider a merge on $Q=\mathcal{B}_j^{(d'+1)}$ from child blocks $L$ and $R$ that are at level $d'$. Let $B_\star\in\mathcal A^{|Q|}$ be the hidden basis in which $\rho_Q$ is $s$-sparse. Let $A:= P_L\otimes P_R$ be the temporary projector for $\rho_Q$ based on the two child projectors. The induction hypothesis gives
  \begin{equation}
    \operatorname{Tr}[(I-A)\rho_Q]
    \leq \operatorname{Tr}[(I-P_L)\rho_L]
    +\operatorname{Tr}[(I-P_R)\rho_R]
    \leq2\eta_{d'}.
    \label{eq:tree-finite-bases-temporary-leakage}
  \end{equation}
  Since $A$ retains at most $s^2$ bitstrings in $B$, the matrix
  $[A\rho_QA]_B$ has at most $s^4$ nonzero entries. Then Lemma~\ref{lem:projection-error} gives
  \begin{equation}
    \lVert\rho_Q-A\rho_QA\rVert_F \leq 2\sqrt{\eta_{d'}}.
    \label{eq:tree-finite-bases-current-basis-tail}
  \end{equation}
  Moreover, $\operatorname{rank}(\rho_Q)\leq s$ and
  $\operatorname{Tr}\rho_Q=1$, so Cauchy--Schwarz applied to the nonzero
  eigenvalues of $\rho_Q$ gives
  \begin{equation}
    \lVert\rho_Q\rVert_F\geq\frac1{\sqrt{s}}.
    \label{eq:tree-finite-bases-state-frobenius}
  \end{equation}
  Therefore, $[A\rho_QA]_B$ is an $s^4$-sparse approximation to $[\rho_Q]_B$ with error
  \begin{equation}
    e_T:=\frac{\lVert\rho_Q-A\rho_QA\rVert_F}{\lVert\rho_Q\rVert_F}
    \leq2\sqrt{s\eta_{d'}} \le \frac{\epsilon}{\sqrt{28}}
    \label{eq:tree-finite-bases-relative-tail}
  \end{equation}

  Suppose that $B$ and $B_{\star}$ differ on $h$ qubits. For each $i\in Q$,
  define the one-qubit change-of-basis unitary $V_i$ by
  \begin{equation}
    (V_i)_{u,v}:=\langle u_{B_i}\mid v_{B_{\star,i}}\rangle,
    \qquad u,v\in\{0,1\},
  \end{equation}
  and set
  \begin{equation}
    V:=\bigotimes_{i\in Q}V_i.
  \end{equation}
  Then
  \begin{equation}
    [\rho_Q]_B=V[\rho_Q]_{B_{\star}}V^\dagger.
  \end{equation}
  If $B_i=B_{\star,i}$, then $\lVert V_i\rVert_{\max}=1$; otherwise,
  $\lVert V_i\rVert_{\max}\leq\mu$. Hence
  \begin{equation}
    \lVert V\rVert_{\max}
    =\prod_{i\in Q}\lVert V_i\rVert_{\max}
    \leq\mu^h.
  \end{equation}
  Apply Corollary~\ref{cor:matrix-donoho-stark} to
  $[\rho_Q]_{B_{\star}}$ and $[\rho_Q]_B$. The first matrix is $s$-sparse,
  so $e_S=0$, while $A\rho_QA$ gives the $s^4$-entry approximation with error
  $e_T$ in the second matrix. Therefore,
  \begin{equation}
    1-e_T\leq\mu^{2h}s^{5/2}.
    \label{eq:donoho-stark-matrix-uncertainty}
  \end{equation}
  Taking logarithms and using $\kappa=-\log\mu$ gives
  \begin{equation}
    h
    \leq\frac{5\log s-2\log(1-e_T)}{4\kappa}
    \leq H.
    \label{eq:donoho-stark-localization}
  \end{equation}
  Hence $B_{\star}\in\mathcal L_j^{(d+1)}$.

  Let $P_{j,\star}^{(d+1)}$ be the rank-at-most-$s$ sparse-support projector of $\rho_Q$ in $B_{\star}$. Since $B_{\star}\in\mathcal L_j^{(d+1)}$, we can take $\overline P = P_{j,\star}^{(d+1)}$ in Lemma~\ref{lem:generic-one-merge-stability}, and applying Lemma~\ref{lem:generic-one-merge-stability} with $\xi=0$ gives
  \begin{equation}
    \operatorname{Tr}[(I-\widehat P_j^{(d+1)})\rho_Q]
    \leq8\eta_{d'}+4s\varepsilon_1
    =\eta_{d'+1}.
  \end{equation}
\end{proof}

\paragraph{Returning output with resource bounds.}
\label{sec:finite-basis-reconstruction}

We are now ready to complete the proof of Theorem~\ref{thm:tree-finite-bases}. We first bound the distance between the input state $\rho$ and the estimated projected state $\widehat\sigma_1^{(D)}$ obtained from the end of Step~2. To this end, we denote $A:=A_1^{(D)}$. Using
\begin{equation}
  \eta_{D-1}=\frac{4s(8^{D-1}-1)}{7}\varepsilon_1 \le \frac{\epsilon^2}{112s},
\end{equation}
and Lemma~\ref{lem:projection-error}, we have
\begin{equation}
  \begin{aligned}
    \left\|\widehat\sigma_1^{(D)}-\rho\right\|_F
     & \leq
    \left\|\widehat\sigma_1^{(D)}-A\rho A\right\|_F
    +\left\|A\rho A-\rho\right\|_F                                                           \\
     & \leq\varepsilon_1
    +\sqrt{2\operatorname{Tr}[(I-A)\rho]}                                                    \\
     & \leq\varepsilon_1+2\sqrt{\eta_{D-1}}                                                  \\
     & \leq\left( \frac{1}{512} + \frac{2}{\sqrt{112}} \right) \frac{\varepsilon}{\sqrt{s}},
  \end{aligned}
  \label{eq:finite-basis-reconstruction-common-residual}
\end{equation}
where we used $\varepsilon_1\leq\varepsilon/(512\sqrt{s})$.

We next account for the error due to compressing the sparsity of $\widehat\sigma_1^{(D)}$ over the final search space $\mathcal L_1^{(D)}$. Let $B_\star \in \mathcal A^n$ be the basis in which $\rho$ is $s$-sparse. Then, with the same logic in the proof of Lemma~\ref{lem:finite-basis-leakage}, we can show $B_\star\in\mathcal{L}_1^{(D)}$. Recall that in each basis in $\mathcal{L}_1^{(D)}$, we retain the at most $s$ entries of the final merge estimate having largest magnitude and set all other entries to zero. We accept the basis $\widehat B$ and the resulting truncated state $\widehat\rho$ if
\begin{equation}
  \|\widehat\rho-\widehat\sigma_1^{(D)}\|_F
  \le \frac{\varepsilon}{4\sqrt{s}}.
\end{equation}
Let $\widehat\rho_{B_\star}$ be the matrix obtained by retaining the entries of $\widehat\sigma_1^{(D)}$ in $B_\star$ with largest magnitude and setting all other entries to zero. Since $\rho$ has at most $s$ nonzero entries in $B_\star$,
\begin{equation}
  \begin{aligned}
    \left\|\widehat\rho_{B_\star}-\widehat\sigma_1^{(D)}\right\|_F
     & \leq
    \left\|\rho-\widehat\sigma_1^{(D)}\right\|_F \\
     & <
    \left(\frac1{512}+\frac2{\sqrt{112}}\right)
    \frac{\varepsilon}{\sqrt{s}}                 \\
     & <\frac{\varepsilon}{4\sqrt{s}}.
  \end{aligned}
  \label{eq:finite-basis-reconstruction-residual-optimality}
\end{equation}
Thus Step~3 accepts at least one basis and Algorithm~2 outputs $\widehat B$ along with $\widehat\rho$ that has at most $s$ entries in $\widehat B$. The output $\widehat\rho$ satisfies $\lVert\rho-\widehat\rho\rVert_F \leq \varepsilon/(4\sqrt{s})$, so Eq.~\eqref{eq:finite-basis-reconstruction-common-residual} gives
\begin{equation}
  \begin{aligned}
    \lVert\widehat\rho-\rho\rVert_F
     & \leq
    \left\|\widehat\rho-\widehat\sigma_1^{(D)}\right\|_F
    +\left\|\widehat\sigma_1^{(D)}-\rho\right\|_F \\
     & <\left(
    \frac14+\frac1{512}+\frac2{\sqrt{112}}
    \right)\frac{\varepsilon}{\sqrt{s}}           \\
     & <\frac{\varepsilon}{\sqrt{2s}}.
  \end{aligned}
  \label{eq:finite-basis-reconstruction-frobenius-closure}
\end{equation}
Finally, since $\rho$ and $\widehat\rho$ have rank at most $s$, the difference $\rho-\widehat\rho$ has rank at most $2s$, so the Cauchy--Schwarz inequality gives
\begin{equation}
  \lVert\widehat\rho-\rho\rVert_1
  \leq\sqrt{2s}\lVert\widehat\rho-\rho\rVert_F
  <\varepsilon.
  \label{eq:finite-basis-reconstruction-trace-closure}
\end{equation}

We now count the resources. Recall that all steps except Step~2(b) and Step~3 are identical to Algorithm~2 for the arbitrary-product-basis case. Thus, the number of copies of $\rho$ and success probability for selected-entry tomography are the same as in Section~\ref{sec:polynomial-copy-baseline}.

For the classical runtime, Eq.~\eqref{eq:tree-finite-bases-radius-count} gives at most $(H+1)(n\lvert\mathcal A\rvert)^H$ product bases in the search space at each merge. Since $H=O(\log(s)/\kappa)$, we have
\begin{equation}
  \lvert\mathcal{L}_j^{(d+1)}\rvert=\left(n\lvert\mathcal A\rvert\right)^{O(\log s/\kappa)}.
\end{equation}
For each basis in $\mathcal{L}_j^{(d+1)}$, we change at most $H$ local bases of a given $s^4$-sparse matrix and each element in the matrix can turn into at most $4^H$ matrix entries, which yields up to
\begin{equation}
  s^4 4^{O(\log s/\kappa)}
\end{equation}
runtime for each basis in $\mathcal{L}_j^{(d+1)}$. Runtime for processing tomography results for each merge is polynomial in $n$, $1/\varepsilon$, and $\log \delta^{-1}$. Since each tomography sample takes $O(n)$ time to process, tomography and sample processing contribute the factor $\varepsilon_1^{-2}\propto\varepsilon^{-4}$, while the finite-basis search is independent of $\varepsilon$. Therefore, accounting for all classical computations, the total runtime is
\begin{equation}
  \left(n\lvert\mathcal A\rvert\right)^{O(\log s/\kappa)}
  \varepsilon^{-4}\operatorname{polylog}\!\left(\delta^{-1}\right).
  \label{eq:tree-finite-bases-runtime-derived}
\end{equation}
This completes the proof of Theorem~\ref{thm:tree-finite-bases}.

%% file: 07_sparsity_certification.tex
\section{Testing sparsity of arbitrary quantum states.}
\label{sec:sparsity-certification}

So far, we assumed that the input state $\rho$ is $s$-sparse in an unknown product basis. In this section, we relax this assumption and show that one can use Algorithm~1 to test whether an arbitrary input state $\rho$ is approximately $s$-sparse or far from every $s$-sparse state. Specifically, when $s$ is a fixed constant, we show that a moderate modification of Algorithm~1 can be used to output
\begin{itemize}
  \item $\mathsf{YES}$ if there exists an $s$-sparse operator $M$ such that $\lVert\rho-M\rVert_1\le c_s\varepsilon^4/n$ for some sufficiently small constant $c_s>0$ depending on $s$, and
  \item $\mathsf{NO}$ if for all $s$-sparse operators $M$, we have $\lVert\rho-M\rVert_1 > \varepsilon$,
\end{itemize}
with high probability. Here, the sample and classical computational complexities are the same as those of Algorithm~1.

\subsection{From sparse operators to sparse states}

For proving the above result, we first show that if $\rho$ can be approximated by a sparse operator, which we do not require to be a valid quantum state, then it can also be approximated by a sparse quantum state. To this end, we denote the class of all $s$-sparse operators as $\mathcal M_s$
\begin{equation}
  \mathcal M_s:=\left\{
    M\in\mathbb C^{2^n\times2^n}:
    \text{There exists a product basis $B$ such that $[M]_B$ is an $s$-sparse matrix}
  \right\}.
  \label{eq:sparse-certification-classes}
\end{equation}
Here, we do not require $M$ to be Hermitian, positive semidefinite, or of unit trace. We then denote the class of all physical $s$-sparse states as
\begin{equation}
  \mathcal S_s:=\left\{
    \sigma\in\mathcal M_s:
    \sigma\succeq0,\ \operatorname{Tr}\sigma=1
  \right\}.
  \label{eq:sparse-certification-physical-classes}
\end{equation}
We show that if $\rho$ is close to some operator $M$ in $\mathcal M_s$, then it is also close to some state $\sigma$ in $\mathcal S_s$.

\begin{lemma}[Sparse operator to sparse state]
  \label{lem:sparse-matrix-physical-repair}
  Let $\rho$ be an $n$-qubit density matrix and suppose that there exists an operator $M \in \mathcal M_s$ such that $\lVert\rho-M\rVert_1\le \varepsilon < 1$. Then, there is a density matrix $\sigma \in \mathcal S_s$ such that
  \begin{equation}
    \lVert\rho-\sigma\rVert_1\le2\sqrt \varepsilon+7s\varepsilon.
    \label{eq:sparse-matrix-physical-repair}
  \end{equation}
\end{lemma}

\begin{proof}
  Fix a product basis in which $M$ is $s$-sparse. Here, without loss of generality, we assume that the basis is the standard computational basis, and express all matrix entries below in this basis. We will construct $\sigma$ from $\rho$ by first restricting to a small subspace, then removing entries outside the support of $M$, and finally restoring positivity and unit trace.

  First, let $R:=\{x:M_{xx}\ne0\}$, and let $P$ be a projector onto the basis vectors indexed by $R$, i.e., $P=\sum_{x\in R}|x\rangle\langle x|$. We denote $r:=|R|\le s$. Since $\operatorname{Tr}[(I-P)M]=0$, the leakage error of $\rho$ outside this subspace therefore satisfies
  \begin{equation}
    \eta:=\operatorname{Tr}[(I-P)\rho]
    =\left|\operatorname{Tr}[(I-P)(\rho-M)]\right|
    \le\lVert I-P\rVert_{\mathrm{op}}\lVert\rho-M\rVert_1
    \le\varepsilon.
  \label{eq:sparse-repair-discarded-mass}
  \end{equation}
  Define $A:=P\rho P$. This operator is positive semidefinite and has $\operatorname{Tr}A=1-\eta\ge1-\varepsilon>0$. Moreover, Lemma~\ref{lem:projection-error-trace} gives
  \begin{equation*}
    \lVert\rho-A\rVert_1\le2\sqrt \eta\le2\sqrt\varepsilon.
  \end{equation*}

  Note here that $A$ acts on $r$ basis vectors, but it may still contain $r^2$ nonzero entries. We reduce its entry support by defining
  \begin{equation*}
    B_{xy}:=
    \begin{cases}
      A_{xy},&\text{if $M_{xy}\ne0$ and $M_{yx}\ne0$},\\
      0,&\text{otherwise}.
    \end{cases}
  \end{equation*}
  Removing off-diagonal entries in conjugate pairs ensures that $B$ is Hermitian. Its support is contained in that of $M$, and all diagonal entries of $A$ remain the same. Therefore, $B$ is an $s$-sparse matrix and $\operatorname{Tr}B=\operatorname{Tr}A=1-\eta$. We bound the error caused by removing these entries.
  For each discarded entry $\rho_{xy}$ for $(x,y) \in R\times R$, either $M_{xy}=0$ or $M_{yx}=0$. Since $A_{xy}=\rho_{xy}$ and $|\rho_{xy}|=|\rho_{yx}|$, we have
  \begin{equation}
    |\rho_{xy}|^2 \le 
    |\rho_{xy}-M_{xy}|^2+|\rho_{yx}-M_{yx}|^2,
  \end{equation}
  for all discarded entries $(x,y)$. Retained entries contribute zero, so summing gives
  \begin{equation}
    \begin{aligned}
      \lVert B-A\rVert_F^2
      &\le\sum_{x,y\in R}
        \left(|\rho_{xy}-M_{xy}|^2+|\rho_{yx}-M_{yx}|^2\right)\\
      &\le2\lVert\rho-M\rVert_F^2
      \le2\lVert\rho-M\rVert_1^2
      \le2\varepsilon^2.
    \end{aligned}
    \label{eq:sparse-repair-mask-error}
  \end{equation}
  Note that $B$ need not be positive semidefinite.

  We now restore the positive semi-definiteness by setting $C:=B+\sqrt2\varepsilon P$. Indeed, $B-A$ is Hermitian and supported on the range of $P$. Also, we have
  \begin{equation}
    \lVert B-A\rVert_{\mathrm{op}}
    \le\lVert B-A\rVert_F
    \le\sqrt2\varepsilon,
  \end{equation}
  Hence $B-A\succeq-\sqrt2\varepsilon P$, and consequently $C=B+\sqrt2\varepsilon P\succeq A\succeq0$.
  Its trace is
  \begin{equation}
    t:=\operatorname{Tr}C=1-\eta+\sqrt2\varepsilon r
    \ge1-\varepsilon>0.
  \end{equation}
  Here, note that
  \begin{equation}
    \lVert C-B\rVert_1
    =\sqrt2\varepsilon\operatorname{Tr}P
    =\sqrt2\varepsilon r.
  \end{equation}
  The shift $C=B+\sqrt2\varepsilon P$ adds entries only at positions $(x,x)$ with $x\in R$, while $C_{xx}$ is already positive for all $x \in R$. Therefore, $C$ is $s$-sparse and positive semi-definite.

  Finally, we normalize $C$ to obtain a density matrix $\sigma:=C/t \in \mathcal S_s$. Since $C\succeq0$, we have $\lVert C\rVert_1=\operatorname{Tr}C=t$, and thus
  \begin{equation}
    \lVert\sigma-C\rVert_1
    = \lVert(1/t-1)C\rVert_1
    =|1-t|
    \le \eta+\sqrt2\varepsilon r \le\varepsilon+\sqrt2\varepsilon r,
  \end{equation}
  where we used $0\le \eta\le\varepsilon$. 

  Now, putting everything together, we have
  \begin{equation}
    \begin{aligned}
      \lVert\rho-\sigma\rVert_1
      &\le\lVert\rho-A\rVert_1+\lVert A-B\rVert_1
        +\lVert B-C\rVert_1+\lVert C-\sigma\rVert_1\\
      &\le2\sqrt\varepsilon+(1+3\sqrt2)s\varepsilon\\
      &\le2\sqrt\varepsilon+7s\varepsilon.
    \end{aligned}
  \end{equation}
\end{proof}

\subsection{Testing sparsity with Algorithm~1}

For testing sparsity with Algorithm~1, we run Algorithm~1 with accuracy $\varepsilon$ and failure probability $\delta$ with a mild modification. Here, we refine only Step~1(a): we estimate each initial Pauli moment to statistical error at most $\varepsilon_1/2$ instead of $\varepsilon_1$. We keep other procedures unchanged. Finally, we output $\mathsf{YES}$ if any candidate basis passes the tests in Steps~2 and~3, and output $\mathsf{NO}$ otherwise. Then, defining $\operatorname{dist}_1(\rho,\mathcal M_s):=\inf_{M\in\mathcal M_s}\lVert\rho-M\rVert_1$, we have the following Theorem.

\begin{theorem}[Testing sparsity with Algorithm~1]
  \label{cor:algorithm2-matrix-gap-tester}
  Let $0<\varepsilon,\delta<1$, and let $\rho$ be an arbitrary $n$-qubit state. For a fixed $s=O(1)$, the above procedure uses only single-qubit measurements and requires
  \begin{equation}
    \begin{aligned}
      N&=\poly\!\left(\frac n\varepsilon\right)\polylog\!\left(\frac1\delta\right)
      &&\text{copies of $\rho$, and}\\
      T&=\poly\!\left(\frac n\varepsilon\right)\polylog\!\left(\frac1\delta\right)
      &&\text{classical runtime},
    \end{aligned}
  \end{equation}
  and it outputs the following:
  \begin{itemize}
    \item If $\operatorname{dist}_1(\rho,\mathcal M_s)\le \zeta=\Theta(\varepsilon^4/n)$, it outputs $\mathsf{YES}$ and an $s$-sparse operator $\widehat\rho$ such that $\lVert\widehat\rho-\rho\rVert_1\le\varepsilon$ with probability at least $1-\delta$.
    \item If $\operatorname{dist}_1(\rho,\mathcal M_s)>\varepsilon$, it outputs $\mathsf{NO}$ with probability at least $1-\delta$.
  \end{itemize}
  The degrees of polynomials in the sample and classical computational complexities and the constant factor of $\zeta$ depend on $s$.
\end{theorem}

\begin{proof}
  Suppose first that $\operatorname{dist}_1(\rho,\mathcal M_s)\le\zeta$, where $\zeta=\Theta(\varepsilon^4/n)$. We use the parameters $\varepsilon_1$, $\varepsilon_2$, $\xi$, and $\theta$ from Eq.~\eqref{eq:pauli-moment-learning-final-tolerances}. For fixed $s$, we have $\varepsilon_1=\Theta(\varepsilon^2/\sqrt n)$. Then we can choose such $\zeta$ satisfying 
  \begin{equation}
    \lVert\rho-\sigma\rVert_1
    \le\frac{\varepsilon_1}{2},
  \end{equation}
  for some $\sigma\in\mathcal S_s$. Specifically, by Lemma~\ref{lem:sparse-matrix-physical-repair}, there exists a $\sigma\in\mathcal S_s$ such that
  \begin{equation}
    \lVert\rho-\sigma\rVert_1
    \le2\sqrt{\zeta}+7s\zeta.
  \end{equation}
  Therefore, we can choose $\zeta=c_s \varepsilon^4/n$ with a sufficiently small constant $c_s$, which depends on $s$, so that such a $\sigma\in\mathcal S_s$ exists.

  Recall that we refined the estimation of Pauli moments to $\varepsilon_1$, i.e.,
  \begin{equation}
    |[\widehat T_W]_\alpha-[T_W]_\alpha|
    \le\frac{\varepsilon_1}{2},
  \end{equation}
  for all $\alpha\in\{I, X, Y, Z\}^W$ and all $W\subseteq[n]$ with $|W|\le\ell_0$. Now, let us denote the Pauli moment tensor of $\sigma$ as $T_W(\sigma)$. Then, we have
  \begin{equation}
    |[T_W]_\alpha-[T_W(\sigma)]_\alpha|
    =\lvert\operatorname{Tr}[(\rho-\sigma) \alpha]\rvert
    \le\lVert\rho-\sigma\rVert_1
    \le\frac{\varepsilon_1}{2},
  \end{equation}
  as every Pauli observable $\alpha$ has operator norm one. Then, we have
  \begin{equation}
    |[\widehat T_W]_\alpha-[T_W(\sigma)]_\alpha|
    \le |[\widehat T_W]_\beta-[T_W]_\beta|
      +|[T_W]_\beta-[T_W(\sigma)]_\beta|
    \le\varepsilon_1.
    \label{eq:algorithm1-tolerant-moment-transfer}
  \end{equation}
  Therefore, the moment estimates are within $\varepsilon_1$ of the Pauli moments of the sparse state $\sigma$.
  
  Consequently, Steps~1(b) and~1(c) output a list of candidate bases that contains a basis that makes $\sigma$ approximately $s$-sparse. Specifically, let $B$ be a product basis in which $\sigma$ is $s$-sparse. Lemma~\ref{lem:contraction-stability} and Lemma~\ref{lem:accurate-candidate-basis} require only that the input state is sparse and that the errors in the estimations of the Pauli moments are within $\varepsilon_1$. Eq.~\eqref{eq:algorithm1-tolerant-moment-transfer} establishes precisely this accuracy relative to $\sigma$, even though the estimates were obtained by measuring $\rho$. Applying those arguments with $\sigma$ as the input state therefore shows that the candidate basis list contains a basis $B'$ whose local bases are within angular distance $\gamma$ of those of $B$, with $\gamma$ as in Eq.~\eqref{eq:pauli-moment-learning-angular-target}.

  We now show that $B'$ passes the tests in Step~2 and Step~3, and thus the algorithm outputs $\mathsf{YES}$. To this end, apply Lemma~\ref{lem:product-basis-stability} with $\sigma$ as the input state and $B'$ as the candidate basis. Then, there exists a state $\sigma'$ that is $s$-sparse in $B'$ such that
  \begin{equation}
    \lVert\sigma-\sigma'\rVert_1
    \le\sqrt{sn}\,\gamma.
  \end{equation}
  The triangle inequality now bounds the distance from the actual input to the sparse comparison state in $B'$:
  \begin{equation}
    \lVert\rho-\sigma'\rVert_1
    \le\lVert\rho-\sigma\rVert_1+\lVert\sigma-\sigma'\rVert_1
    \le\sqrt{sn}\,\gamma+\frac{\varepsilon_1}{2} \le \frac{\varepsilon^2}{64s^{3/2}}.
  \end{equation}

  For Step~2, let $q$ and $\widehat q$ be the true and empirical measurement distributions of $\rho$ in $B'$. Recall that Step~2 estimates $\widehat q$ to accuracy
  \begin{equation}
    \lvert \widehat q(T) - q(T) \rvert \le \xi/2,
  \end{equation}
  for all $T\subseteq\{0,1\}^n$ with $|T|\le s$, where $\xi$ is as in Eq.~\eqref{eq:pauli-moment-learning-final-tolerances}. Let $R:=\{x:[\sigma']_{B'}(x,x)>0\}$ be the diagonal support of $\sigma'$. Since $\sigma'$ has at most $s$ entries, $|R|\le s$, and its measurement distribution assigns no probability to $R^c$. Therefore, we have
  \begin{equation}
    q(R^c)\le\lVert\rho-\sigma'\rVert_1\le \frac{\varepsilon^2}{64s^{3/2}} <\frac\xi2.
  \end{equation}
  Step~2 selects a set $S$ of at most $s$ bitstrings maximizing $\widehat q(S)$, so $\widehat q(S)\ge\widehat q(R)$ and
  \begin{equation}
    \widehat q(S^c)
    \le\widehat q(R^c)
    \le q(R^c)+\frac\xi2
    \le\lVert\rho-\sigma'\rVert_1+\frac\xi2<\xi.
  \end{equation}
  Therefore, $B'$ passes the test of Step~2 with $S$.

  For Step~3, recall that it estimates the $S\times S$ submatrix of $\rho$ in $B'$, where accuracy for each element in $S\times S$ is within $\theta/2$. Let $\widetilde M$ be the symmetrized estimate of the $S\times S$ submatrix of $\rho$ in $B'$. Consider any $x,y\in S$ for which $[\sigma']_{B'}(x,y)=0$. Then, we have
  \begin{equation}
      |[\widetilde M]_{xy}|
      \le |[\widetilde M]_{xy}-[\rho]_{B'}(x,y)|
        +|[\rho-\sigma']_{B'}(x,y)|
      \le\frac\theta2+\lVert\rho-\sigma'\rVert_1
      <\theta,
  \end{equation}
  where we used $\lVert\rho-\sigma'\rVert_1 \le \varepsilon^2/64s^{3/2} < \theta/2$.
  Every such entry is therefore discarded by thresholding. The retained entries are contained in the entry support of $\sigma'$, which has size at most $s$, so $B'$ also passes the test of Step~3. This is conditional on the success of the estimations in Step~1 to Step~3, which can fail with probability at most $\delta$. Therefore, the algorithm outputs $\mathsf{YES}$ with probability at least $1-\delta$.
  
  When the algorithm outputs $\mathsf{YES}$, it returns an $s$-sparse operator $\widehat\rho$ that is $\varepsilon$-close to the input $\rho$ in trace distance. To see this, note that Lemma~\ref{lem:reconstruction-soundness} applies to the actual input $\rho$: as long as the basis and the output state pass the tests in Step~2 and Step~3, we have $\lVert \widehat \rho - \rho \rVert_1 \le \varepsilon$, with no sparsity assumption on $\rho$ in the proof.
  
  This also proves the $\mathsf{NO}$ case. Specifically, suppose $\operatorname{dist}_1(\rho,\mathcal M_s)>\varepsilon$. Assume that all estimates in Steps~1 to 3 are successful with probability at least $1-\delta$. If the algorithm outputs $\mathsf{YES}$, then it returns an $s$-sparse operator $\widehat\rho$ such that $\lVert\rho-\widehat\rho\rVert_1\le\varepsilon$. This contradicts the promise that $\operatorname{dist}_1(\rho,\mathcal M_s)>\varepsilon$. Therefore, the algorithm must output $\mathsf{NO}$ with probability at least $1-\delta$.
\end{proof}